\documentclass{article}
\PassOptionsToPackage{numbers, compress}{natbib}

\usepackage[main, final]{neurips_2026}
\usepackage[utf8]{inputenc}
\usepackage[T1]{fontenc}

\usepackage{microtype}
\usepackage{graphicx}
\usepackage{subcaption}
\usepackage{booktabs}
\usepackage{tikz}
\usepackage{wrapfig}

\usepackage{amsmath,amsfonts,amssymb,mathtools,amsthm}
\usepackage{bm}

\usepackage{url}
\usepackage{hyperref}
\usepackage{xcolor}
\usepackage{nicefrac}
\usepackage{enumitem}
\usepackage{comment}
\usepackage{apxproof}
\usepackage[createShortEnv]{proof-at-the-end}
\makeatletter
\@ifundefined{newcounteralias}{}{%
  \def\thmt@autorefsetup{%
    \@xa\def\csname\thmt@envname autorefname\@xa\endcsname\@xa{\thmt@thmname}}%
}
\makeatother
\theoremstyle{plain}
\newtheorem{theorem}{Theorem}[section]
\newtheorem{proposition}[theorem]{Proposition}
\newtheorem{lemma}[theorem]{Lemma}

\theoremstyle{definition}
\newtheorem{definition}[theorem]{Definition}
\theoremstyle{remark}

\newif\ifcrblue \crbluefalse
\ifdefined\CRBLACK \crbluefalse \fi
\ifcrblue \colorlet{crcolor}{blue}\else \colorlet{crcolor}{black}\fi
\makeatletter
\DeclareRobustCommand{\cameraready}[1]{\ifmmode\expandafter\cr@math\else\expandafter\cr@text\fi{#1}}
\newcommand{\cr@text}[1]{{\leavevmode\color{crcolor}#1}}
\newcommand{\cr@math}[1]{\begingroup\color{crcolor}#1\endgroup}
\makeatother
\newcommand{\crcolor}{\color{crcolor}}
\newenvironment{crblock}{\leavevmode\color{crcolor}}{}

\newcommand{\equalcontrib}{\textsuperscript{*}}

\newcommand{\TV}{\operatorname{TV}}
\newcommand{\PP}{\mathbb{P}}

\newcommand{\argmax}{\operatorname*{argmax}}

\newcommand{\yc}{y_c^*}

\title{A Statistical Framework for Algorithmic Collective Action with Multiple Collectives}

\author{%
  Claudio~Battiloro\equalcontrib$^{,1,6}$ 
  Pietro~Greiner\equalcontrib$^{,2,5}$ 
  Dario~Rancati\equalcontrib$^{,3}$ 
  Bret~Nestor$^{4}$ 
  Oumaima~Amezgar$^{5}$ \\
  \textbf{Francesca~Dominici}$^{6}$ \\[0.5em]
  $^{1}$Brown University \quad
  $^{2}$LawZero \quad
  $^{3}$Institute of Science and Technology Austria \\
  $^{4}$University of British Columbia \quad
  $^{5}$Universit\'e de Montr\'eal and Mila \quad $^{6}$Harvard University\\ [0.3em]
 \equalcontrib Equal contribution. Correspondence to: \texttt{claudio\_battiloro@brown.edu}
}

\begin{document}
\maketitle

\begin{abstract}
As learning systems increasingly shape everyday decisions, Algorithmic Collective Action (ACA), i.e., users coordinating changes to shared data to steer model behavior, offers a complement to regulator-side policy and corporate model design. Real-world collective actions have traditionally been decentralized and fragmented into multiple collectives, despite sharing overarching objectives, with each collective differing in size, strategy, and actionable goals. However, most of the ACA literature focuses on single collective settings. To address this, we propose the first comprehensive statistical framework for ACA with multiple collectives acting on the same system. In particular, we focus on collective action in classification, studying how multiple collectives can influence a classifier's behavior. We provide quantitative statistical bounds on the success of the collectives, considering the role and the interplay of the collectives' sizes and the alignment of their goals. We make such bounds computable by each collective with only partial knowledge of other collectives' sizes and strategies. Finally, we numerically illustrate our framework on simulations inspired by interventions for climate adaptation in smart cities, demonstrating the usefulness of our bounds\cameraready{, and we verify that the bounds hold for classifiers trained on collectively edited synthetic and US Census data}. 
\end{abstract}

\vspace{-0.3cm}\section{Introduction}\vspace{-0.2cm}\label{sec:intro}
AI’s rapid expansion is powered by massive training datasets that produce more capable predictive models and enable broader applications. However, as these data-driven systems increasingly shape everyday life, they pose tangible risks, including privacy breaches \cite{Okunyte2026AndroidAIAppLeak, nakka2025pii}, labor exploitation \cite{perrigo2023openaiKenyaWorkers}, and biased decisions that exacerbate socio-economic disparities \cite{ai2024artificial}.

Responses to these risks span firms, regulators, and users. At the firm level, “Trustworthy AI” programs integrate fairness checks, bias mitigation, privacy audits, and red-teaming across the system lifecycle \cite{barocas2023fairness}, but they often conflict with objectives tied to model performance and user engagement. At the regulatory level, data-protection regimes such as GDPR  \cite{EU2016GDPR}, PIPEDA \cite{Canada2000PIPEDA}, and CPRA \cite{California2020CPRA}  establish baseline obligations for different stakeholders, although formal adherence alone is insufficient to ensure substantively responsible or equitable outcomes \cite{selbst2019fairness,utz2019informed}. Finally, at the user level, \textit{Algorithmic Collective Action in machine learning}  (ACA) has emerged as a promising framework to coordinate grassroots efforts, by enabling individual users to collectively contribute, withhold, or strategically structure their data to steer model behavior \cite{HardtEtAl2023,devrio2024building}.

This work is motivated by the fact that the ACA literature has almost entirely focused on scenarios involving one single collective. However, traditionally, real-world actions have been carried out by \textit{multiple collectives}, whose uneven capacities, temporal misalignment, and varying priorities may result in fragmented, decentralized, and partially overlapping interventions. Different collectives often \textit{vary} widely in size, strategy, and even actionable goals, yet they can still \textit{align} on common overarching objectives (e.g., climate justice or gender equality). Latour’s actor-network theory \cite{Latour2005Reassembling} conceptualized why such movements rarely behave as a single unit: what is regarded as a unitary movement is really a web of heterogeneous interactions with no fixed or permanent center. These loose networks manage to act together only through alignment and coordination. Collectives therefore, act coherently only insofar as such coordination and alignment efforts are continuously sustained. An interesting example is the cyberfeminist struggle, which has never been one monolithic front; instead, it consists of distributed efforts ranging from feminist hacker-art collectives to global hashtag campaigns. Haraway’s Cyborg Manifesto \cite{Haraway1991CyborgManifesto} famously proposed the cyborg as a metaphor for coalition across difference, emphasizing affinity over rigid identity categories. In that spirit, cyberfeminism has taken many decentralized forms around the world while still sharing core aims of challenging patriarchal tech culture and gender inequality.  Environmental activism provides another strong example: the climate justice movement is a patchwork of groups employing diverse tactics (from youth-led social media strikes to militant direct-action cells) to achieve various actionable goals (from permit denials to pipeline stoppages and enforcement actions), but they all strive toward the same overarching objective of halting climate breakdown. Here, we aim to reflect and model the need for multiple collectives in ACA.

\textbf{Related Works.} ACA as users steering ML outputs toward a group goal was recently rigorously formalized in \cite{HardtEtAl2023}, building on the Data Leverage framework’s view of data as an active instrument rather than a passive input \cite{VincentEtAl2020}. In \cite{SiggEtAl2024}, the authors proposed a combinatorial model for ACA to study the strategic interaction between drivers and delivery platforms, inspired by the \#DeclineNow DoorDash campaign. The work in \cite{BenDovEtAl2024} studied how the firm's chosen learning algorithms affect the success of ACA.  Motivated by the growing regulatory focus on privacy and data protection, \cite{solanki2025crowding} found that differentially private training can impede the success of algorithmic collective action. The works in \cite{baumann2024algorithmic} and \cite{fedorova2025altruistic} focused on ACA in recommender systems. Simulation-based analyses of platform--worker dynamics were studied in \cite{lewandowska2025workers}. Complementary directions connect ACA to explainability/recourse \cite{hussain2025empowering} and to fairness interventions by minority groups \cite{bendov2025minority}. The authors of \cite{gauthier2025statistical} proposed an alternative theoretical framework to \cite{HardtEtAl2023} that empowers collectives via statistical inference, enabling them to learn better ACA strategies and infer the parameters that determine their success. The closest works to ours are \cite{karan2025algorithmic}, \cite{karan2025sync_or_sink}, and \cite{battiloro2025algorithmic}. The authors of \cite{karan2025algorithmic} introduced a conceptual framework for ACA with multiple collectives, and experiments are performed. However, no theoretical treatment is provided, and most of the focus is on ACA with two collectives. The workshop papers in \cite{karan2025sync_or_sink} and \cite{battiloro2025algorithmic} generalized the original probabilistic framework from \cite{HardtEtAl2023} to the (noisy) two collective and the multiple collective settings, respectively. Our work is inspired by the empirical findings of \cite{karan2025algorithmic} and builds on the framework introduced in \cite{battiloro2025algorithmic}. We move substantially beyond prior work by providing, to the best of our knowledge, the \textit{first statistical treatment for Algorithmic Collective Action with multiple collectives}, generalizing the framework from \cite{gauthier2025statistical}.

\textbf{Contribution.} We introduce the first statistical framework for ACA with multiple collectives. We focus on collective action in classification. In particular, we study how $M$ collectives can plant and unplant signals, i.e., bias a classifier to learn or unlearn an association between an altered version of the features $x$ and a chosen, possibly overlapping, multiset of target classes $\{\{y_c\}\}_c$. We measure collective action success using both per-collective and global metrics, where global success can be defined as a summary of the per-collective success measures, such as their average or the minimum. We generalize the statistical framework from \cite{gauthier2025statistical}. As such, for signal planting, we analyze two regimes: when collectives can act on both features and labels, and when collectives are limited to acting on features only. For each regime, we derive lower bounds on the per-collective success. For signal unplanting, we analyze the case in which collectives employ an adaptive strategy acting on both features and labels, and we again derive a lower bound on the per-collective success. Both for signal planting and unplanting, our results reveal interesting trade-offs driven by the interplay of the sizes of the collectives and by how closely their actionable goals align. Although each collective’s bounds depend on other collectives’ strategies, we make them individually computable by introducing a hierarchy of approximations, ranging from worst case (the collective knows nothing about others) to best case (the collective knows some aggregated quantity about other collectives’ strategies). Finally, we first conceptually illustrate our framework with a use case in the space of interventions for climate adaptation, and we design a synthetic dataset that simulates it, numerically showing the usefulness of and the interesting tradeoffs emerging from our framework. \cameraready{Then, we also verify that the bounds hold for classifiers actually trained on collectively edited synthetic and US Census data. Code is available at \url{https://github.com/clabat9/multi-collective-aca}}.

\vspace{-0.3cm}\section{Preliminaries and Setting}\vspace{-0.2cm}\label{sec:MACA}
We start from the definition of ACA with a single collective as given in \cite{HardtEtAl2023,gauthier2025statistical}. A firm deploys a learning model on data drawn from a certain user population. Within the overall user population, a collective forms that represents a given share of users (a positive fraction, but not the whole). This collective wants to steer the platform’s outcome toward a certain objective. To do so, the collective agrees on one strategy for editing their contributions—potentially changing features, labels, or both—before those contributions reach the platform. The platform then receives a training set that blends unmodified data from the baseline process with the collective’s edited data, and it deploys its algorithm on this blended data. The collective’s impact is evaluated by a tailored success measure. Given its size, the collective’s problem is to choose the editing rule that maximizes this success once the platform has trained on the blended data. Here, we rigorously generalize the statistical treatment of ACA from \cite{gauthier2025statistical} to a setting with multiple collectives, each acting with a possibly different strategy and actionable goal, but attempting to maximize both per-collective and global measures of success. \cameraready{The measure-theoretic details of the setting are collected in Appendix~\ref{app:formal}.}

\textbf{Data space.} We define the data space $\mathcal{Z}=\mathcal{X}\times\mathcal{Y}$ as the product of two finite spaces $\mathcal{X}$ and $\mathcal{Y}$ being the feature and label spaces, respectively.

\textbf{Population.} \cameraready{The population is a finite set $\Omega$ of users with a sampling distribution $\pi$, and $f_0:\Omega\to\mathcal{Z}$ assigns a feature--label pair to each user. The population distribution is the push-forward $P_0:=(f_0)_{\#}\pi$, the \emph{pre-edit} distribution over feature--label pairs: for each $z=(x,y)\in\mathcal{Z}$, $P_0(z)$ is the $\pi$-fraction of users whose assigned pair is $z$.}

\textbf{Collectives, Strategies, and Masses.}
We define an ensemble of $M$ collectives as disjoint  subsets $\Omega_1,\dots,\Omega_M$ of $\Omega$ such that $\pi(\Omega_c)=\alpha_c>0$. We refer to $\alpha_c$ as the mass of the $c$-th collective and set $\bar\alpha:=\sum_{c=1}^M \alpha_c$. Each collective $\Omega_c$ agrees on a (deterministic) strategy $h_c:\mathcal{Z}\rightarrow\mathcal{Z}$ from a space of feasible edits. 
For every collective $\Omega_c$, we make the \textit{neutrality assumption} \cite{HardtEtAl2023,gauthier2025statistical} that the pre-edit distribution within the collective matches the population pre-edit distribution, i.e., $P_0(z)=\PP\left[f_0(\omega)=z \,\middle|\, \omega\in\Omega_c\right]\; \forall z\in\mathcal Z$. Intuitively, before applying $h_c$, each collective is a random slice of the same underlying population (membership is independent of unedited examples).

\textit{Remark.} The neutrality assumption is reasonable in several cases, e.g., when collectives form based on attributes unrelated to data and unknown to the platform. In other cases, collectives may form based on the data they aim to edit; then one can relax the assumption to $\operatorname{TV}\!\bigl(P_0,\ \PP\left[f_0(\omega) \in \cdot \mid \omega \in \Omega_c\right]\bigr)\le \phi_{N,c}$. \cameraready{In what follows we keep neutrality; Appendix~\ref{app:tv} shows that our bounds continue to hold under the relaxed assumption, with explicit bias terms that are linear in $\phi_{N,c}$.}

\textbf{Mixture Distribution.} The action of collective $\Omega_c$ via $h_c$ maps $P_0$ to $P_c:=(h_c)_{\#}P_0$. Thus, in the infinite-data regime, a firm that trains a learning model on data coming from $\Omega$ observes the \textit{post-edit} mixture distribution
\begin{equation}\label{eq:mixture}
    P
    \;=\;
    \bigl(1-\bar\alpha\bigr)\,P_{0}
    \;+\;
    \sum_{c=1}^{M}\alpha_{c}\,P_{c}.
\end{equation}
\textbf{Training Set and Empirical Distribution.} In practice, if a collective action is not performed, a firm would use a training set of $N$ data points
$\{Z_i=(X_i,Y_i)\}_{i=1}^N$. We model the training set by sampling individuals $S_1,\dots,S_N$ i.i.d.\ from $\pi$, and setting $Z_i := f_0(S_i)$. Then $\{Z_i\}_{i=1}^N$ are i.i.d.\ from the pre-edit distribution $P_0$\cameraready{; this covers, e.g., survey-like sampling and platform logging (Appendix~\ref{app:formal})}. If a collective action is performed, i.e., if the collectives edit these examples according to their strategies $h_c$, the firm would train its learning model using the training set of the $N$ edited data points. The  individuals belonging to collectives that appear in the training set, which we refer to as \textit{sample collectives}, are then captured by the random \emph{index sets}
\begin{equation*}
I_{N,c}\;:=\;\bigl\{\,i\in\{1,\dots,N\}:\ S_i\in \Omega_c\,\bigr\},\quad c=1,\dots,M,
\end{equation*}
with sizes $n_{N,c}:=|I_{N,c}|$ and empirical masses $\alpha_{N,c}:=n_{N,c}/N$, so that $\bar\alpha_N:=\sum_{c=1}^M \alpha_{N,c}$. We define the empirical non-collective
mass $\alpha_{N,0} := 1-\bar\alpha_N$, and 
denote the number of non-collective training points with  $n_{N,0} \;:=\; N - \sum_{d=1}^M n_{N,d}$. \cameraready{We write $\hat P_{N,c}^{\mathrm{pre}}$ and $\hat P_{N,c}:=(h_c)_{\#}\hat P_{N,c}^{\mathrm{pre}}$ for the empirical pre-edit and post-edit distributions of the sample collective $I_{N,c}$, with feature marginals $\hat P_{N,c}^{\mathrm{pre},X}$ and $\hat P_{N,c}^X$, and $\hat P_{N,0}$ for the empirical distribution of the non-collective training points (formal definitions in Appendix~\ref{app:formal}).} Thus the post-edit empirical
mixture distribution is
\begin{align}\label{eq:emp_mixture}
\hat P_{N}
\;=\;
\alpha_{N,0}\hat P_{N,0}
+
\sum_{c:\alpha_{N,c}>0}\alpha_{N,c}\,\hat P_{N,c}.
\end{align}
\cameraready{As $N\to\infty$, $\hat P_N$ converges almost surely to the mixture \eqref{eq:mixture} (Proposition~\ref{prop:finite_to_infinite} in Appendix~\ref{app:formal}).}

We study the case where the firm's learning model is a classifier.  In particular, we assume such a classifier is an $\varepsilon$-contamination-suboptimal classifier.
\begin{definition}[$\varepsilon$-contamination-suboptimal classifier]
A classifier $\hat m:\mathcal X\to\mathcal Y$ is
$\varepsilon$-contamination-suboptimal on
$\mathcal X'\subseteq \operatorname{supp}(\hat P_N^X)$ under $\hat P_N$
if there exists a probability measure $\hat Q_N$ on
$\mathcal X\times\mathcal Y$ and a probability measure $R_N$ on
$\mathcal X\times\mathcal Y$ such that
\begin{equation}
\hat Q_N=(1-\varepsilon)\hat P_N+\varepsilon R_N,
\end{equation}
and, for every $x\in\mathcal X'$ with $\hat Q_N^X(x)>0$,
\begin{equation}
\hat m(x)\in \argmax_{y\in\mathcal Y}\hat Q_N(y\mid x).
\end{equation}
\end{definition}
Consistently with previous literature \cite{HardtEtAl2023, gauthier2025statistical}, we take $\varepsilon< 1/2$.
\cameraready{Appendix~\ref{app:formal} relates this contamination model to total-variation suboptimality, under which all results remain valid with $\varepsilon/(1-\varepsilon)$ replaced by $2\varepsilon$.}

\noindent\textbf{Per-collective and Global Success.} Given a strategy profile $\{h_c\}_{c=1}^{M}$, a classifier $\hat m:\cal X\to Y$ trained on the training set $\{Z_i\}_{i=1}^N$, and a test set of $N^{\textrm{test}}$ data points $\{Z^{\textrm{test}}_i\}_{i=1}^{N^{\textrm{test}}}$ again drawn i.i.d. from $P_0$, each sample collective $I_{N,c}$ is interested in maximizing a per-collective test success metric $\hat S_c$, and a global test success metric $\cameraready{\hat G(\{\hat S_c\}_c)}$ being an aggregation of some sort of the per-collective successes\cameraready{; we write $\hat G$ for global metrics to stress that they are functions of the per-collective successes $\{\hat S_c\}_c$, and reserve $\hat S_c$ for the per-collective success}. If such aggregation is a $\min$, we obtain an egalitarian metric which takes the per-collective success of the least successful collective as the global measure of success~\cite{Sagawa2020Distributionally}, i.e.,
\begin{equation}\label{eq:global_min}
         \cameraready{\hat G_{\min}(\{\hat S_c\}_c)}
        := \min_{c:\alpha_{N,c}>0} \; \hat S_c(\alpha_{N,c}, N^{\textrm{test}}).
    \end{equation}  
    Alternatively, if we average, we obtain a metric proportional to each collective's mass, i.e.,
\begin{equation}\label{eq:global_avg}
       \cameraready{\hat G_{\mathrm{w}}(\{\hat S_c\}_c)}
        := \frac{1}{\bar\alpha_N}\sum_{c=1}^{M}\alpha_{N,c}\,\hat S_c(\alpha_{N,c}, N^{\textrm{test}}).
    \end{equation}
Similarly to Proposition \ref{prop:finite_to_infinite}, in the infinite-data regime, the versions of these success metrics computed under the empirical test measure $\hat P_{N^{\textrm{test}}}$ converge almost surely to their population counterparts computed under $P_0$.

\textit{Remark.} In the rest of the paper, we derive bounds for the per-collective success. The bounds on global success can be readily derived by aggregating the per-collective bounds with the chosen summary and adapting union bounds.

\textbf{Further Setting Details.} Following \cite{gauthier2025statistical}, we assume each collective has access only to their own data and not to the data of other collectives or of those who are not part of any collective. However, a collective could have some sort of aggregated knowledge about other collectives (see Sections \ref{sec:signal_planting}-\ref{sec:signal_unplanting}). Moreover, we will use Hoeffding's concentration inequality (see Lemma D.1 in \cite{gauthier2025statistical}) without loss of generality. We denote Hoeffding error terms with $R_\gamma(k):=\sqrt{\frac{\log (1 / \gamma)}{2 k}}$,
for any $\gamma>0$ and $k \in \mathbb{N}^+$. In the bounds we derive, whenever
$R_\gamma(n_{N,0})$ appears, we condition on $n_{N,0}>0$. Moreover, in the theoretical results below, we assume that the deployed classifier is $\varepsilon$-contamination-suboptimal under the post-edit empirical distribution
$\hat P_N$ on the relevant planted feature set.

%\textcolor{black}{Important note: Our results never assume that collectives need to share the same strategy! We can say this explicitly.}

\vspace{-0.3cm}\section{Signal Planting}\vspace{-0.2cm}\label{sec:signal_planting}
 We start by analyzing the setting in which each sample collective $I_{N,c}$ wants the classifier $\hat m$ to \textit{learn} an association between an altered version of the features $g_c(x)$ and a chosen target class $y_c^*$. Formally, let $\hat P_{N^{\textrm{test}}}^X$ denote the empirical distribution of test features
  (the marginal of $\hat P_{N^{\textrm{test}}}$ on $\mathcal X$) and $g_c: \cal X \rightarrow \cal X$  the feature map used on target-label points induced by $h_c$, then each sample collective $I_{N,c}$ measures per-collective success as
\begin{equation}
\hat S_c(\alpha_{N,c})=\PP_{x \sim \hat P^X_{N^{\textrm{test}}}}\left[\hat m(g_c(x))=y_c^*\right].
\end{equation}
We denote with $\mathcal{X}_c^*:=g_c(\mathcal{X})\subseteq \cal X$ the subset of the feature space representing the image of the eligible set under $h_c$, i.e., the image of the features alterable by the adoption of $h_c$. As usually done \cite{HardtEtAl2023,gauthier2025statistical,solanki2025crowding, battiloro2025algorithmic}, we consider two cases: the case in which the users can modify both features and labels, referring to the resulting strategies as \textit{feature-label} strategies, and the case in which the users can access both features and labels, but modify only the features, referring to the resulting strategies as \textit{feature-only} strategies.

\noindent \cameraready{In the rest of this and the following section, when there is no risk of ambiguity we fix a collective $\Omega_c$, together with its sample counterpart $I_{N,c}$, and assume its empirical mass $\alpha_{N,c} > 0$.} Moreover, our results are conditioned on the realization of the sizes
of the sample collectives $\{n_{N,0},n_{N,1},\dots,n_{N,M}\}$.  

For each $x^*\in \mathcal{X}_c^*$, we define the empirical prevalence of the planted feature
in collective $I_{N,c}$ and the empirical (planting) \textit{conflict} terms with other collectives as
\begin{align}
\hat p_c^X(x^*)
&:= \hat P_{N,c}^X(x^*), \textrm{ and }
\\[0.2em]
\Delta_{c\to d}(x^*)
&:= \max_{y'\in\mathcal Y\setminus\{\yc\}}
\bigl(
\hat P_{N,d}(x^*,y') - \hat P_{N,d}(x^*,\yc)
\bigr),\quad  d\in\{1,\dots,M\},\ d\neq c,
\end{align}
respectively. Intuitively, $\Delta_{c\to d}(x^*)$ measures how strongly sample collective $I_{N,d}$'s post-edit data associates the planted feature value $x^*$ with some label other than $y_c^*$ compared to $y_c^*$. Therefore, larger absolute values indicate greater potential for $I_{N,d}$ to \textit{confuse or accrete} (based on its sign)  $I_{N,c}$'s planted signal at $x^*$. \cameraready{To} approximate the non-collective conflict using only the pre-edit data of 
the sample collective $I_{N,c}$,  we define
\begin{equation}
\Delta_{c,0}^{(n_{N,c})}(x^*)
\;:=\;
\max_{y'\in\mathcal Y\setminus\{\yc\}}
\bigl(
\hat P_{N,c}^{\mathrm{pre}}(x^*,y') - \hat P_{N,c}^{\mathrm{pre}}(x^*,\yc)
\bigr).
\end{equation}
Intuitively, $\Delta_{c,0}^{(n_{N,c})}(x^*)$ is a data-driven proxy for how much the non-collective population could counteract the planted signal at $x^*$ by supporting an alternative label over $y_c^*$, estimated under neutrality.

\textbf{Feature-label Signal Planting Strategy.} For each sample collective $I_{N,c}$, the \textit{feature-label} signal strategy is given by
\begin{equation}\label{eq:FL_strat}
h_c(x, y)=\left(g_c(x), y_c^*\right).
\end{equation}
 This strategy forces every edited example to carry the same label $y_c^*$ while mapping its features into the planted set via $g_c$. It directly aims to create an association between the transformed features and the target class in the training data.

 \textbf{Feature-only Signal Planting Strategy.}
For each sample collective $I_{N,c}$, assume that $
\mathcal X\setminus\mathcal X_c^*\neq\emptyset$, and fix a point $x_{0,c}\in\mathcal X\setminus\mathcal X_c^*$.
The \textit{feature-only} signal planting strategy is given by
\begin{equation}\label{eq:FO_strat}
h_c(x,y)
=
\begin{cases}
\left(g_c(x),y\right), & y=y_c^*,\\
\left(x_{0,c},y\right), & y\neq y_c^*.
\end{cases}
\end{equation}
The intuition here is that, while the features of the target-label examples
are modified according to $g_c$, examples with labels $y\neq y_c^*$ are pushed
to a feature value outside the planted signal set. In this setting, the label $y$ is assumed to be accessible or computable, but not modifiable by
the collective.

We can now state our main result for this section.

\begin{theoremE}[Signal planting with multiple collectives][end, restate]\label{thm:multi_feature_label}
Under the neutrality assumption, assuming that $\hat m$ is
$\varepsilon$-contamination-suboptimal on $\mathcal X_c^*$ under
$\hat P_N$, and assuming all collectives either employ the feature-label or the feature-only strategy from \eqref{eq:FL_strat}-\eqref{eq:FO_strat}, respectively, with probability at least $1-\delta$ over the draw of the training and test samples, with $\delta>0$, the per-collective success of collective $I_{N,c}$ satisfies
\begin{align}\label{eq:multi_FL_bound_final_clean_preonly}
\hat S_c(\alpha_{N,c})
\;&\ge\; \;
\PP_{x\sim\hat P_{N,c}^{\mathrm{pre},X}}
\Big[
\;
\alpha_{N,c}\Bigl(
\hat p_c^X\bigl(g_c(x)\bigr)
-
2R_{\tilde\delta}(n_{N,c})
\Bigr)
-
\sum_{d\neq c}\alpha_{N,d}\,\Delta_{c\to d}\bigl(g_c(x)\bigr) \nonumber
\\
&-
\alpha_{N,0}\,\Delta_{c,0}^{(n_{N,c})}\bigl(g_c(x)\bigr)
-\;
2\alpha_{N,0}\Bigl(
R_{\tilde\delta}(n_{N,0}) + R_{\tilde\delta}(n_{N,c})
\Bigr)
-
\frac{\varepsilon}{1-\varepsilon}
>0
\Big]\nonumber\\
&-
R_{\delta/4}(n_{N,c})
-
R_{\delta/4}(N^{\mathrm{test}}),
\end{align}
where $K_c \;:=\; |\mathcal{X}_c^*|\,|\mathcal Y|$ and $\tilde\delta \;:=\; \frac{\delta}{10K_c}$.
\end{theoremE}

\begin{proofE}
The proof follows six steps, here is the outline:
\begin{enumerate}
\item We use $\varepsilon$-contamination-suboptimality of $\hat m$ to obtain a pointwise margin condition under $\hat P_N$ that guarantees correct prediction at a feature value $x^*$ whenever a suitable margin inequality holds.
\item We express this margin in terms of the empirical mixture $\hat P_N$ by decomposing contributions from each collective, and define the resulting exact empirical margin $M_c(x^*)$.
\item We lower bound $M_c(x^*)$ by a fixed population-based margin $\overline M_c(x^*)$ and then compare $\overline M_c(x^*)$ to the corrected empirical margin appearing in the theorem statement.
\item We convert the resulting fixed margin condition into a lower bound on the test success probability under $P_0^X$, using concentration on the test sample.
\item We replace this population probability by the corresponding empirical probability over the pre-edit feature distribution of collective $c$, again via Hoeffding’s inequality, and then pass pointwise to the corrected empirical margin event.
\item We combine all concentration events via a union bound and obtain the desired lower bound on $\hat S_c(\alpha_{N,c})$.
\end{enumerate}

Throughout the proof, we work conditionally on the realized sample sizes
$\{n_{N,0},n_{N,1},\dots,n_{N,M}\}$, and all probabilities $\PP[\cdot]$ are
with respect to the remaining randomness in the training and test draws. We
write $R_\gamma(k)=\sqrt{\log(1/\gamma)/(2k)}$.

\noindent\emph{Step 1.}
We first need the following simple lemma.
\begin{lemma}[Margin implication for planting]
Assume $\hat m$ is $\varepsilon$-contamination-suboptimal on
$\mathcal X'$ under $\hat P_N$. If $x^*\in\mathcal X'$ and
\[
\hat P_N(x^*,y^*)>
\hat P_N(x^*,y')+\frac{\varepsilon}{1-\varepsilon}
\qquad \forall y'\neq y^*,
\]
then $\hat m(x^*)=y^*$.
\end{lemma}
\begin{proof}
By $\varepsilon$-contamination-suboptimality, there exists
$\hat Q_N=(1-\varepsilon)\hat P_N+\varepsilon R_N$ such that
$\hat m(x)\in\argmax_y \hat Q_N(y\mid x)$ on $\mathcal X'$.
Fix $x^*\in\mathcal X'$ and $y'\neq y^*$. Then
\[
\hat Q_N(x^*,y^*)-\hat Q_N(x^*,y')
=
(1-\varepsilon)
\bigl(\hat P_N(x^*,y^*)-\hat P_N(x^*,y')\bigr)
+
\varepsilon
\bigl(R_N(x^*,y^*)-R_N(x^*,y')\bigr).
\]
Since $R_N(x^*,y^*)-R_N(x^*,y')\ge -1$, the assumed margin gives
\[
\hat Q_N(x^*,y^*)-\hat Q_N(x^*,y')
>
(1-\varepsilon)\frac{\varepsilon}{1-\varepsilon}-\varepsilon
=0.
\]
Thus $\hat Q_N(x^*,y^*)>\hat Q_N(x^*,y')$ for all $y'\neq y^*$.
Moreover, the margin condition implies $\hat P_N^X(x^*)>0$, hence
$\hat Q_N^X(x^*)>0$. Therefore $y^*$ is the unique maximizer of
$\hat Q_N(\cdot\mid x^*)$, and $\hat m(x^*)=y^*$.
\end{proof}
Since $\hat m$ is $\varepsilon$-contamination-suboptimal on
$\mathcal X_c^*$ under $\hat P_N$, the margin implication for planting above
gives that, for any $x^*\in\mathcal X_c^*$,
\begin{equation}\label{eq:eps_margin_joint}
\forall y'\neq y_c^*:\;
\hat P_N(x^*,y_c^*)>
\hat P_N(x^*,y')+\frac{\varepsilon}{1-\varepsilon}
\quad\Longrightarrow\quad
\hat m(x^*)=y_c^*.
\end{equation}

\noindent\emph{Step 2.}
Fix the target collective $c$ and an arbitrary $x^*\in \mathcal{X}_c^*$. By the empirical
mixture decomposition,
\[
\hat P_N
=
\alpha_{N,0}\,\hat P_{N,0}
+
\sum_{d=1}^M \alpha_{N,d}\,\hat P_{N,d},
\]
we have, for any $y\in\mathcal Y$,
\begin{equation}\label{eq:PN_decomp}
\hat P_N(x^*,y)
=
\alpha_{N,0}\,\hat P_{N,0}(x^*,y)
+
\sum_{d=1}^M \alpha_{N,d}\,\hat P_{N,d}(x^*,y).
\end{equation}

Under either planting strategy, every edited point from collective $c$ with post-edit feature
value $x^*$ carries label $y_c^*$. Hence, for any $x^*\in \mathcal{X}_c^*$ and any $y'\neq y_c^*$,
\[
\hat P_{N,c}(x^*,y') = 0,
\qquad
\hat P_{N,c}(x^*,y_c^*) = \hat P_{N,c}^X(x^*) = \hat p_c^X(x^*).
\]
Subtracting \eqref{eq:PN_decomp} for $(x^*,y')$ from that for
$(x^*,y_c^*)$ and defining
\[
\Delta_{c,0}(x^*)
\;:=\;
\max_{y'\in\mathcal Y\setminus\{\yc\}}
\bigl(
\hat P_{N,0}(x^*,y') - \hat P_{N,0}(x^*,\yc)
\bigr),
\]
we obtain for any $y'\neq y_c^*$:
\begin{align}
\hat P_N(x^*,y_c^*) - \hat P_N(x^*,y')
&=
\alpha_{N,c}\,\hat p_c^X(x^*)
+
\sum_{d\neq c}\alpha_{N,d}\bigl(\hat P_{N,d}(x^*,y_c^*) - \hat P_{N,d}(x^*,y')\bigr)
\notag\\
&\quad
+\alpha_{N,0}\bigl(\hat P_{N,0}(x^*,y_c^*) - \hat P_{N,0}(x^*,y')\bigr)
\label{eq:margin_exact_decomp}
\\
&\ge
\alpha_{N,c}\,\hat p_c^X(x^*)
- \sum_{d\neq c}\alpha_{N,d}\,\Delta_{c\to d}(x^*)
- \alpha_{N,0}\,\Delta_{c,0}(x^*).
\notag
\end{align}

Define the exact empirical margin
\[
M_c(x^*)
\;:=\;
\alpha_{N,c}\,\hat p_c^X(x^*)
- \sum_{d\neq c}\alpha_{N,d}\,\Delta_{c\to d}(x^*)
- \alpha_{N,0}\,\Delta_{c,0}(x^*),
\qquad x^*\in \mathcal{X}_c^*.
\]
Then \eqref{eq:margin_exact_decomp} yields, for all $x^*\in \mathcal{X}_c^*$ and all
$y'\neq y_c^*$,
\[
\hat P_N(x^*,y_c^*) - \hat P_N(x^*,y') \;\ge\; M_c(x^*).
\]
Combining this with \eqref{eq:eps_margin_joint} gives
\begin{equation}\label{eq:marg_cond_implies_correct_Mc}
x^*\in \mathcal{X}_c^*,
\quad
M_c(x^*) > \frac{\varepsilon}{1-\varepsilon}
\quad\Longrightarrow\quad
\hat m(x^*) = y_c^*.
\end{equation}

\noindent\emph{Step 3.}
Define the population signal prevalence
\[
p_{c,\mathrm{pop}}^X(x^*)
\;:=\;
P_c^X(x^*),
\qquad x^*\in\mathcal X_c^*,
\]
and the population non-collective conflict
\[
\Delta_{c,0}^{\mathrm{pop}}(x^*)
\;:=\;
\max_{y'\neq\yc}
\bigl(
P_0(x^*,y') - P_0(x^*,\yc)
\bigr).
\]
By neutrality, the edited points of collective $c$ are i.i.d.\ from $P_c$, the non-collective points are i.i.d.\ from $P_0$, and the pre-edit points of collective $c$ are i.i.d.\ from $P_0$.

Applying Hoeffding’s inequality and a union bound over the scalar
one-sided concentration events displayed below, we obtain that, with
probability at least $1-\delta/2$, the following inequalities hold
simultaneously. Indeed, the first two families contain
$2|\mathcal X_c^*|$ scalar events, and the last four families contain
$4|\mathcal X_c^*||\mathcal Y|=4K_c$ scalar events. Since
$|\mathcal Y|\ge2$, the total number of scalar events is at most
$5K_c$. With $\tilde\delta=\delta/(10K_c)$, the total failure probability
is therefore at most $\delta/2$. The inequalities are:
\begin{align}
\hat p_c^X(x^*)
&\ge
p_{c,\mathrm{pop}}^X(x^*) - R_{\tilde\delta}(n_{N,c}),
\qquad \forall x^*\in\mathcal X_c^*,
\label{eq:hoeff_signal_collective}
\\[0.3em]
\textcolor{black}{\hat p_c^X(x^*)}
&\textcolor{black}{\le
p_{c,\mathrm{pop}}^X(x^*) + R_{\tilde\delta}(n_{N,c}),
\qquad \forall x^*\in\mathcal X_c^*,}
\label{eq:hoeff_signal_collective_upper}
\\[0.3em]
\hat P_{N,0}(x^*,y)
&\le
P_0(x^*,y) + R_{\tilde\delta}(n_{N,0}),
\qquad \forall (x^*,y)\in\mathcal X_c^*\times\mathcal Y,
\label{eq:hoeff_noncoll_upper}
\\[0.3em]
P_0(x^*,y)
&\le
\hat P_{N,0}(x^*,y) + R_{\tilde\delta}(n_{N,0}),
\qquad \forall (x^*,y)\in\mathcal X_c^*\times\mathcal Y,
\label{eq:hoeff_noncoll_lower}
\\[0.3em]
\hat P_{N,c}^{\mathrm{pre}}(x^*,y)
&\le
P_0(x^*,y) + R_{\tilde\delta}(n_{N,c}),
\qquad \forall (x^*,y)\in\mathcal X_c^*\times\mathcal Y,
\label{eq:hoeff_collpre_upper}
\\[0.3em]
P_0(x^*,y)
&\le
\hat P_{N,c}^{\mathrm{pre}}(x^*,y) + R_{\tilde\delta}(n_{N,c}),
\qquad \forall (x^*,y)\in\mathcal X_c^*\times\mathcal Y.
\label{eq:hoeff_collpre_lower}
\end{align}

Using \eqref{eq:hoeff_noncoll_upper}--\eqref{eq:hoeff_noncoll_lower}, we obtain, for every $x^*\in\mathcal X_c^*$,
\begin{equation}\label{eq:delta_exact_vs_pop}
\Delta_{c,0}(x^*)
\;\le\;
\Delta_{c,0}^{\mathrm{pop}}(x^*)
+
2R_{\tilde\delta}(n_{N,0}).
\end{equation}
Using \eqref{eq:hoeff_collpre_upper}--\eqref{eq:hoeff_collpre_lower}, we obtain, for every $x^*\in\mathcal X_c^*$,
\begin{equation}\label{eq:delta_pop_vs_proxy}
\Delta_{c,0}^{\mathrm{pop}}(x^*)
\;\le\;
\Delta_{c,0}^{(n_{N,c})}(x^*)
+
2R_{\tilde\delta}(n_{N,c}).
\end{equation}

Define the fixed population-based margin
\[
\overline M_c(x^*)
\;:=\;
\alpha_{N,c}\Bigl(
p_{c,\mathrm{pop}}^X(x^*)
-
R_{\tilde\delta}(n_{N,c})
\Bigr)
-
\sum_{d\neq c}\alpha_{N,d}\,\Delta_{c\to d}(x^*)
-
\alpha_{N,0}\Bigl(
\Delta_{c,0}^{\mathrm{pop}}(x^*)
+
2R_{\tilde\delta}(n_{N,0})
\Bigr),
\]
and the corrected empirical margin
\begin{align}
\widetilde M_c(x^*)
\;:=\;
\alpha_{N,c}\Bigl(
\hat p_c^X(x^*)
-
2R_{\tilde\delta}(n_{N,c})
\Bigr)&-
\sum_{d\neq c}\alpha_{N,d}\,\Delta_{c\to d}(x^*)
-
\alpha_{N,0}\,\Delta_{c,0}^{(n_{N,c})}(x^*)
\nonumber \\
&-
2\alpha_{N,0}\Bigl(
R_{\tilde\delta}(n_{N,0})
+
R_{\tilde\delta}(n_{N,c})
\Bigr). \nonumber 
\end{align}

By \eqref{eq:hoeff_signal_collective} and \eqref{eq:delta_exact_vs_pop}, we have, for all $x^*\in\mathcal X_c^*$,
\begin{equation}\label{eq:Mc_vs_barMc}
M_c(x^*)
\;\ge\;
\overline M_c(x^*).
\end{equation}
By \eqref{eq:hoeff_signal_collective_upper} and \eqref{eq:delta_pop_vs_proxy}, we also have, for all $x^*\in\mathcal X_c^*$,
\begin{equation}\label{eq:barMc_vs_tildeMc}
\overline M_c(x^*)
\;\ge\;
\widetilde M_c(x^*).
\end{equation}

Combining \eqref{eq:marg_cond_implies_correct_Mc} and \eqref{eq:Mc_vs_barMc}, we obtain
\begin{equation}\label{eq:marg_cond_implies_correct_barMc}
x^*\in \mathcal{X}_c^*,
\quad
\overline M_c(x^*) > \frac{\varepsilon}{1-\varepsilon}
\quad\Longrightarrow\quad
\hat m(x^*) = y_c^*.
\end{equation}
By \eqref{eq:barMc_vs_tildeMc}, this further implies
\begin{equation}\label{eq:marg_cond_implies_correct_tildeMc}
x^*\in \mathcal{X}_c^*,
\quad
\widetilde M_c(x^*) > \frac{\varepsilon}{1-\varepsilon}
\quad\Longrightarrow\quad
\hat m(x^*) = y_c^*.
\end{equation}

\noindent\emph{Step 4.}
Let $X^{\mathrm{test}}$ denote a generic feature drawn from the test
distribution $P_0^X$, and write
\[
\varphi(x) \;:=\; \mathbf 1_{\{\hat m(g_c(x)) = y_c^*\}},
\qquad x\in\mathcal X.
\]
By definition,
\[
\hat S_c(\alpha_{N,c})
=
\frac{1}{N^{\mathrm{test}}}\sum_{i=1}^{N^{\mathrm{test}}}
\varphi(X_i^{\mathrm{test}}).
\]
Applying Hoeffding’s inequality to the i.i.d.\ Bernoulli variables
$\varphi(X_i^{\mathrm{test}})$, with failure budget $\delta/4$, we obtain that, with
probability at least $1-\delta/4$,
\begin{equation}\label{eq:test_vs_P0}
\hat S_c(\alpha_{N,c})
\;\ge\;
\PP_{X\sim P_0^X}\bigl[\hat m(g_c(X))=y_c^*\bigr]
-
R_{\delta/4}(N^{\mathrm{test}}).
\end{equation}

Since $g_c(X)\in \mathcal X_c^*$ for all $X\in\mathcal X$, \eqref{eq:marg_cond_implies_correct_barMc} implies
\[
\Bigl\{
\overline M_c(g_c(X))
>
\tfrac{\varepsilon}{1-\varepsilon}
\Bigr\}
\subseteq
\bigl\{
\hat m(g_c(X))=y_c^*
\bigr\}.
\]
Therefore, with probability at least $1-\delta/4$ conditional on the Step~3 event,
\begin{equation}\label{eq:test_vs_P0_margin_pop}
\hat S_c(\alpha_{N,c})
\;\ge\;
\PP_{X\sim P_0^X}\Bigl[
\overline M_c(g_c(X))
>
\tfrac{\varepsilon}{1-\varepsilon}
\Bigr]
-
R_{\delta/4}(N^{\mathrm{test}}).
\end{equation}

\noindent\emph{Step 5.}
Define the fixed event
\[
A_c^{\mathrm{pop}}
\;:=\;
\Bigl\{
x\in\mathcal X:
\overline M_c(g_c(x))
>
\tfrac{\varepsilon}{1-\varepsilon}
\Bigr\}.
\]
Conditional on all training data outside collective $c$, the event $A_c^{\mathrm{pop}}$ is fixed. By neutrality and the sampling model, the pre-edit features
$\{X_i : i\in I_{N,c}\}$ are i.i.d.\ with law $P_0^X$,
conditionally on $n_{N,c}$. Applying Hoeffding’s inequality to the i.i.d.\
variables $\mathbf 1_{\{X_i\in A_c^{\mathrm{pop}}\}}$, $i\in I_{N,c}$, with failure budget $\delta/4$, we obtain
that, with probability at least $1-\delta/4$,
\begin{equation}\label{eq:P0_vs_collective_pre_pop}
\PP_{X\sim P_0^X}\bigl[X\in A_c^{\mathrm{pop}}\bigr]
\;\ge\;
\PP_{x\sim\hat P_{N,c}^{\mathrm{pre},X}}\bigl[x\in A_c^{\mathrm{pop}}\bigr]
-
R_{\delta/4}(n_{N,c}).
\end{equation}

On the Step~3 event, \eqref{eq:barMc_vs_tildeMc} implies the pointwise inclusion
\[
\Bigl\{
\widetilde M_c(g_c(x))
>
\tfrac{\varepsilon}{1-\varepsilon}
\Bigr\}
\subseteq
A_c^{\mathrm{pop}}.
\]
Hence,
\begin{equation}\label{eq:P0_vs_collective_pre_margin}
\PP_{X\sim P_0^X}\Bigl[
\overline M_c(g_c(X))
>
\tfrac{\varepsilon}{1-\varepsilon}
\Bigr]
\;\ge\;
\PP_{x\sim\hat P_{N,c}^{\mathrm{pre},X}}
\Bigl[
\widetilde M_c(g_c(x))
>
\tfrac{\varepsilon}{1-\varepsilon}
\Bigr]
-
R_{\delta/4}(n_{N,c}).
\end{equation}

\noindent\emph{Step 6.}
By the union bound, the Step~3 concentration event, \eqref{eq:test_vs_P0},
and \eqref{eq:P0_vs_collective_pre_pop} hold simultaneously with probability at least
$1-\delta$. Combining \eqref{eq:test_vs_P0_margin_pop} and
\eqref{eq:P0_vs_collective_pre_margin}, we conclude that, with probability at
least $1-\delta$,
\begin{align*}
\hat S_c(\alpha_{N,c})
&\ge
\PP_{x\sim\hat P_{N,c}^{\mathrm{pre},X}}
\Bigl[
\widetilde M_c\bigl(g_c(x)\bigr)
>
\tfrac{\varepsilon}{1-\varepsilon}
\Bigr]
-
R_{\delta/4}(n_{N,c})
-
R_{\delta/4}(N^{\mathrm{test}})
\\
&=
\PP_{x\sim\hat P_{N,c}^{\mathrm{pre},X}}
\Big[
\;
\alpha_{N,c}\Bigl(
\hat p_c^X\bigl(g_c(x)\bigr)
-
2R_{\tilde\delta}(n_{N,c})
\Bigr)
-
\sum_{d\neq c}\alpha_{N,d}\,\Delta_{c\to d}\bigl(g_c(x)\bigr)
\\[-0.1em]
&
-\alpha_{N,0}\,\Delta_{c,0}^{(n_{N,c})}\bigl(g_c(x)\bigr)-\;
2\alpha_{N,0}\Bigl(
R_{\tilde\delta}(n_{N,0}) + R_{\tilde\delta}(n_{N,c})
\Bigr)
-
\frac{\varepsilon}{1-\varepsilon}
>0
\Big]\\[-0.1em]
&-
R_{\delta/4}(n_{N,c})
-
R_{\delta/4}(N^{\mathrm{test}}),
\end{align*}
which is exactly \eqref{eq:multi_FL_bound_final_clean_preonly}. This completes
the proof.
\end{proofE}

\cameraready{\textbf{Intuitive Interpretation.} The bound in \eqref{eq:multi_FL_bound_final_clean_preonly} formalizes a race at each planted feature value $x^*=g_c(x)$: collective $c$ succeeds at $x^*$ when the mass it concentrates on $(x^*,y^*_c)$ exceeds, by more than the classifier slack $\varepsilon/(1-\varepsilon)$ and the finite-sample radii, the largest mass that any other data source places on a competing label at the same $x^*$. The guarantee is the pre-edit probability that $c$'s signal wins this race, and each term identifies one racer: the own signal $\alpha_{N,c}\hat p^X_c$, the cross-collective interference $\sum_{d\neq c}\alpha_{N,d}\Delta_{c\to d}$, and the organic resistance of the non-collective population $\alpha_{N,0}\Delta^{(n_{N,c})}_{c,0}$. The interference terms are sign-indefinite, so the sign of $\Delta_{c\to d}$ classifies every other collective $d$ from the perspective of $c$, i.e., an \emph{ally} when $\Delta_{c\to d}<0$ ($d$'s edits support $y^*_c$ at $x^*$, and its mass adds to $c$'s signal), \emph{neutral} when $\Delta_{c\to d}=0$ (e.g., $d$ plants feature values outside $\mathcal{X}^*_c$), and an \emph{adversary} when $\Delta_{c\to d}>0$ ($d$ associates $x^*$ with a competing label). Unlike the single-collective setting, success is thus governed by an alignment-weighted effective mass rather than by size alone. The symmetric regime makes this explicit.}

\begingroup\crcolor
\begin{corollaryE}[Symmetric regime][end, restate]\label{cor:symmetric}
\crcolor In the setting of Theorem~\ref{thm:multi_feature_label}, let $M\ge2$ and $\alpha_{N,d}=\bar\alpha_N/M$ for every collective $d$, and fix $x^*\in\mathcal X_c^*$ with $\hat p_c^X(x^*)=p$, $\Delta_{c\to d}(x^*)=\Delta$ for all $d\neq c$, and $\Delta_{c,0}^{(n_{N,c})}(x^*)=\Delta_0$. Let $\kappa:=(1-\bar\alpha_N)\Delta_0+\frac{\varepsilon}{1-\varepsilon}+\rho_N$, where $\rho_N:=\frac{2\bar\alpha_N}{M}R_{\tilde\delta}(n_{N,c})+2(1-\bar\alpha_N)\bigl(R_{\tilde\delta}(n_{N,0})+R_{\tilde\delta}(n_{N,c})\bigr)$ collects the radii. Then the margin in~\eqref{eq:multi_FL_bound_final_clean_preonly} is positive at $x^*$ iff
\begin{equation}\label{eq:symmetric_margin}
p \;>\; (M-1)\,\Delta+\frac{M}{\bar\alpha_N}\,\kappa .
\end{equation}
In particular, (i) under full alignment ($\Delta=-p$) the condition reads $\bar\alpha_N\, p>\kappa$, (ii) under mutual neutrality ($\Delta=0$) it reads $\frac{\bar\alpha_N}{M}\,p>\kappa$, (iii) under full conflict ($\Delta=p$) it fails for every $M\ge2$ whenever $\kappa\ge0$, e.g., when $\Delta_0\ge0$.
\end{corollaryE}
\begin{proofE}
\crcolor With $\alpha_{N,c}=\bar\alpha_N/M$ for every collective and $\alpha_{N,0}=1-\bar\alpha_N$, the margin of \eqref{eq:multi_FL_bound_final_clean_preonly} at $x^*$ is
\begin{align*}
&\frac{\bar\alpha_N}{M}\bigl(p-2R_{\tilde\delta}(n_{N,c})\bigr)-(M-1)\frac{\bar\alpha_N}{M}\,\Delta-(1-\bar\alpha_N)\Delta_0\\
&\qquad-2(1-\bar\alpha_N)\bigl(R_{\tilde\delta}(n_{N,0})+R_{\tilde\delta}(n_{N,c})\bigr)-\frac{\varepsilon}{1-\varepsilon}
=\frac{\bar\alpha_N}{M}\bigl(p-(M-1)\Delta\bigr)-\kappa,
\end{align*}
which is positive if and only if \eqref{eq:symmetric_margin} holds, since $\bar\alpha_N/M>0$. Substituting $\Delta=-p$ and $\Delta=0$ gives (i) and (ii). For $\Delta=p$ the margin equals $\frac{\bar\alpha_N}{M}(2-M)\,p-\kappa\le-\kappa\le0$, because $M\ge2$ and $p\ge0$; this gives (iii).
\end{proofE}
\endgroup

\cameraready{Corollary~\ref{cor:symmetric} exhibits a trichotomy in the alignment. Under full alignment (all collectives plant the same label at the same features), the collectives pool into a single collective of mass $\bar\alpha_N$, and $M$ matters only through the radii. Under mutual neutrality only the individual mass $\bar\alpha_N/M$ counts, so fragmentation dilutes success by a factor of $M$. Under full conflict (equal-size collectives planting distinct labels at shared features), comparable collectives destroy each other's guarantees, consistently with the label-entropy collapse in Figure~\ref{fig:heatmap}. None of these phenomena exists at $M=1$.}

\textbf{Impact of Strategies.}
The lower bound in Theorem~\ref{thm:multi_feature_label} depends on the collective strategy through the feature map $g_c$ and the resulting empirical measures $\hat P_{N,c}$ and $\hat P_{N,d}$ that enter $\hat p_c^X(\cdot)$ and $\Delta_{c\to d}(\cdot)$. Indeed, the \emph{signal} term $\alpha_{N,c}\hat p_c^X(g_c(x))$ increases when the strategy concentrates the collective edited mass on the set $\mathcal X_c^*$, while the \textit{overall weighted conflict} $\sum_{d\neq c}\alpha_{N,d}\Delta_{c\to d}(g_c(x))$ \cameraready{reflects} how other collectives' edits create competing label associations at the same planted feature values. Under \textit{feature-label} planting, $h_c(x,y)=(g_c(x),y_c^*)$ forces all edited examples to carry $y_c^*$, typically boosting the signal term at planted values and eliminating within-collective label ambiguity at those values. Under \textit{feature-only} planting, labels are not changed, and only target-label points are mapped via $g_c$, which can strengthen separation by reducing the presence of non-$y_c^*$ labels inside $\mathcal X_c^*$ but may limit how much mass can be moved onto the planted set when $y\neq y_c^*$ points cannot be relabeled.

\textbf{On the Computability of the Bound.} The lower bound in~\eqref{eq:multi_FL_bound_final_clean_preonly} can be evaluated by the sample collective $I_{N,c}$ if it has access to its own pre-edit sample and strategy, and a limited set of global aggregate quantities. Concretely, collective $c$ can compute $\hat P_{N,c}^{\mathrm{pre}}$ and its marginal $\hat P_{N,c}^{\mathrm{pre},X}$ from $\{(X_i,Y_i):i\in I_{N,c}\}$, evaluate the proxy non-collective conflict $\Delta_{c,0}^{(n_{N,c})}(x^*)$ for $x^*\in\mathcal X_c^*$ using only $\hat P_{N,c}^{\mathrm{pre}}$ (under neutrality), and then estimate the outer probability $\PP_{x\sim \hat P_{N,c}^{\mathrm{pre},X}}[\cdot]$ as the fraction of its own pre-edit features for which the corresponding margin at $g_c(x)$ is positive. The remaining terms are deterministic given the radii, the confidence level $\delta$, and the suboptimality parameter $\varepsilon$. The substantive non-local \cameraready{requirement} is the overall weighted conflict $\sum_{d\neq c}\alpha_{N,d}\,\Delta_{c\to d}\bigl(g_c(x)\bigr)$; in particular, for each $d\neq c$, it depends on the empirical mass $\alpha_{N,d}$ and the conflict  $\Delta_{c\to d}(x^*)$ at the relevant signal features $x^*=g_c(x)$, with $\Delta_{c\to d}(x^*)=\max_{y'\neq y_c^*}\big(\hat P_{N,d}(x^*,y')-\hat P_{N,d}(x^*,y_c^*)\big)$. Here, we propose different ways to relax these requirements based on how much information about other collectives a collective can access. The best-case scenario is the one in which the collective $I_{N,c}$ can access a single ``others'' term $\Delta_c^{\mathrm{others}}(x^*)$ that a platform (or trusted aggregator) publishes at the needed signal features without revealing per-collective distributions. Alternatively, the worst case scenario is upper-bounding $\sum_{d\neq c}\alpha_{N,d}\Delta_{c\to d}(x^*)$ by the total collective mass $\bar\alpha_N-\alpha_{N,c}$ (or by a known upper bound $\bar\alpha_N^{\max}-\alpha_{N,c}$), yielding a generally more conservative bound. Finally, an intermediate scenario is to assume mild symmetry or collective anonymity: while different collectives may adopt different edits, their post-edit behavior at the signal features $x^*=g_c(x)$ is statistically homogeneous in the sense that, from the perspective of collective $I_{N,c}$, another collective induces a conflict level that can be approximated by a proxy being a functional of the $I_{N,c}$'s pre-edit distribution. Under neutrality, $I_{N,c}$ can then estimate $\Delta_{c\to d}(x^*)$ from its own pre-edit sample by relying on such a proxy $\Phi$ so that $\Delta_{c\to d}(x^*) \approx \Phi\!\bigl(\hat P_{N,c}^{\mathrm{pre}},x^*\bigr)$,
for all  $d\neq c$ and relevant $x^*\in\mathcal X_c^*$, where $\Phi$ can be derived from a shared model class for strategies (e.g., edits drawn from a known family, or strategies that coincide up to relabeling/noise). Overall, this approach trades stronger structural assumptions (exchangeability/anonymous behavior) for a substantially weaker global information requirement.

\textbf{On the Consistency of the Bound.} In Appendix \ref{app:single_collective}, we show how the bound in Theorem \ref{thm:multi_feature_label} reduces to the single collective bounds in Theorems 3.3 and 3.5 of \cite{gauthier2025statistical} when $M=1$.

\vspace{-0.3cm}\section{Signal Unplanting}\label{sec:signal_unplanting}\vspace{-0.2cm}
We now analyze the case where a sample collective $I_{N,c}$ aims to \emph{remove} the association between the planted features $g_c(x)$ and its original target class $y_c^*$. Accordingly, each sample collective $I_{N,c}$ measures per-collective success as
\begin{equation}\label{eq:unplant_success}
\hat S_c(\alpha_{N,c})
\;:=\;
\PP_{x \sim \hat P^X_{N^{\textrm{test}}}}\!\left[\hat m(g_c(x)) \neq y_c^*\right].
\end{equation}
Throughout this section, we reuse the notation and objects from Section~\ref{sec:signal_planting} when possible and again employ and generalize the strategies proposed in \cite{gauthier2025statistical}.

\textbf{Naive Strategy.}
A simple way to increase \eqref{eq:unplant_success} is to pick any alternative label $y'_c\in \mathcal Y\setminus\{y_c^*\}$ and \emph{plant} the association $(g_c(x),y'_c)$ as explained in Section \ref{sec:signal_planting}. Indeed, for any such $y'_c$, it holds
\begin{equation}\label{eq:unplant_ge_plant}
\PP_{x \sim \hat P^X_{N^{\textrm{test}}}}\!\left[\hat m(g_c(x)) \neq y_c^*\right]
\;\ge\;
\PP_{x \sim \hat P^X_{N^{\textrm{test}}}}\!\left[\hat m(g_c(x)) = y'_c\right].
\end{equation}
Thus, the sample collective $I_{N,c}$ may choose $y'_c$ to maximize the right-hand side of \eqref{eq:unplant_ge_plant} and then directly leverage Theorem~\ref{thm:multi_feature_label} to perform unplanting. However, this reduction forces a single competing label at all planted features.

\textbf{Adaptive Strategy.} An alternative possibility is making the label $y'_c$ depend on the planted feature value $x^*\in\mathcal X_c^*$. Intuitively, at each $x^*$, the collective chooses a label $y'_c(x^*)$ that makes $y_c^*$ least competitive at that feature. To select $y'_c(x^*)$, the collective $I_{N,c}$ selects an estimation sub-collective $E_{N,c}\subseteq I_{N,c}$ of size $0 < n_{e,c}<n_{N,c}$, independently of all data. We denote with $R_{N,c}:=I_{N,c}\setminus E_{N,c}$ the remaining fraction of the collective of size $n_{N,c}-n_{e,c}$. Thus, the corresponding pre-edit empirical measures are defined as
\begin{gather}
\hat P_{N,c}^{\mathrm{pre},E}
:=
\frac{1}{n_{e,c}}\sum_{i\in E_{N,c}}\delta_{Z_i},\quad \hat P_{N,c}^{\mathrm{pre},R}
:=
\frac{1}{n_{N,c}-n_{e,c}}\sum_{i\in R_{N,c}}\delta_{Z_i}.
\end{gather}
In this setting, for each $x^*\in \mathcal X_c^*$, the feature-dependent selected label is chosen as
\begin{equation}\label{eq:unplant_hat_y}
y'_c(x^*)
\;:=\;
\underset{y \in \mathcal{Y} \setminus \{y_c^*\}}{\arg\max}\;
\hat P_{N,c}^{\mathrm{pre},E}(x^*,y),
\end{equation}
The resulting \textit{adaptive unplanting strategy} is the feature--label map
\begin{equation}\label{eq:unplant_strat}
h_c(x,y)
\;:=\;
\bigl(g_c(x),\, y'_c(g_c(x))\bigr).
\end{equation}
In other words, collective $I_{N,c}$ plants the same features as in signal planting, but assigns to each planted feature $x^*$ a competing label $y'_c(x^*)\neq y_c^*$ chosen using  $E_{N,c}$.

For each $x^*\in\mathcal X_c^*$, we define the empirical (unplanting) conflict terms with other collectives as
\begin{align}\label{eq:unp_conflict}
\Delta^{\mathrm{unp}}_{c\leftarrow d}(x^*)
\;&:=\;
\hat P_{N,d}(x^*,y_c^*)-\hat P_{N,d}\!\bigl(x^*,y'_c(x^*)\bigr), \quad d\in\{1,\dots,M\},\ d\neq c,
\end{align}
and the corresponding non-collective approximate conflict using only data independent of the label-selection step as
\begin{align}\label{eq:unp_conflict0_proxy}
\Delta&^{\mathrm{unp},(n_{N,c}-n_{e,c})}_{c,0}(x^*)
\;=\; \hat P_{N,c}^{\mathrm{pre},R}(x^*,y_c^*)
-\hat P_{N,c}^{\mathrm{pre},R}\!\bigl(x^*,y'_c(x^*)\bigr).
\end{align}
Intuitively, $\Delta^{\mathrm{unp}}_{c\leftarrow d}(x^*)$ is large when other collectives still support $(x^*,y_c^*)$ more than $(x^*,y'_c(x^*))$, thereby resisting unplanting at $x^*$. The approximate conflict in \eqref{eq:unp_conflict0_proxy} plays the analogous role for non-collective users under neutrality. We can now state our main result for this section.

\begin{theoremE}[Signal unplanting with multiple collectives][end, restate]\label{thm:multi_unplant}
Under the neutrality assumption, assuming that $\hat m$ is
$\varepsilon$-contamination-suboptimal on $\mathcal X_c^*$ under
$\hat P_N$, assuming all collectives employ the adaptive unplanting strategy from \eqref{eq:unplant_strat} with estimation splits $\{E_{N,c}\}_c$, then, with probability at least $1-\delta$ over the draw of the training and test samples, with $\delta > 0$, and the random choice of $\{E_{N,c}\}_c$, the per-collective success of $I_{N,c}$ satisfies
\begin{align}\label{eq:multi_unplant_bound}
\hat S_c(\alpha_{N,c})
\;&\ge\;\PP_{x\sim\hat P_{N,c}^{\mathrm{pre},R,X}}
\Big[
\;
\alpha_{N,c}\Bigl(
\hat p_c^X\bigl(g_c(x)\bigr)
-
2R_{\tilde\delta}(n_{N,c})
\Bigr)
-
\sum_{d\neq c}\alpha_{N,d}\,
\Delta^{\mathrm{unp}}_{c\leftarrow d}\bigl(g_c(x)\bigr)
\nonumber \\
&-
\alpha_{N,0}\,\Delta^{\mathrm{unp},(n_{N,c}-n_{e,c})}_{c,0}\bigl(g_c(x)\bigr)-
2\alpha_{N,0}\Bigl(
R_{\tilde\delta}(n_{N,0}) + R_{\tilde\delta}(n_{N,c}-n_{e,c})
\Bigr)
-
\frac{\varepsilon}{1-\varepsilon}
>0
\Big] \nonumber\\
&-
R_{\delta/4}(n_{N,c}-n_{e,c})
-
R_{\delta/4}(N^{\mathrm{test}}),
\end{align}
where $K_c \;:=\; |\mathcal{X}_c^*|\,|\mathcal Y|$ and $\tilde\delta \;:=\; \frac{\delta}{10K_c}$.
\end{theoremE}
\begin{proofE}
The proof follows closely the same six-step structure as
Theorem~\ref{thm:multi_feature_label}. We work conditionally on the realized
sample sizes
\[
\{n_{N,0},n_{N,1},\dots,n_{N,M}\}.
\]
The split $E_{N,c}$ is chosen independently of the data. For concentration
bounds involving the data-dependent labels $y'_c(\cdot)$, we further condition
on the estimation sigma-field
\[
\mathcal F_{E,c}
:=
\sigma\!\left(E_{N,c},\{Z_i:i\in E_{N,c}\}\right).
\]
Under this conditioning, the map $y'_c(\cdot)$ is fixed, while the holdout
pre-edit points in $R_{N,c}$ remain i.i.d.\ from $P_0$ and independent of the
label-selection step. The concentration bound for the post-edit feature
marginal $\hat p_c^X$ is justified separately, before conditioning on the
observed estimation data, because the feature part of the adaptive strategy is
the fixed map $g_c$ and does not depend on $y'_c(\cdot)$.

\noindent\emph{Step 1.}
We first prove a simple margin implication.

\begin{lemma}[Contamination margin implication: unplanting]
Assume $\hat m$ is $\varepsilon$-contamination-suboptimal on
$\mathcal X'$ under $\hat P_N$. If $x^*\in\mathcal X'$ and there exists
$y'\neq y^*$ such that
\[
\hat P_N(x^*,y')
>
\hat P_N(x^*,y^*)+\frac{\varepsilon}{1-\varepsilon},
\]
then
\[
\hat m(x^*)\neq y^*.
\]
\end{lemma}

\begin{proof}
By $\varepsilon$-contamination-suboptimality, there exist probability
measures $\hat Q_N$ and $R_N$ on $\mathcal X\times\mathcal Y$ such that
\[
\hat Q_N=(1-\varepsilon)\hat P_N+\varepsilon R_N,
\]
and such that, for every $x\in\mathcal X'$ with $\hat Q_N^X(x)>0$,
\[
\hat m(x)\in \argmax_{y\in\mathcal Y}\hat Q_N(y\mid x).
\]
Fix $x^*\in\mathcal X'$ and suppose that, for some $y'\neq y^*$,
\[
\hat P_N(x^*,y')
>
\hat P_N(x^*,y^*)+\frac{\varepsilon}{1-\varepsilon}.
\]
Then
\[
\hat P_N(x^*,y')-\hat P_N(x^*,y^*)
>
\frac{\varepsilon}{1-\varepsilon}.
\]
Using the contamination representation, we have
\begin{align*}
\hat Q_N(x^*,y')-\hat Q_N(x^*,y^*)
&=
(1-\varepsilon)
\bigl(
\hat P_N(x^*,y')-\hat P_N(x^*,y^*)
\bigr)
\\
&\quad
+
\varepsilon
\bigl(
R_N(x^*,y')-R_N(x^*,y^*)
\bigr).
\end{align*}
Since $R_N$ is a probability measure,
\[
R_N(x^*,y')-R_N(x^*,y^*)\ge -1.
\]
Therefore,
\begin{align*}
\hat Q_N(x^*,y')-\hat Q_N(x^*,y^*)
&>
(1-\varepsilon)\frac{\varepsilon}{1-\varepsilon}
-\varepsilon
\\
&=0.
\end{align*}
Hence
\[
\hat Q_N(x^*,y')>\hat Q_N(x^*,y^*).
\]
Moreover, the strict inequality
\[
\hat P_N(x^*,y')>
\hat P_N(x^*,y^*)+\frac{\varepsilon}{1-\varepsilon}
\]
implies $\hat P_N^X(x^*)>0$. Since $\varepsilon<1$, we also have
\[
\hat Q_N^X(x^*)
=
(1-\varepsilon)\hat P_N^X(x^*)+\varepsilon R_N^X(x^*)
>0.
\]
Thus the conditional probabilities $\hat Q_N(y\mid x^*)$ are well defined.
Because
\[
\hat Q_N(x^*,y')>\hat Q_N(x^*,y^*),
\]
and both labels share the same marginal denominator $\hat Q_N^X(x^*)$, we get
\[
\hat Q_N(y'\mid x^*)>\hat Q_N(y^*\mid x^*).
\]
Therefore $y^*$ cannot be an argmax of $\hat Q_N(\cdot\mid x^*)$. Since
$\hat m(x^*)$ is chosen among the argmax labels under $\hat Q_N$, it follows
that
\[
\hat m(x^*)\neq y^*.
\]
\end{proof}

Since $\hat m$ is $\varepsilon$-contamination-suboptimal on
$\mathcal X_c^*$ under $\hat P_N$, the margin implication above gives that,
for any $x^*\in\mathcal X_c^*$,
\begin{equation}\label{eq:unplant_margin_implies_pred}
\hat P_N(x^*,y'_c(x^*))
>
\hat P_N(x^*,y_c^*)+\frac{\varepsilon}{1-\varepsilon}
\quad\Longrightarrow\quad
\hat m(x^*)\neq y_c^*.
\end{equation}

\noindent\emph{Step 2.}
Fix $x^*\in \mathcal{X}_c^*$. Using the empirical mixture decomposition of
$\hat P_N$ and the adaptive feature--label unplanting property
\[
\hat P_{N,c}(x^*,y_c^*)=0,
\qquad
\hat P_{N,c}(x^*,y'_c(x^*))=\hat p_c^X(x^*),
\]
we obtain
\begin{align*}
&\hat P_N(x^*,y'_c(x^*))-\hat P_N(x^*,y_c^*)
\\
&\quad =
\alpha_{N,c}\hat p_c^X(x^*)
+
\sum_{d\neq c}
\alpha_{N,d}
\Bigl(
\hat P_{N,d}(x^*,y'_c(x^*))
-
\hat P_{N,d}(x^*,y_c^*)
\Bigr)
\\
&\qquad
+
\alpha_{N,0}
\Bigl(
\hat P_{N,0}(x^*,y'_c(x^*))
-
\hat P_{N,0}(x^*,y_c^*)
\Bigr)
\\
&\quad =
\alpha_{N,c}\hat p_c^X(x^*)
-
\sum_{d\neq c}
\alpha_{N,d}\Delta^{\mathrm{unp}}_{c\leftarrow d}(x^*)
-
\alpha_{N,0}\Delta^{\mathrm{unp}}_{c,0}(x^*),
\end{align*}
where
\[
\Delta^{\mathrm{unp}}_{c,0}(x^*)
:=
\hat P_{N,0}(x^*,y_c^*)
-
\hat P_{N,0}(x^*,y'_c(x^*)).
\]
Define the exact empirical unplanting margin
\[
M_c^{\mathrm{unp}}(x^*)
:=
\alpha_{N,c}\hat p_c^X(x^*)
-
\sum_{d\neq c}\alpha_{N,d}\Delta^{\mathrm{unp}}_{c\leftarrow d}(x^*)
-
\alpha_{N,0}\Delta^{\mathrm{unp}}_{c,0}(x^*).
\]
Then, by \eqref{eq:unplant_margin_implies_pred},
\begin{equation}\label{eq:unplant_marg_cond_implies_exact}
x^*\in\mathcal X_c^*,
\quad
M_c^{\mathrm{unp}}(x^*)>\frac{\varepsilon}{1-\varepsilon}
\quad\Longrightarrow\quad
\hat m(x^*)\neq y_c^*.
\end{equation}

\noindent\emph{Step 3.}
Define the population signal prevalence
\[
p_{c,\mathrm{pop}}^X(x^*)
:=
(g_c)_{\#}P_0^X(x^*)
=
P_0^X\!\left(g_c^{-1}(\{x^*\})\right),
\qquad x^*\in\mathcal X_c^*.
\]
Since the feature part of the adaptive unplanting strategy is the fixed map
$g_c$, the empirical post-edit feature marginal $\hat p_c^X$ is the empirical
marginal of the fixed pushforward of the collective pre-edit sample under
$g_c$. Hence, conditionally on $n_{N,c}$, Hoeffding's inequality gives the
two-sided feature-marginal concentration bounds
\begin{align}
\hat p_c^X(x^*)
&\ge
p_{c,\mathrm{pop}}^X(x^*) - R_{\tilde\delta}(n_{N,c}),
\qquad \forall x^*\in\mathcal X_c^*,
\label{eq:unplant_hoeff_signal_collective}
\\[0.3em]
\hat p_c^X(x^*)
&\le
p_{c,\mathrm{pop}}^X(x^*) + R_{\tilde\delta}(n_{N,c}),
\qquad \forall x^*\in\mathcal X_c^*.
\label{eq:unplant_hoeff_signal_collective_upper}
\end{align}

Next, condition on the estimation sigma-field $\mathcal F_{E,c}$. Under this
conditioning, the map $y'_c(\cdot)$ is fixed, and the holdout pre-edit points
in $R_{N,c}$ are i.i.d.\ from $P_0$. Define the population non-collective
unplanting conflict
\[
\Delta_{c,0}^{\mathrm{unp,pop}}(x^*)
:=
P_0(x^*,y_c^*)
-
P_0(x^*,y'_c(x^*)).
\]
Together with the two feature-marginal concentration families
\eqref{eq:unplant_hoeff_signal_collective}--\eqref{eq:unplant_hoeff_signal_collective_upper},
Hoeffding's inequality and a union bound over the scalar one-sided events
displayed below imply that, with probability at least $1-\delta/2$, all the
concentration inequalities in this step hold simultaneously:
\begin{align}
\hat P_{N,0}(x^*,y)
&\le
P_0(x^*,y) + R_{\tilde\delta}(n_{N,0}),
\qquad \forall (x^*,y)\in\mathcal X_c^*\times\mathcal Y,
\label{eq:unplant_hoeff_noncoll_upper}
\\[0.3em]
P_0(x^*,y)
&\le
\hat P_{N,0}(x^*,y) + R_{\tilde\delta}(n_{N,0}),
\qquad \forall (x^*,y)\in\mathcal X_c^*\times\mathcal Y,
\label{eq:unplant_hoeff_noncoll_lower}
\\[0.3em]
\hat P_{N,c}^{\mathrm{pre},R}(x^*,y)
&\le
P_0(x^*,y) + R_{\tilde\delta}(n_{N,c}-n_{e,c}),
\qquad \forall (x^*,y)\in\mathcal X_c^*\times\mathcal Y,
\label{eq:unplant_hoeff_holdout_upper}
\\[0.3em]
P_0(x^*,y)
&\le
\hat P_{N,c}^{\mathrm{pre},R}(x^*,y) + R_{\tilde\delta}(n_{N,c}-n_{e,c}),
\qquad \forall (x^*,y)\in\mathcal X_c^*\times\mathcal Y.
\label{eq:unplant_hoeff_holdout_lower}
\end{align}
Indeed, the two feature-marginal families contain
$2|\mathcal X_c^*|$ scalar events and the four joint-mass families contain
$4|\mathcal X_c^*||\mathcal Y|=4K_c$ scalar events. Since $|\mathcal Y|\ge2$,
the total number of scalar events is at most $5K_c$. With
$\tilde\delta=\delta/(10K_c)$, the total failure probability is therefore at
most $\delta/2$.

Using
\eqref{eq:unplant_hoeff_noncoll_upper}--\eqref{eq:unplant_hoeff_noncoll_lower},
we obtain, for every $x^*\in\mathcal X_c^*$,
\begin{equation}\label{eq:unplant_delta_exact_vs_pop}
\Delta_{c,0}^{\mathrm{unp}}(x^*)
\le
\Delta_{c,0}^{\mathrm{unp,pop}}(x^*)
+
2R_{\tilde\delta}(n_{N,0}).
\end{equation}
Using
\eqref{eq:unplant_hoeff_holdout_upper}--\eqref{eq:unplant_hoeff_holdout_lower},
we obtain, for every $x^*\in\mathcal X_c^*$,
\begin{equation}\label{eq:unplant_delta_pop_vs_proxy}
\Delta_{c,0}^{\mathrm{unp,pop}}(x^*)
\le
\Delta_{c,0}^{\mathrm{unp},(n_{N,c}-n_{e,c})}(x^*)
+
2R_{\tilde\delta}(n_{N,c}-n_{e,c}).
\end{equation}

Define the fixed population-based unplanting margin
\begin{align*}
\overline M_c^{\mathrm{unp}}(x^*)
:=
\alpha_{N,c}\Bigl(
p_{c,\mathrm{pop}}^X(x^*)
-
R_{\tilde\delta}(n_{N,c})
\Bigr)
&-
\sum_{d\neq c}\alpha_{N,d}\Delta^{\mathrm{unp}}_{c\leftarrow d}(x^*)\\
&-
\alpha_{N,0}\Bigl(
\Delta_{c,0}^{\mathrm{unp,pop}}(x^*)
+
2R_{\tilde\delta}(n_{N,0})
\Bigr),
\end{align*}
and the corrected empirical unplanting margin
\begin{align*}
\widetilde M_c^{\mathrm{unp}}(x^*)
:=
\alpha_{N,c}\Bigl(
\hat p_c^X(x^*)
-
2R_{\tilde\delta}(n_{N,c})
\Bigr)
&-
\sum_{d\neq c}\alpha_{N,d}\Delta^{\mathrm{unp}}_{c\leftarrow d}(x^*)
-
\alpha_{N,0}\Delta_{c,0}^{\mathrm{unp},(n_{N,c}-n_{e,c})}(x^*)\\
&-
2\alpha_{N,0}\Bigl(
R_{\tilde\delta}(n_{N,0})
+
R_{\tilde\delta}(n_{N,c}-n_{e,c})
\Bigr).
\end{align*}

On the concentration event above, for all $x^*\in\mathcal X_c^*$,
\begin{equation}\label{eq:unplant_exact_vs_pop_margin}
M_c^{\mathrm{unp}}(x^*)
\ge
\overline M_c^{\mathrm{unp}}(x^*),
\end{equation}
and, by
\eqref{eq:unplant_hoeff_signal_collective_upper}
and
\eqref{eq:unplant_delta_pop_vs_proxy},
\begin{equation}\label{eq:unplant_pop_vs_tilde_margin}
\overline M_c^{\mathrm{unp}}(x^*)
\ge
\widetilde M_c^{\mathrm{unp}}(x^*).
\end{equation}

Combining \eqref{eq:unplant_marg_cond_implies_exact} and
\eqref{eq:unplant_exact_vs_pop_margin}, we obtain
\begin{equation}\label{eq:unplant_marg_cond_implies_pop}
x^*\in\mathcal X_c^*,
\quad
\overline M_c^{\mathrm{unp}}(x^*)>
\frac{\varepsilon}{1-\varepsilon}
\quad\Longrightarrow\quad
\hat m(x^*)\neq y_c^*.
\end{equation}
By \eqref{eq:unplant_pop_vs_tilde_margin}, this further implies
\begin{equation}\label{eq:unplant_marg_cond_implies_tilde}
x^*\in\mathcal X_c^*,
\quad
\widetilde M_c^{\mathrm{unp}}(x^*)>
\frac{\varepsilon}{1-\varepsilon}
\quad\Longrightarrow\quad
\hat m(x^*)\neq y_c^*.
\end{equation}

\noindent\emph{Step 4.}
Let
\[
\varphi(x)
:=
\mathbf 1_{\{\hat m(g_c(x))\neq y_c^*\}},
\qquad x\in\mathcal X.
\]
By definition,
\[
\hat S_c(\alpha_{N,c})
=
\frac{1}{N^{\mathrm{test}}}
\sum_{i=1}^{N^{\mathrm{test}}}\varphi(X_i^{\mathrm{test}}).
\]
Conditionally on the training sample, the variables
$\varphi(X_i^{\mathrm{test}})$ are i.i.d.\ Bernoulli random variables.
Applying Hoeffding's inequality with failure budget $\delta/4$, we obtain
that, with probability at least $1-\delta/4$,
\begin{equation}\label{eq:unplant_test_vs_P0}
\hat S_c(\alpha_{N,c})
\ge
\PP_{X\sim P_0^X}\bigl[\hat m(g_c(X))\neq y_c^*\bigr]
-
R_{\delta/4}(N^{\mathrm{test}}).
\end{equation}

Since $g_c(X)\in\mathcal X_c^*$ for all $X\in\mathcal X$,
\eqref{eq:unplant_marg_cond_implies_pop} implies
\[
\left\{
\overline M_c^{\mathrm{unp}}(g_c(X))
>
\frac{\varepsilon}{1-\varepsilon}
\right\}
\subseteq
\bigl\{
\hat m(g_c(X))\neq y_c^*
\bigr\}.
\]
Therefore, on the Step~3 concentration event,
\begin{equation}\label{eq:unplant_test_vs_P0_margin_pop}
\hat S_c(\alpha_{N,c})
\ge
\PP_{X\sim P_0^X}
\left[
\overline M_c^{\mathrm{unp}}(g_c(X))
>
\frac{\varepsilon}{1-\varepsilon}
\right]
-
R_{\delta/4}(N^{\mathrm{test}}).
\end{equation}

\noindent\emph{Step 5.}
Define the fixed event
\[
A_c^{\mathrm{unp,pop}}
:=
\left\{
x\in\mathcal X:
\overline M_c^{\mathrm{unp}}(g_c(x))
>
\frac{\varepsilon}{1-\varepsilon}
\right\}.
\]
Conditional on $\mathcal F_{E,c}$ and on all training data outside the holdout
split $R_{N,c}$, the event $A_c^{\mathrm{unp,pop}}$ is fixed. By neutrality
and the sampling model, the pre-edit features
$\{X_i:i\in R_{N,c}\}$ are i.i.d.\ with law $P_0^X$, conditionally on
$n_{N,c}-n_{e,c}$. Applying Hoeffding's inequality to the i.i.d.\ variables
$\mathbf 1_{\{X_i\in A_c^{\mathrm{unp,pop}}\}}$, $i\in R_{N,c}$, with failure
budget $\delta/4$, we obtain that, with probability at least $1-\delta/4$,
\begin{equation}\label{eq:unplant_P0_vs_holdout_pop}
\PP_{X\sim P_0^X}\bigl[X\in A_c^{\mathrm{unp,pop}}\bigr]
\ge
\PP_{x\sim\hat P_{N,c}^{\mathrm{pre},R,X}}
\bigl[x\in A_c^{\mathrm{unp,pop}}\bigr]
-
R_{\delta/4}(n_{N,c}-n_{e,c}).
\end{equation}

On the Step~3 concentration event, \eqref{eq:unplant_pop_vs_tilde_margin}
implies the pointwise inclusion
\[
\left\{
\widetilde M_c^{\mathrm{unp}}(g_c(x))
>
\frac{\varepsilon}{1-\varepsilon}
\right\}
\subseteq
A_c^{\mathrm{unp,pop}}.
\]
Hence,
\begin{equation}\label{eq:unplant_P0_vs_holdout_margin}
\PP_{X\sim P_0^X}
\left[
\overline M_c^{\mathrm{unp}}(g_c(X))
>
\frac{\varepsilon}{1-\varepsilon}
\right]
\ge
\PP_{x\sim\hat P_{N,c}^{\mathrm{pre},R,X}}
\left[
\widetilde M_c^{\mathrm{unp}}(g_c(x))
>
\frac{\varepsilon}{1-\varepsilon}
\right]
-
R_{\delta/4}(n_{N,c}-n_{e,c}).
\end{equation}

\noindent\emph{Step 6.}
By the union bound, the Step~3 concentration event,
\eqref{eq:unplant_test_vs_P0}, and \eqref{eq:unplant_P0_vs_holdout_pop} hold
simultaneously with probability at least $1-\delta$. Combining
\eqref{eq:unplant_test_vs_P0_margin_pop} and
\eqref{eq:unplant_P0_vs_holdout_margin}, we conclude that, with probability at
least $1-\delta$,
\begin{align*}
\hat S_c(\alpha_{N,c})
&\ge
\PP_{x\sim\hat P_{N,c}^{\mathrm{pre},R,X}}
\left[
\widetilde M_c^{\mathrm{unp}}\bigl(g_c(x)\bigr)
>
\frac{\varepsilon}{1-\varepsilon}
\right]
-
R_{\delta/4}(n_{N,c}-n_{e,c})
-
R_{\delta/4}(N^{\mathrm{test}})
\\
&=
\PP_{x\sim\hat P_{N,c}^{\mathrm{pre},R,X}}
\Big[
\;
\alpha_{N,c}\Bigl(
\hat p_c^X\bigl(g_c(x)\bigr)
-
2R_{\tilde\delta}(n_{N,c})
\Bigr)
-
\sum_{d\neq c}\alpha_{N,d}\,
\Delta^{\mathrm{unp}}_{c\leftarrow d}\bigl(g_c(x)\bigr)
\\
&\qquad
-
\alpha_{N,0}\,
\Delta_{c,0}^{\mathrm{unp},(n_{N,c}-n_{e,c})}
\bigl(g_c(x)\bigr)
-
2\alpha_{N,0}
\Bigl(
R_{\tilde\delta}(n_{N,0})
+
R_{\tilde\delta}(n_{N,c}-n_{e,c})
\Bigr)
\\
&\qquad
-
\frac{\varepsilon}{1-\varepsilon}
>0
\Big]
-
R_{\delta/4}(n_{N,c}-n_{e,c})
-
R_{\delta/4}(N^{\mathrm{test}}),
\end{align*}
which is exactly \eqref{eq:multi_unplant_bound}. This completes the proof.
\end{proofE}

\cameraready{\textbf{Intuitive Interpretation.} The bound in \eqref{eq:multi_unplant_bound} mirrors \eqref{eq:multi_FL_bound_final_clean_preonly} with the race reversed: collective $c$ succeeds at $x^*$ when the association it wants to erase, $(x^*,y^*_c)$, is overtaken by the data-driven alternative $(x^*,y'_c(x^*))$ with a margin of $\varepsilon/(1-\varepsilon)$. Accordingly, $\Delta^{\mathrm{unp}}_{c\leftarrow d}$ measures the resistance of collective $d$, positive when $d$'s data keeps supporting $(x^*,y^*_c)$ against the alternative and negative when $d$, deliberately or not, helps demote it, while $\alpha_{N,0}\Delta^{\mathrm{unp},(n_{N,c}-n_{e,c})}_{c,0}$ estimates the organic strength of the association in the general population, which is what the collective fights. This reveals an asymmetry between the two goals. In fact, planting can target rare or synthetic feature values whose organic resistance is negligible, whereas unplanting confronts the ambient support for an existing association and pays a cost proportional to the non-collective mass $\alpha_{N,0}$. Corollary~\ref{cor:symmetric} carries over verbatim with $\Delta^{\mathrm{unp}}$ in place of $\Delta$ and the radii of Theorem~\ref{thm:multi_unplant}, and yields the same trichotomy: aligned unplanters pool their masses, neutral ones act alone, and collectives protecting associations that others erase generate destructive interference. The estimation split $E_{N,c}$ can be any subset of the collective's sample chosen independently of the data, e.g., uniformly at random as in our experiments. This independence keeps the holdout $R_{N,c}$ i.i.d., as the concentration argument requires. Its size governs a bias--variance trade-off, because enlarging $n_{e,c}$ improves the label-selection step, since $y'_c(x^*)$ is more likely the strongest competitor at $x^*$, but it shrinks the holdout $n_{N,c}-n_{e,c}$ on which the bound is evaluated and inflates the radii $R_{\tilde\delta}(n_{N,c}-n_{e,c})$ and $R_{\delta/4}(n_{N,c}-n_{e,c})$. Since the selection step only ranks the competing labels at each planted value, rather than estimating their probabilities, a sublinear split ($n_{e,c}\ll n_{N,c}$) suffices.}

\textbf{On the Computability of the Bound.} \cameraready{The discussion of Section~\ref{sec:signal_planting} carries over. Collective $I_{N,c}$ selects the alternative labels $y'_c(x^*)$ on its estimation sub-sample as in~\eqref{eq:unplant_hat_y}, and evaluates the non-collective conflict and the outer probability on the holdout sub-sample, which is independent of the label selection; the radii now also depend on the split size $n_{e,c}$. The only non-local requirement is again the weighted cross-collective term $\sum_{d\neq c}\alpha_{N,d}\Delta^{\mathrm{unp}}_{c\leftarrow d}(g_c(x))$, with $\Delta^{\mathrm{unp}}_{c\leftarrow d}(x^*)=\hat P_{N,d}(x^*,y_c^*)-\hat P_{N,d}(x^*,y'_c(x^*))$, so the same hierarchy of approximations applies.}

\textbf{On the Consistency of the Bound.} Similarly to Section \ref{sec:signal_planting}, in Appendix \ref{app:single_collective}, we show how the bound in Theorem \ref{thm:multi_unplant} reduces to the single collective bounds in Theorems 3.7 of \cite{gauthier2025statistical} when $M=1$.
\begin{figure*}[h] 
    \centering
    \begin{subfigure}[]{0.3\linewidth}
        \includegraphics[width=\linewidth]{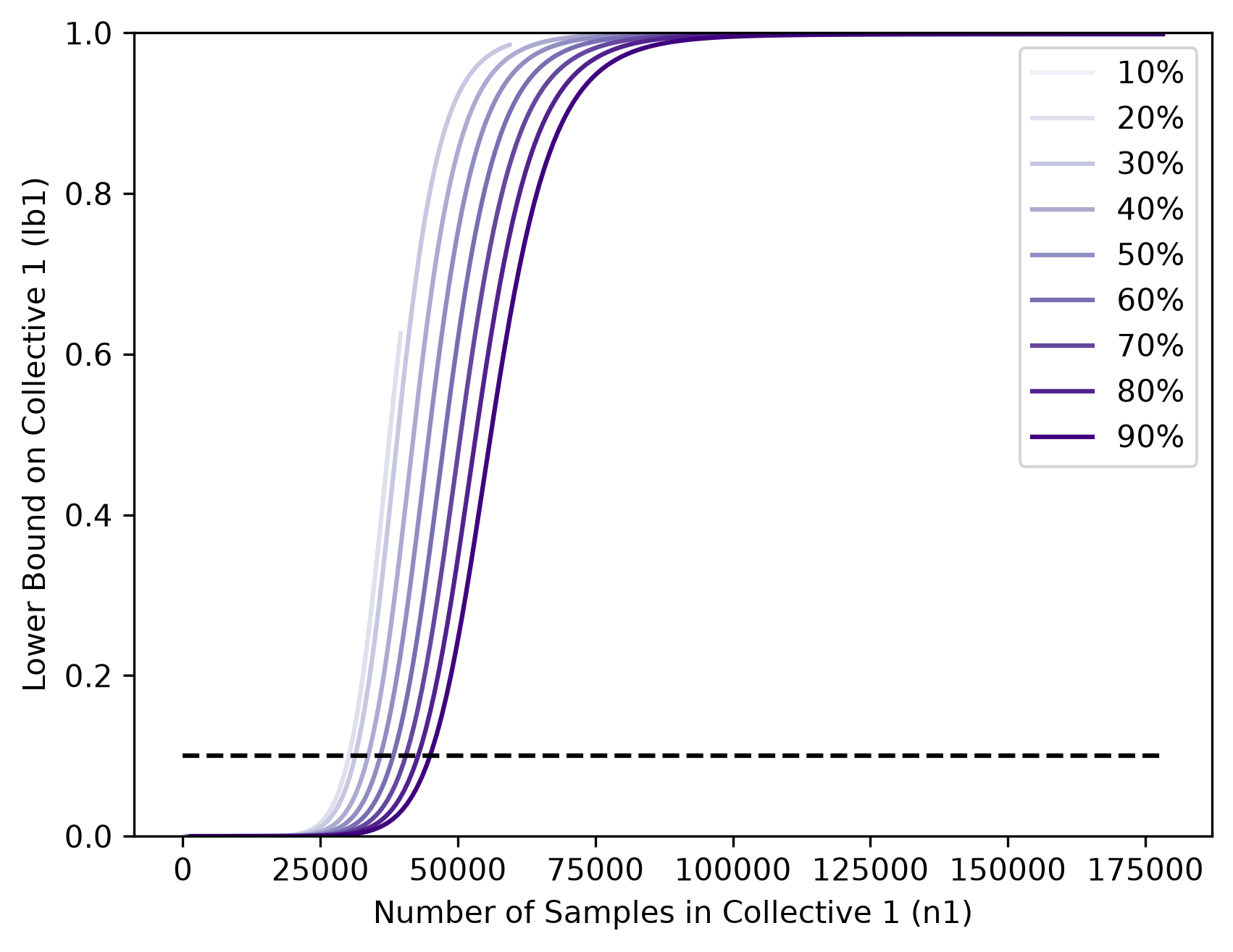}
        \caption{Collective 1: 17 vs 15}
    \end{subfigure}
    \centering
    \begin{subfigure}[]{0.3\linewidth}
        \includegraphics[width=\linewidth]{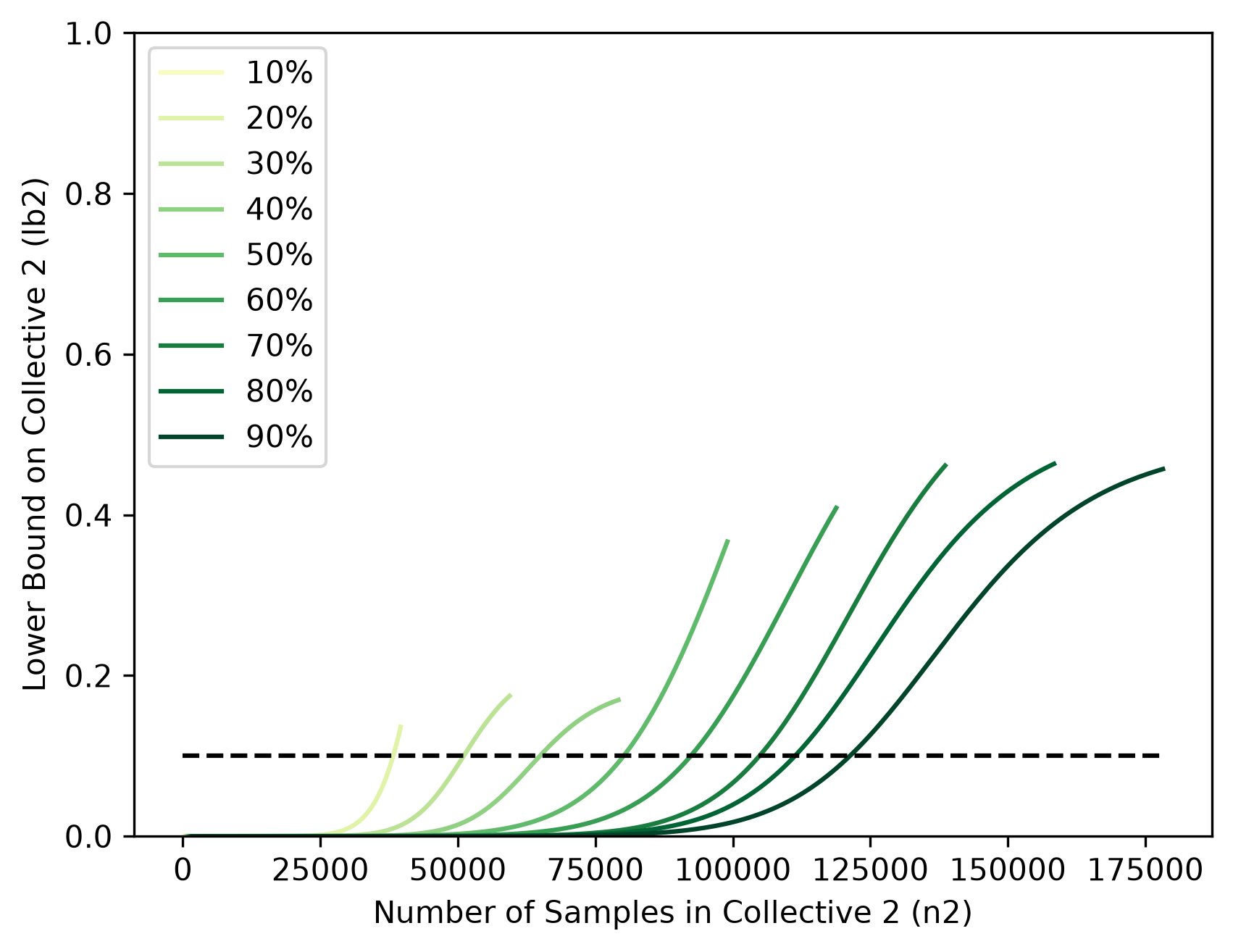}
        \caption{Collective 2: 17 vs 15}
    \end{subfigure}
    \centering
    \begin{subfigure}[]{0.3\linewidth}
\includegraphics[width=\linewidth]{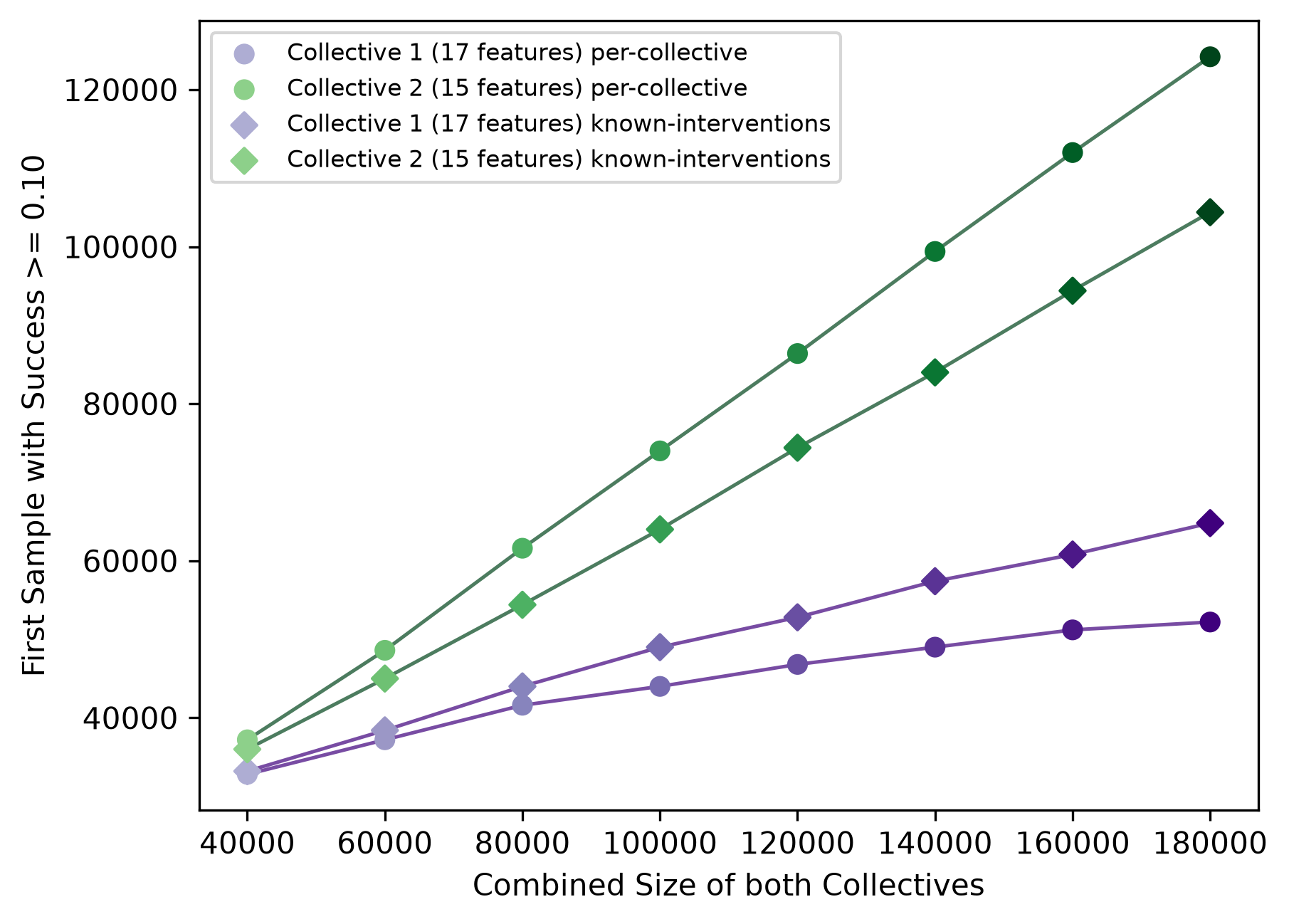}
        \caption{Combined view 17 vs 15}
    \end{subfigure}
    % 17 vs 13

    \begin{subfigure}[]{0.3\linewidth}
        \includegraphics[width=\linewidth]{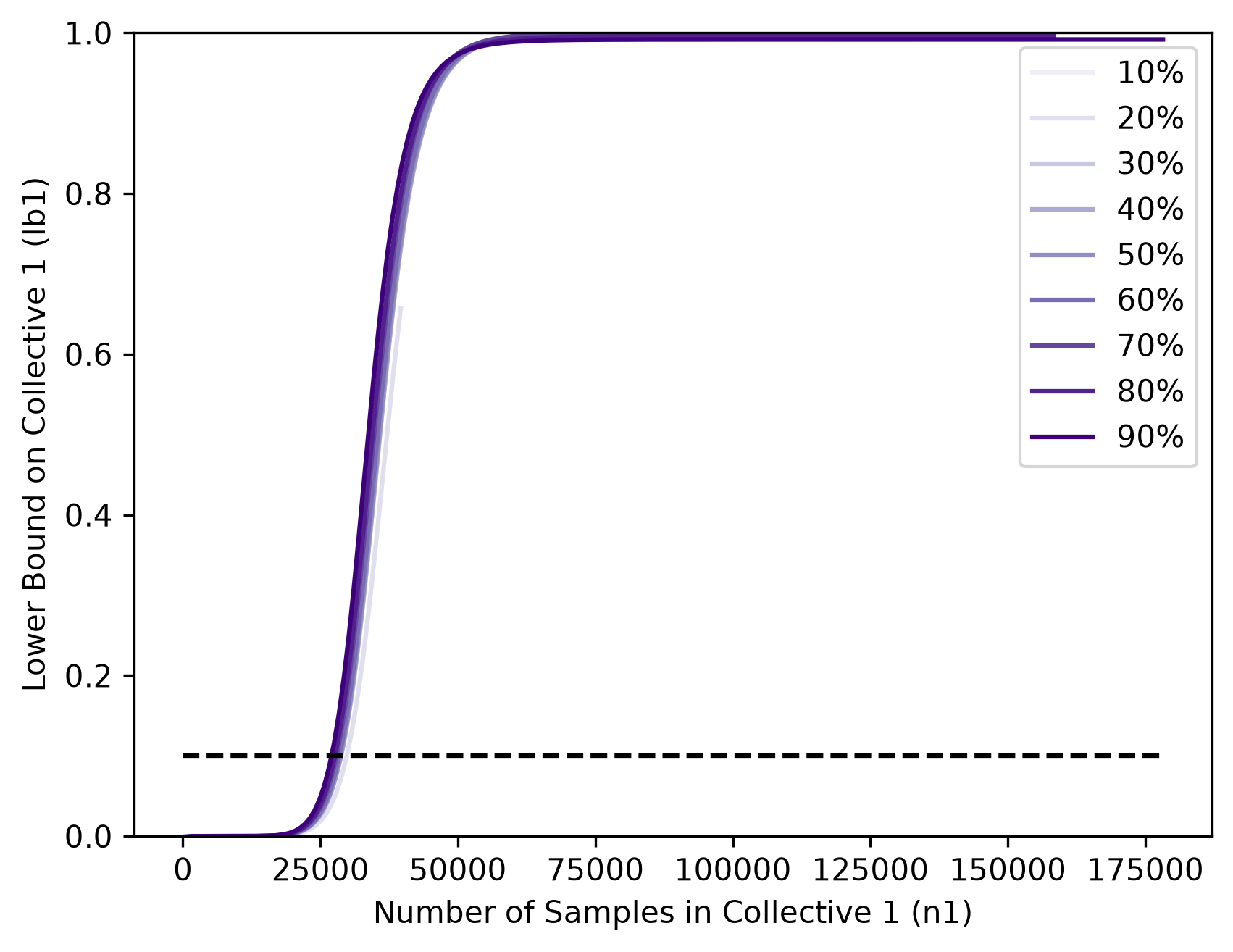}
        \caption{Collective 1: 17 vs 13}
    \end{subfigure}
    \centering
    \begin{subfigure}[]{0.3\linewidth}
        \includegraphics[width=\linewidth]{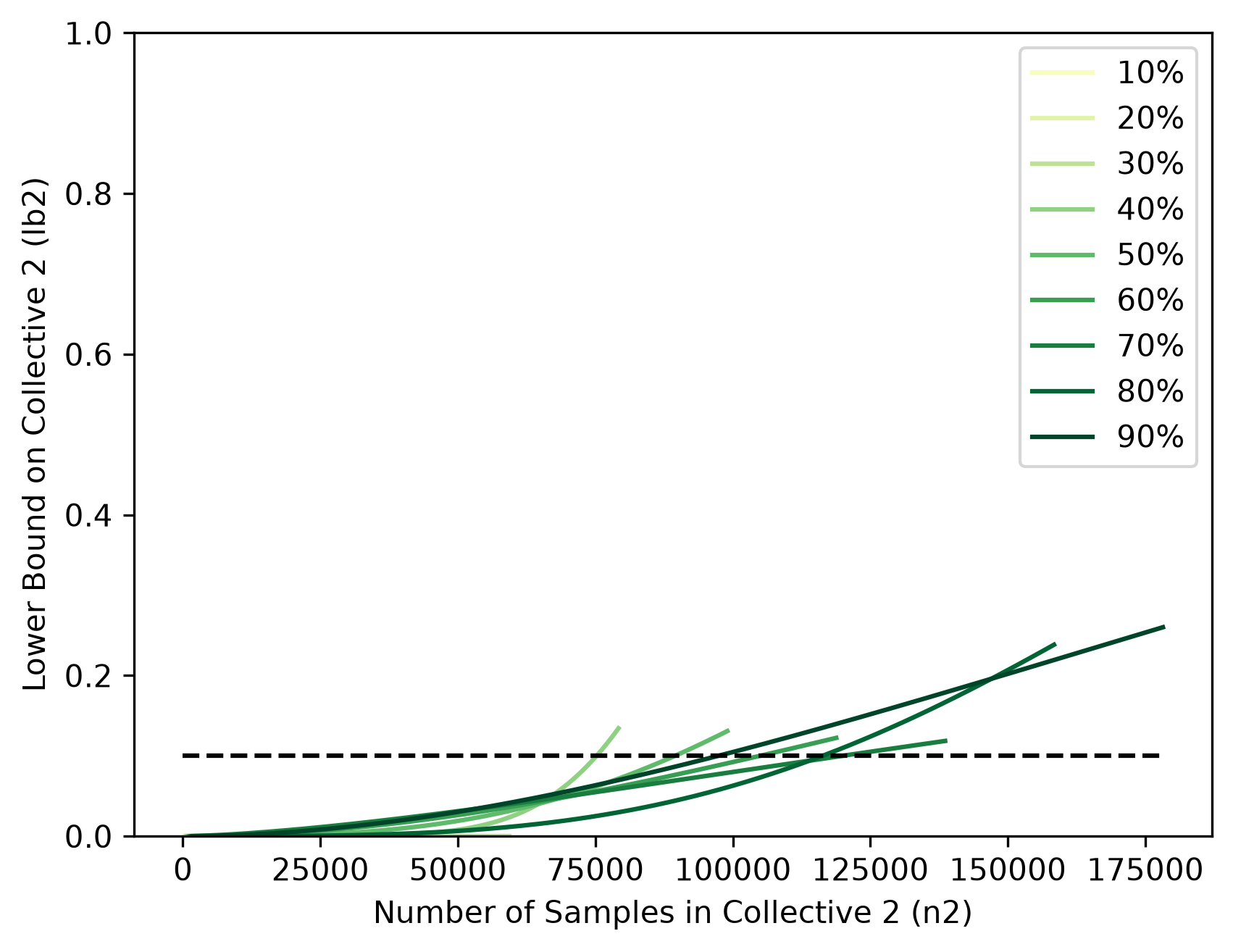}
        \caption{Collective 2: 17 vs 13}
    \end{subfigure}
    \centering
    \begin{subfigure}[]{0.3\linewidth}
        \includegraphics[width=\linewidth]{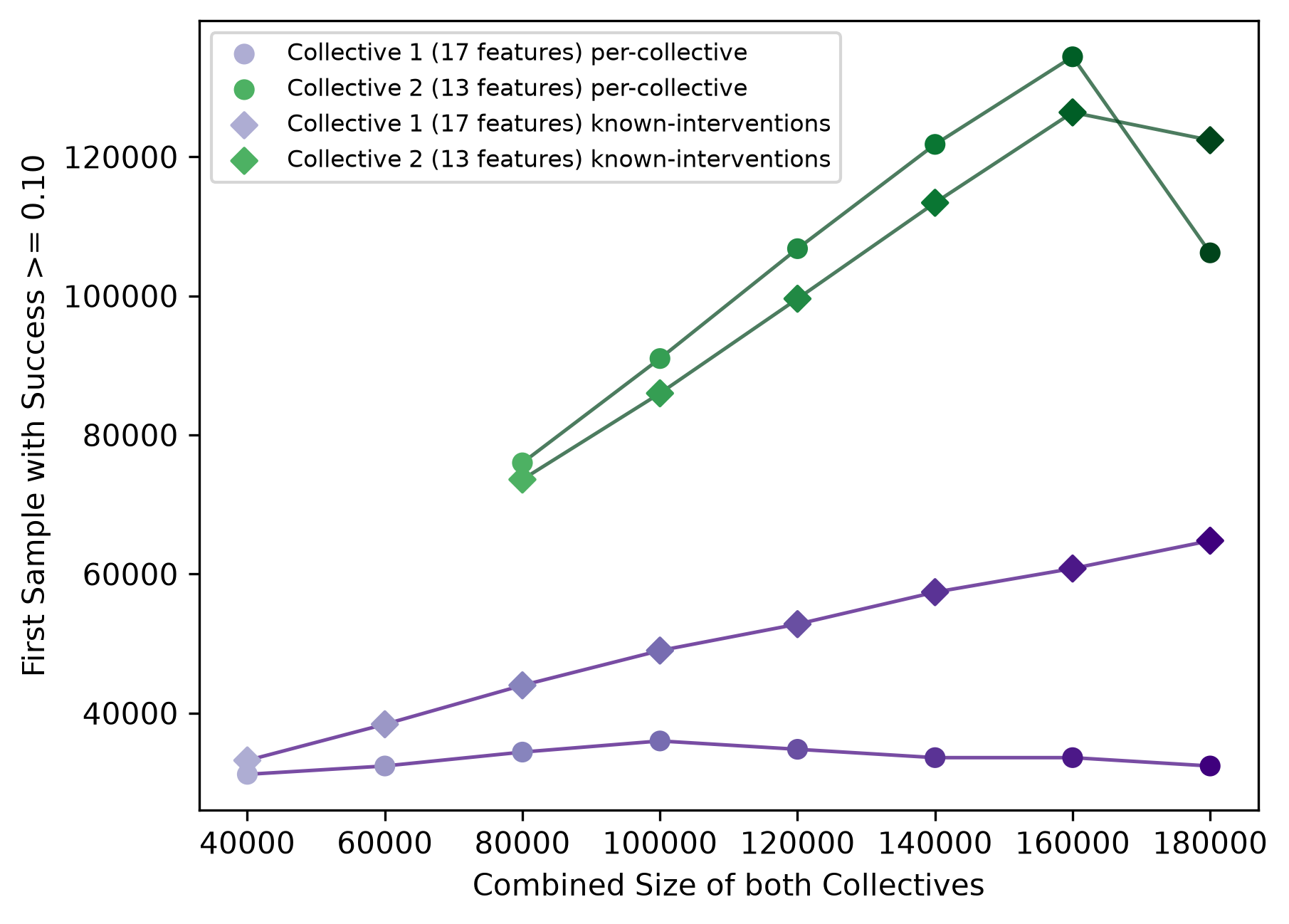}
        \caption{Combined view 17 vs 13}
    \end{subfigure}
    
    \caption{\cameraready{Lower bound on success for two feature-label collectives under varying total collective mass, mass split and feature overlap ($N=2{,}000{,}000$, $N^{\mathrm{test}}=100{,}000$). (a), (b), (d), (e): per-collective setting, one curve per total mass. (c), (f): smallest collective size reaching a lower bound of $0.1$, per-collective and known-interventions settings.}}
    \label{fig:feature_analysis}
\end{figure*}

\vspace{-0.3cm}\section{Experiments}\vspace{-0.2cm}
\label{sec:experiments}

Our experiments build upon the protocol of \citep{gauthier2025statistical} and a motivating example in climate adaptation for smart cities (detailed in Appendix~\ref{sec:aca_climate_adaptation}). We simulate a platform collecting neighborhood-level adaptation records, each comprising 18 features (e.g., adaptation strategy type, implementation timeline, funding sources) and one of four label classes: Greening \& Shade Expansion, Reflective Surface Retrofits, Cooling Services \& Outreach, and Green Stormwater Infrastructure. Details are presented in Appendix~\ref{app:simulated_adaptation_records}. We focus on the feature-label regime in the main body and defer feature-only results to Appendix~\ref{app:experiment-results}.

\textbf{Feature Overlap and Goal Alignment with Two Collectives.} We study how goal alignment affects the per-collective lower bounds for $M=2$, assuming each collective can evaluate the aggregated weighted conflict term (see Section~\ref{sec:signal_planting}). \textit{Collective~1} always plants 17 feature values, while \textit{Collective~2} plants (i) the same 17, (ii) 15, or (iii) 13 of them. Planted labels differ across collectives, so overlap directly induces conflict. We fix \cameraready{$N=2{,}000{,}000$, $N^{\mathrm{test}}=100{,}000$}, and vary \cameraready{the total collective mass $\bar\alpha_N$ from $0.1$ to $0.9$ in steps of $0.1$}, sweeping the mass ratio $\alpha_{N,1}/\alpha_{N,2}$ at each level\cameraready{ over $99$ equally spaced shares}. The \textit{per-collective} and \textit{known-interventions} settings instantiate two different levels of access to the non-local cross-collective term discussed in the computability paragraph of Section~\ref{sec:signal_planting}. In the \textit{per-collective} setting, each collective is assumed to have enough information to evaluate the relevant weighted conflict term using the other collective's actual planted strategy and empirical post-edit distribution. This corresponds to the direct evaluation of the cross-collective contribution in the bound. In the \textit{known-interventions} setting, by contrast, the collective does not know which intervention the other collective implements, but only the set of possible interventions. We therefore replace the exact cross-collective conflict by its average over this intervention set\cameraready{, which here contains the nested patterns of 13 to 17 planted features, each with every target label other than the collective's own (Appendix~\ref{app:experiment-results}). The resulting curves estimate the full-information bound rather than certify it: they are more conservative when the realized strategy conflicts less than the average intervention, and optimistic otherwise}. Success curves are smoothed via a Hill curve, and we also report the minimum sample size for 10\% success to ease comparisons~\cite{NEURIPS2019_eb1e7832}. Figure~\ref{fig:feature_analysis} reports the 17-vs.-15, (a) and (b), and 17-vs.-13, (d) and (e), regimes. The symmetric 17-vs.-17 case is in Appendix~\ref{app:experiment-results}. A clear trade-off frontier emerges: shifting mass toward one collective raises its bound at the other's expense, while increasing total mass shifts the frontier outward. As overlap decreases, the frontier skews: \textit{Collective~2} loses signal prevalence from fewer planted features, while \textit{Collective~1} benefits from uncontested features that weaken cross-collective conflict. The symmetric case confirms that, absent this asymmetry, performance is driven primarily by the mass split. Finally, from Figure~\ref{fig:feature_analysis} \cameraready{(f), when \textit{Collective~1} plants many more features than \textit{Collective~2} (17 vs.\ 13), it requires a near-constant number of samples (between 31k and 36k) to obtain a 10\% lower bound on success, regardless of the combined size, in the per-collective setting}.

\textbf{Mass Concentration and Goal Alignment with Multiple Collectives.} To illustrate how global success responds to overlap beyond two collectives, we now consider $M=4$ collectives (one per class) sharing the same 18-feature planted pattern under feature-label planting. 
\begin{figure}[t]
    \centering
    \begin{minipage}[c]{0.53\linewidth}
        \centering
        \includegraphics[width=\linewidth]{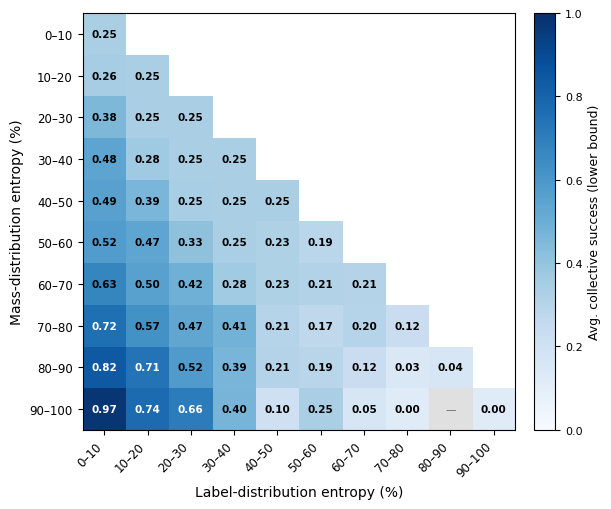}
        \caption{\cameraready{Avg. of the four per-collective success lower bounds, $\hat G_{\mathrm{avg}}$, by mass entropy and mass-weighted label entropy ($M=4$, $N=500{,}000$, $\bar\alpha_N=0.4$).}}
        \label{fig:heatmap}
    \end{minipage}\hfill
    \begin{minipage}[c]{0.39\linewidth}
        \centering\crcolor\small
        \captionof{table}{\cameraready{Trained classifiers: configurations (out of 30) in which the empirical success is at least the bound of Theorem~\ref{thm:multi_feature_label}, and configurations with a non-vacuous bound.}}\label{tab:trained}
        \begin{tabular}{@{}lccc@{}}
        \toprule
        & GBT & LR & Non-vacuous\\
        \midrule
        Synthetic & 30 & 28 & 22\\
        ACSIncome & 30 & 30 & 21\\
        \bottomrule
        \end{tabular}
    \end{minipage}
\end{figure}
We sweep \cameraready{$291{,}984$ configurations, fixing $N=500{,}000$, $N^{\mathrm{test}}=50{,}000$ and $\bar\alpha_N=0.4$, with per-collective conflict terms. We sample $1{,}985$ mass distributions so as to cover the range of mass entropy uniformly; for each of them we evaluate all $4^4=256$ label assignments when the mass entropy is below $0.5$, and a subset stratified by label entropy above}.
       Each configuration is summarized by two normalized entropies: a \emph{mass entropy} measuring how uniform the collective sizes are, and a \emph{mass-weighted label entropy} measuring how concentrated the total collective mass is on the same target label. Both are computed as the normalized Shannon entropy of the relevant distribution.
       \cameraready{We report the unweighted average of the four per-collective lower bounds, $\hat G_{\mathrm{avg}}(\{\hat S_c\}_c):=\frac{1}{M}\sum_{c=1}^M \hat S_c$, a global metric in the sense of Section~\ref{sec:MACA} that gives every collective the same weight regardless of its mass, averaged within $10\times10$ equal-width entropy bins} (Figure~\ref{fig:heatmap}). Missing entries correspond to bins with no sampled configurations. The main effect is driven by label alignment: \cameraready{in every row,} global success is highest when label entropy is low, indicating that concentrating effort on the same target label substantially reduces destructive interference. As label entropy increases, success degrades sharply, with the most severe drop when both entropies are high, i.e., collectives of comparable size targeting different labels strongly conflict at shared planted features, driving the global guarantee near zero. Mass concentration plays a more nuanced role: at intermediate label entropies, higher mass entropy can lift aggregate success once label overlap is partially resolved, but without label alignment, uniform mass alone is insufficient to prevent interference. \cameraready{At low label entropy the average rises from $0.25$ for one dominant collective to $0.97$ for four comparable ones: aligned collectives pool their masses (Corollary~\ref{cor:symmetric}), but the three small collectives in a concentrated split are individually too small for their own radii, and $\hat G_{\mathrm{avg}}$ counts them as failures.}

\cameraready{\textbf{Trained Classifiers.} The bounds hold for any $\varepsilon$-contamination-suboptimal classifier. To test them on classifiers that a platform would actually train, two collectives edit the training data (aligned, full-conflict and partial-conflict scenarios, five mass splits at $\bar\alpha_N=0.3$), the platform trains logistic regression (LR) or gradient-boosted trees (GBT), and we compare each collective's empirical success on fresh test data with its bound from Theorem~\ref{thm:multi_feature_label} ($\delta=0.05$, $\varepsilon=0$). We use a synthetic generator with the schema of Appendix~\ref{app:simulated_adaptation_records} ($N=2$M) and ACSIncome, built from US Census data~\cite{ding2021retiring} $N=877{,}909$, see Appendix~\ref{app:trained}. GBT never violates the bound (Table~\ref{tab:trained}), which is close to tight where the signal is well identified. LR violates it in two synthetic partial-conflict configurations, where the linear model is misspecified and the suboptimality premise fails.}
\vspace{-0.5cm}\section{Conclusion}\vspace{-0.4cm}\label{sec:conclusion}
In this work, we introduced the first statistical framework for Algorithmic Collective Action with multiple collectives. We focused on collective action in classification, studying how multiple collectives can bias a classifier to learn an association between an altered version of their features and a chosen, possibly overlapping, set of target classes. We provided quantitative results showcasing the interesting interplay of collectives' sizes and their alignment of values. Finally, we discussed a potential use case for our framework in the space of interventions for climate adaptation, and validated the bounds numerically, exposing multiple tradeoffs. \cameraready{As a final remark, our approach is not naively equivalent to the single-collective approach of \cite{gauthier2025statistical} instantiated $M$ times. In particular, our results rest on phenomena that only exist for $M\ge2$. First, the fact that cross-collective conflict terms are sign-indefinite, so every other collective is an ally, neutral, or an adversary, which yields the pooling, dilution, and destructive-interference regimes of Corollary~\ref{cor:symmetric}. Second, folding the other collectives into a worst-case background recovers only the lowest tier of our hierarchy, which our bounds strictly dominate. For instance, it cannot certify aligned collectives that are individually too small but jointly large enough. Third, the finite-sample analysis must decompose the empirical mixture across $M$ random disjoint sample collectives, and the aggregation of per-collective successes into global success is a design dimension with no single-collective analogue.} Our framework opens several research directions. It would be interesting to study mixed-objective settings in which some collectives plant signals while others unplant them. One of the limitations of our framework is the initial assumption that collectives are disjoint\cameraready{, and relaxing it is the extension we consider most valuable}. Similar to \cite{gauthier2025statistical}, we also assume that the feature space $\mathcal{X}$ is finite\cameraready{. With continuous or high-dimensional features, the union bound over $\mathcal X^*_c\times\mathcal Y$ must be replaced by uniform convergence over the class of events induced by $g_c$ (e.g., via VC or covering-number arguments), and the conflict terms become functionals that require regularity of $P_0$}. Finally, our framework has an inherently collaborative motivation, from which the described global success metrics stem, but adversarial/game-theoretic systems with multiple antagonist groups  can be of interest.

\cameraready{\textbf{Societal Impact.} Coordinated data editing is dual use: the mechanisms that let grassroots collectives steer a classifier can also serve harmful coordination, and the same bounds inform a platform's defensive analysis. The bounds also make power asymmetries quantitative. A collective's mass multiplies the conflict it inflicts, so a small collective's guarantee vanishes once the weighted conflict it faces approaches its own signal mass. The egalitarian metric is designed to surface the worst-served collective, and Figure~\ref{fig:heatmap} shows the failure mode of comparable masses with misaligned labels.}

% Page-budget probe (invisible): records the page on which the main text ends; see make budget.
\par\label{cr:mainend}\message{^^JCRBUDGET page=\thepage\space total=\the\pagetotal\space goal=\the\pagegoal^^J}

\begingroup\crcolor
\begin{ack}
CB and FD acknowledge support from the National Institutes of Health Grant 1R01ES037156-01. DR is financed by the project “Building Energy Systems on causal reasoning (BOSS)“, funded within the “Technologies and Innovations for the Climate-Neutral City” (TIKS) Programme of the Austrian Research Promotion Agency (FFG). OA is supported by the Social Sciences and Humanities Research Council of Canada (SSHRC) through a Canada Impact+ Research Training Award [award no. 712-2026-1003]. BN declares indirect competing interests in the past 36 months through part-time employment at Modular Insurance Solutions Inc., a course development honorarium from the Harvard Data Science Initiative, and subsequent funding from the UBC Faculty of Medicine AI in Health Network. CB, DR, OA, PG and FD declare no competing interests.
\end{ack}
\endgroup

\bibliographystyle{plainnat}
\renewcommand{\bibsection}{\subsubsection*{References}}
\bibliography{ref}

@inproceedings{NEURIPS2019_eb1e7832,
 author = {Raghu, Maithra and Zhang, Chiyuan and Kleinberg, Jon and Bengio, Samy},
 booktitle = {Advances in Neural Information Processing Systems},
 editor = {H. Wallach and H. Larochelle and A. Beygelzimer and F. d\textquotesingle Alch\'{e}-Buc and E. Fox and R. Garnett},
 pages = {},
 publisher = {Curran Associates, Inc.},
 title = {Transfusion: Understanding Transfer Learning for Medical Imaging},
 url = {https://proceedings.neurips.cc/paper_files/paper/2019/file/eb1e78328c46506b46a4ac4a1e378b91-Paper.pdf},
 volume = {32},
 year = {2019}
}

@inproceedings{karan2025sync_or_sink,
  title     = {Sync or Sink: Bounds on Algorithmic Collective Action with Noise and Multiple Groups},
  author    = {Karan, Aditya and Kalle, Prabhat and Vincent, Nicholas and Sundaram, Hari},
  booktitle = {NeurIPS 2025 (Non-archival) Workshop on Algorithmic Collective Action (ACA)},
  year      = {2025},

}

@inproceedings{fedorova2025altruistic,
  title     = {Altruistic Collective Action in Recommender Systems},
  author    = {Fedorova, Ekaterina and Kitch, Madeline Celi and Podimata, Chara},
  booktitle = {NeurIPS 2025 (Non-archival) Workshop on Algorithmic Collective Action (ACA)},
  year      = {2025},
}

@inproceedings{lewandowska2025workers,
  title     = {Workers vs. The Algorithm: Simulating Collective Action in Gig-Economy Platforms},
  author    = {Lewandowska, Kristina},
  booktitle = {NeurIPS 2025 (Non-archival) Workshop on Algorithmic Collective Action (ACA)},
  year      = {2025},
}

@inproceedings{hussain2025empowering,
  title     = {Empowering Users Together: Connecting Algorithmic Collective Action and Explainable AI},
  author    = {Hussain, Ayana and Thacker, Cole Michael and Zhao, Patrick and Karan, Aditya and Vincent, Nicholas},
  booktitle = {NeurIPS 2025 (Non-archival) Workshop on Algorithmic Collective Action (ACA)},
  year      = {2025},
}

@inproceedings{
Sagawa2020Distributionally,
title={Distributionally Robust Neural Networks},
author={Shiori Sagawa and Pang Wei Koh and Tatsunori B. Hashimoto and Percy Liang},
booktitle={International Conference on Learning Representations},
year={2020},
url={https://openreview.net/forum?id=ryxGuJrFvS}
}

@article{ai2024artificial,
  title={Artificial intelligence risk management framework: Generative artificial intelligence profile},
  author={NIST},
  journal={NIST Trustworthy and Responsible AI Gaithersburg, MD, USA},
  year={2024}
}

@article{perrigo2023openaiKenyaWorkers,
  author  = {Perrigo, Billy},
  title   = {Exclusive: OpenAI Used Kenyan Workers on Less Than \$2 Per Hour to Make ChatGPT Less Toxic},
  journal = {TIME},
  year    = {2023},
  month   = jan,
  note    = {Published Jan 18, 2023. Accessed Feb 24, 2026.}
}

@inproceedings{nakka2025pii,
  title={PII-Scope: A Comprehensive Study on Training Data Privacy Leakage in Pretrained LLMs},
  author={Nakka, Krishna Kanth and Frikha, Ahmed and Mendes, Ricardo and Jiang, Xue and Zhou, Xuebing},
  booktitle={Proceedings of the 14th International Joint Conference on Natural Language Processing and the 4th Conference of the Asia-Pacific Chapter of the Association for Computational Linguistics},
  pages={3731--3765},
  year={2025}
}

@misc{Okunyte2026AndroidAIAppLeak,
  author       = {Okunyt{\.e}, Paulina},
  title        = {Android AI app exposes nearly 2M user images and videos: anyone can watch your videos},
  howpublished = {Cybernews},
  year         = {2026},
  month        = feb,
  note         = {Published: 19 Feb 2026; last updated: 20 Feb 2026. Accessed: 24 Feb 2026},
}

@inproceedings{
battiloro2025algorithmic,
title={Algorithmic Collective Action with Multiple Collectives},
author={Claudio Battiloro and Pietro Greiner and Bret Nestor and Oumaima Amezgar and Francesca Dominici},
booktitle={NeurIPS 2025 (Non-archival) Workshop on Algorithmic Collective Action},
year={2025},
}

@inproceedings{bendov2025minority,
  title     = {Fairness for the People, by the People: Minority Collective Action},
  author    = {Ben-Dov, Omri and Samadi, Samira and Sanyal, Amartya and Tifrea, Alexandru},
  booktitle = {NeurIPS 2025 (Non-archival) Workshop on Algorithmic Collective Action (ACA)},
  year      = {2025},

}

@inproceedings{HardtEtAl2023,
  author    = {Moritz Hardt and Eric Mazumdar and Celestine Mendler-D{\"u}nner and Tijana Zrnic},
  title     = {Algorithmic collective action in machine learning},
  booktitle = {International Conference on Machine Learning (ICML)},
  pages     = {12570--12586},
  publisher = {PMLR},
    volume = 	 {202},
  year      = {2023}
}

@inproceedings{VincentEtAl2020,
  title={Data leverage: A framework for empowering the public in its relationship with technology companies},
  author={Vincent, Nicholas and Li, Hanlin and Tilly, Nicole and Chancellor, Stevie and Hecht, Brent},
  booktitle={Proceedings of the 2021 ACM Conference on Fairness, Accountability, and Transparency},
  pages={215--227},
  year={2021}
}

@inproceedings{SiggEtAl2024,
author = {Sigg, Dorothee and Hardt, Moritz and Mendler-D\"{u}nner, Celestine},
title = {Decline Now: A Combinatorial Model for Algorithmic Collective Action},
year = {2025},
publisher = {Association for Computing Machinery},
doi = {10.1145/3706598.3713966},
booktitle = {Proceedings of the 2025 CHI Conference on Human Factors in Computing Systems},
articleno = {912},
numpages = {17},
location = {
},
series = {CHI '25}
}

@book{Latour2005Reassembling,
  author    = {Bruno Latour},
  title     = {Reassembling the Social: An Introduction to Actor-Network-Theory},
  year      = {2005},
  publisher = {Oxford University Press},
  address   = {Oxford},
  isbn      = {9780199256044}
}

@incollection{Haraway1991CyborgManifesto,
  title={A cyborg manifesto: Science, technology, and socialist-feminism in the late twentieth century},
  author={Haraway, Donna},
  booktitle={The transgender studies reader},
  pages={103--118},
  year={2013},
  publisher={Routledge}
}

@misc{capastrategies_heatwatch,
  author       = {{CAPA Strategies, LLC}},
  title        = {Heat Mapping -- Heat Watch},
  year         = {2025},
}

@misc{noaa_heatwatch,
  author       = {{National Oceanic and Atmospheric Administration (NOAA)}},
  title        = {Mapping Campaigns},
  howpublished = {HEAT.gov},
  year         = {2025},
}

@misc{rotterdam_digital_twin_vision,
  author       = {{City of Rotterdam, CIO Office}},
  title        = {Rotterdam in Transformation: Vision on the Digital City 1.0},
  year         = {2024},
}

@inproceedings{devrio2024building,
  title={Building, shifting, \& employing power: A taxonomy of responses from below to algorithmic harm},
  author={DeVrio, Alicia and Eslami, Motahhare and Holstein, Kenneth},
  booktitle={Proceedings of the 2024 ACM Conference on Fairness, Accountability, and Transparency},
  pages={1093--1106},
  year={2024}
}

@misc{California2020CPRA,
  author = {{State of California}},
  title  = {California Privacy Rights Act of 2020 (CPRA)},
  year   = {2020},
  note   = {Proposition 24, approved November 3, 2020}
}

@inproceedings{utz2019informed,
  title={(Un) informed consent: Studying GDPR consent notices in the field},
  author={Utz, Christine and Degeling, Martin and Fahl, Sascha and Schaub, Florian and Holz, Thorsten},
  booktitle={Proceedings of the 2019 acm sigsac conference on computer and communications security},
  pages={973--990},
  year={2019}
}

@inproceedings{selbst2019fairness,
  title={Fairness and abstraction in sociotechnical systems},
  author={Selbst, Andrew D and Boyd, Danah and Friedler, Sorelle A and Venkatasubramanian, Suresh and Vertesi, Janet},
  booktitle={Proceedings of the conference on fairness, accountability, and transparency},
  pages={59--68},
  year={2019}
}

@misc{Canada2000PIPEDA,
  author = {{Government of Canada}},
  title  = {Personal Information Protection and Electronic Documents Act},
  year   = {2000},
  note   = {S.C. 2000, c. 5}
}

@article{EU2016GDPR,
  author  = {{European Parliament and Council of the European Union}},
  title   = {Regulation (EU) 2016/679 of the European Parliament and of the Council of 27 April 2016 on the protection of natural persons with regard to the processing of personal data and on the free movement of such data, and repealing Directive 95/46/EC (General Data Protection Regulation)},
  journal = {Official Journal of the European Union},
  number  = {L 119},
  pages   = {1--88},
  year    = {2016},
  month   = may,
}

@book{barocas2023fairness,
  title={Fairness and machine learning: Limitations and opportunities},
  author={Barocas, Solon and Hardt, Moritz and Narayanan, Arvind},
  year={2023},
  publisher={MIT press}
}

@inproceedings{karan2025algorithmic,
  title={Algorithmic Collective Action with Two Collectives},
  author={Karan, Aditya and Vincent, Nicholas and Karahalios, Karrie and Sundaram, Hari},
  booktitle={Proceedings of the 2025 ACM Conference on Fairness, Accountability, and Transparency},
  pages={1468--1483},
  year={2025}
}

@inproceedings{
gauthier2025statistical,
title={Statistical Collusion by Collectives on Learning Platforms},
author={Etienne Gauthier and Francis Bach and Michael I. Jordan},
booktitle={Forty-second International Conference on Machine Learning (ICML)},
year={2025}
}

@article{baumann2024algorithmic,
  title={Algorithmic collective action in recommender systems: promoting songs by reordering playlists},
  author={Baumann, Joachim and Mendler-D{\"u}nner, Celestine},
  journal={Advances in Neural Information Processing Systems},
  volume={37},
  pages={119123--119149},
  year={2024}
}

@inproceedings{
solanki2025crowding,
title={Crowding Out The Noise: Algorithmic Collective Action Under Differential Privacy},
author={Rushabh Solanki and Meghana Bhange and Ulrich A{\"\i}vodji and Elliot Creager},
booktitle={NeurIPS 2025 (Non-archival) Workshop on Algorithmic Collective Action},
year={2025},
}

@inproceedings{BenDovEtAl2024,
  author    = {Omri Ben{-}Dov and Jake Fawkes and Samira Samadi and Amartya Sanyal},
  title     = {The role of learning algorithms in collective action},
  booktitle = {International Conference on Machine Learning (ICML)},
  pages     = {3443--3461},
  publisher = {PMLR},
  year      = {2024},
  volume = 	 {235},
}

@inproceedings{ding2021retiring,
  title     = {Retiring Adult: New Datasets for Fair Machine Learning},
  author    = {Ding, Frances and Hardt, Moritz and Miller, John and Schmidt, Ludwig},
  booktitle = {Advances in Neural Information Processing Systems},
  volume    = {34},
  pages     = {6478--6490},
  year      = {2021}
}

@article{pedregosa2011scikit,
  title   = {Scikit-learn: Machine Learning in {P}ython},
  author  = {Pedregosa, Fabian and Varoquaux, Ga{\"e}l and Gramfort, Alexandre and Michel, Vincent and Thirion, Bertrand and Grisel, Olivier and Blondel, Mathieu and Prettenhofer, Peter and Weiss, Ron and Dubourg, Vincent and Vanderplas, Jake and Passos, Alexandre and Cournapeau, David and Brucher, Matthieu and Perrot, Matthieu and Duchesnay, {\'E}douard},
  journal = {Journal of Machine Learning Research},
  volume  = {12},
  pages   = {2825--2830},
  year    = {2011}
}
\appendix
\onecolumn

\section{\cameraready{Formal Setting and }Proofs}
\subsection{\cameraready{Formal Setting}}\label{app:formal}
\begin{crblock}
This section collects the measure-theoretic details of the setting of Section~\ref{sec:MACA}. All probability measures on $\mathcal Z$ are defined on the discrete $\sigma$-algebra $2^{\mathcal Z}$ (the power set of $\mathcal Z$, of cardinality $2^{|\mathcal Z|}$). The population is a probability space $(\Omega,2^{\Omega},\pi)$, with $\Omega$ finite, and $f_0:\Omega\to\mathcal{Z}$ is a measurable map, so that $P_0=(f_0)_{\#}\pi$ is a push-forward measure; each strategy $h_c:\mathcal Z\to\mathcal Z$ is measurable. If strategies are randomized, replace $h_c$ by a Markov kernel $K_c$ and $(h_c)_{\#}$ by the kernel pushforward $K_c^{\#}$; we use deterministic $h_c$ for notational simplicity.

\textbf{Sampling model.} Sampling individuals i.i.d.\ from $\pi$ is general enough to cover different settings, e.g., survey-like sampling, by interpreting $\pi$ as the survey's sampling design over individuals (e.g., uniform, stratified, or weighted), with $f_0$ mapping each sampled respondent to their recorded $(x,y)$, and internet/firm sampling, by interpreting $\pi$ as the platform's exposure/participation process, so that i.i.d.\ draws from $\pi$ represent repeated user interactions or users selected by the platform's logging pipeline. In both cases, the observed training data are the $(x,y)$ pairs induced by sampling users via $\pi$ and then recording their associated outcome through $f_0$.

\textbf{Empirical measures.} For $n_{N,c}>0$, the empirical \emph{pre-edit} distribution of collective $\Omega_c$ is $\hat P_{N,c}^{\mathrm{pre}}:=\frac{1}{n_{N,c}}\sum_{i\in I_{N,c}} \delta_{Z_i}$, with feature marginal $\hat P_{N,c}^{\mathrm{pre},X}(B):=\frac{1}{n_{N,c}}\sum_{i\in I_{N,c}} \mathbf 1_{\{X_i\in B\}}$, $B\subseteq\mathcal X$. The \emph{post-edit} empirical measure of collective $\Omega_c$ is $\hat P_{N,c}:=(h_c)_{\#}\hat P_{N,c}^{\mathrm{pre}}$, with $\hat P_{N,c}^X$ denoting its marginal on $\mathcal X$. For the non-collective users, if $n_{N,0}>0$, define $\hat P_{N,0}:=\frac{1}{n_{N,0}}\sum_{i:S_i\notin\cup_{c=1}^M\Omega_c}\delta_{Z_i}$. If $n_{N,0}=0$, we set $\hat P_{N,0}$ equal to an arbitrary fixed probability measure on $\mathcal Z$; this convention has no effect on $\hat P_N$ in \eqref{eq:emp_mixture}, because then $\alpha_{N,0}=0$.

\textbf{Contamination-suboptimality.} The contamination-suboptimality condition above makes explicit the perturbation
model underlying the margin argument used in
\cite{HardtEtAl2023,gauthier2025statistical}. Namely, the classifier is assumed to be Bayes-optimal, or approximately optimal through the choice of $\varepsilon$, for a distribution $\hat Q_N$ obtained by an $\varepsilon$-contamination of the empirical training distribution $\hat P_N$. This is slightly stronger and more structured than merely requiring
$\TV(\hat Q_N,\hat P_N)\le \varepsilon$. In fact, under symmetric total variation alone,
the threshold $\varepsilon/(1-\varepsilon)$ is not generally implied, and the standard conservative margin threshold is $2\varepsilon$. Thus, the contamination formulation is the precise source of the $\varepsilon/(1-\varepsilon)$ terms in our bounds. We use this formulation in
order to stay close to the perturbation model and constants used in \cite{HardtEtAl2023,gauthier2025statistical}, while making explicit the additional structure needed for the $\varepsilon/(1-\varepsilon)$ threshold.
All results below remain valid under the looser symmetric total-variation suboptimality model, provided every occurrence of
$\varepsilon/(1-\varepsilon)$ in the margin conditions and final bounds is replaced by the conservative threshold $2\varepsilon$.

As a sanity check, the following proposition shows that the empirical distribution in \eqref{eq:emp_mixture} converges to the distribution in \eqref{eq:mixture} in the infinite-sample regime.
\begin{proposition}\label{prop:finite_to_infinite}
Under the neutrality assumption $P_0(z)=P_0\bigl(z\mid \omega\in\Omega_c\bigr)$, it holds $
\hat P_N \rightarrow P$ almost surely. In addition, for each collective $\Omega_c$, it holds  $\alpha_{N,c}\to \alpha_c$ and
$\hat P_{N,c} \rightarrow P_c$ almost surely.
\end{proposition}
\begin{proof}
Fix $z\in\mathcal{Z}$. For $c=1,\dots,M$ define
\[
Y_i^{(c,z)}\ :=\ \mathbf 1_{\{\,S_i\in\Omega_c,\ h_c(Z_i)=z\,\}},
\qquad
Y_i^{(0,z)}\ :=\ \mathbf 1_{\{\,S_i\notin \cup_{c=1}^M\Omega_c,\ Z_i=z\,\}} .
\]
Since $(S_i)_{i\ge1}$ are i.i.d.\ with law $\pi$ and $Z_i=f_0(S_i)$, the vectors $\bigl(Y_i^{(0,z)},Y_i^{(1,z)},\dots,Y_i^{(M,z)}\bigr)$ are i.i.d.
By neutrality, $Z_i\mid S_i\in\Omega_c\sim P_0$ for each $c$, hence
\[
\mathbb E\,Y_1^{(0,z)}=(1-\bar\alpha)\,P_0(z),
\qquad
\mathbb E\,Y_1^{(c,z)}=\alpha_c\,P_c(z)\ \ (c=1,\dots,M),
\]
where $P_c=(h_c)_{\#}P_0$. By the strong law of large numbers,
\[
\frac{1}{N}\sum_{i=1}^N Y_i^{(0,z)} \xrightarrow[\text{a.s.}]{} (1-\bar\alpha)\,P_0(z),
\qquad
\frac{1}{N}\sum_{i=1}^N Y_i^{(c,z)} \xrightarrow[\text{a.s.}]{} \alpha_c\,P_c(z).
\]
Noting that
\[
\hat P_{N}(z) \;=\; \bigl(1-\bar\alpha_N\bigr)\,\hat P_{N,0}(z)\;+\;\sum_{c:\alpha_{N,c}>0}\alpha_{N,c}\,\hat P_{N,c}(z)
\;=\; \frac{1}{N}\sum_{i=1}^N\Bigl(Y_i^{(0,z)}+\sum_{c=1}^M Y_i^{(c,z)}\Bigr),
\]
we obtain, for each $z\in\mathcal{Z}$,
\[
\hat P_{N}(z) \xrightarrow[\text{a.s.}]{} (1-\bar\alpha)\,P_0(z)+\sum_{c=1}^M \alpha_c\,P_c(z)
\;=\; P(z).
\]
Since $\mathcal{Z}$ is finite, coordinatewise convergence implies convergence of measures, i.e.,
$\hat P_N\to P$ almost surely. 
Similarly, $\alpha_{N,c}=N^{-1}\sum_{i=1}^N \mathbf 1_{\{S_i\in\Omega_c\}}\to\alpha_c$ almost surely again by the strong law of large numbers, and, for $\alpha_c>0$, the convergence $\alpha_{N,c}\hat P_{N,c}(z)\to \alpha_c P_c(z)$ for all $z$ together with $\alpha_{N,c}\to\alpha_c$ yields $\hat P_{N,c}\rightarrow P_c$ almost surely.
\end{proof}
\end{crblock}

\subsection{\cameraready{Proofs}}\label{app:proofs}
\printProofs

\vspace{-0.3cm}\section{A Climate Adaptation Use Case}\vspace{-0.3cm}
\label{sec:aca_climate_adaptation}
In this section, we provide a real-world use case for ACA with multiple collectives in the space of interventions for climate adaptation. Consider a national adaptation intake platform used by a city to prioritize interventions at the neighborhood level. The platform stores standardized neighborhood-level adaptation records to identify which interventions are most needed. These records combine official sources (e.g., inspections, infrastructure inventories, incident reports) with community inputs (e.g., forum discussions, neighborhood association reports), and are stored in a standardized format with a fixed set of categorical fields, so that each neighborhood is summarized in the same structured way. The model recommends one intervention per neighborhood from a finite menu $\cal Y$. The training corpus is large as it contains millions of neighborhood planning units (e.g., census tracts/block groups) pooled across many municipalities, while the decision remains neighborhood-level. Assume there are $M$ collectives, being grassroots associations. Each collective is run by residents from different neighborhoods that face similar conditions (e.g., recurrent flooding, extreme heat exposure, limited cooling access) and therefore has a nuanced, ground-level understanding of local needs. The collective seeks to coordinately change the adaptation records corresponding to its member neighborhoods (and the shared community inputs that feed into those records) to steer the classifier toward the intervention $y^*_c \in \cal Y$ it has judged to be most necessary for those neighborhoods. Each collective thus deploys an editing strategy $h_c$ to achieve its actionable goal. For example, in a feature-label regime, a collective with residents from waterfront neighborhoods documents chronic flooding with household infiltration tests and evidence collected by interviewing community members, then updates the neighborhoods' records so the categorical fields consistently reflect that risk profile and explicitly requests rain gardens/bioswales in the intervention request field. Similarly, collectives would likely need to resort to feature-only strategies if the city’s intake pipeline does not allow community label edits, i.e.,\ if the intervention field is locked and labeled by staff/administrative codes. Global success would then emphasize equity (lifting the worst-served neighborhood with $\cameraready{\hat G_{\min}}$) or scale (a mass-weighted average with $\cameraready{\hat G_{\mathrm{w}}}$). Even if this example is structured in a slightly simplistic way, there is precedent for community-reported concerns to be ingested into decision-making tools. NOAA and CAPA-led ``Heat Watch'' campaigns have produced heat maps in \(\,120+\) communities and feed them into planning and public-health practice \cite{noaa_heatwatch,capastrategies_heatwatch}. European cities are deploying digital twins for climate-resilience analysis and participatory planning, e.g., \cite{rotterdam_digital_twin_vision}.

\cameraready{\textbf{Information Regimes as Protocols.} The information regimes of Section~\ref{sec:signal_planting} correspond to concrete protocols in this use case. (i) In the best case, the platform or a regulator acts as a trusted aggregator and publishes the aggregate conflict $\sum_{d\neq c}\alpha_{N,d}\Delta_{c\to d}(x^*)$ at the planted feature values, in the spirit of transparency reporting, without revealing any per-collective distribution. (ii) Without a trusted party, the aggregate conflict is a sum of bounded per-collective empirical averages, so the collectives can compute it jointly by secure aggregation, optionally adding differentially private noise; since the bound already absorbs additive perturbations of the conflict term through its radii, we expect calibrated noise to enter as one more additive term of the same form, degrading the guarantee gracefully rather than invalidating it. (iii) Aligned collectives need only exchange their planted sets and target labels, not their data, to determine the sign of $\Delta_{c\to d}$ and whether mass pooling applies (Corollary~\ref{cor:symmetric}); in the example above, two waterfront associations requesting rain gardens for overlapping record profiles can establish that they are allies by comparing the fields they plan to edit.}
\vspace{-0.3cm}
\section{Simulated Neighborhood Adaptation-Records Dataset}\vspace{-0.2cm}
\label{app:simulated_adaptation_records}

We build on the codebase of \cite{gauthier2025statistical} and  generate a dataset of $\mathbf{3{,}000{,}000}$ instances. Each instance represents a neighborhood-level adaptation record for a \emph{planning unit} (e.g., a census tract/block group), characterized by multiple categorical features that capture key aspects of local conditions, planned measures, governance, and operational capacity. Conceptually, the base dataset is generated by drawing instances \emph{i.i.d.} from a population distribution over planning units.

The features included in this synthetic dataset are as follows:
\begin{itemize}[leftmargin=*]
    \item \text{Adaptation Strategy Type:} the type of adaptation strategy, categorized as \textit{Infrastructure Upgrade}, \textit{Nature-Based Solution}, \textit{Policy/Regulation}, \textit{Community Outreach}, \textit{Emergency Preparedness}, \textit{Technology Deployment}, \textit{Land-Use Planning}, or \textit{Public Health Services}.
    \item \text{Primary Energy Source for Adaptation Measures:} the primary energy source used by the deployed measures, which can be \textit{Renewable (Solar)}, \textit{Renewable (Wind)}, \textit{Grid Electricity}, or \textit{None}.
    \item \text{Implementation Timeline:} the timeline for implementation, classified as \textit{Short-term (<1 yr)}, \textit{Medium-term (1--5 yr)}, or \textit{Long-term (>5 yr)}.
    \item \text{Geographic Scope:} the geographic scope of the measure, identified as \textit{Neighborhood}, \textit{City-wide}, or \textit{Regional}.
    \item \text{Resilience Rating:} the resilience rating of the neighborhood profile, rated as \textit{Very Low}, \textit{Low}, \textit{Moderate}, \textit{High}, or \textit{Very High}.
    \item \text{Construction Material for Adaptation Projects:} the primary construction material, which can be \textit{Concrete}, \textit{Steel}, or \textit{Timber}.
    \item \text{Public Awareness / Communication Platform:} the level of the communication platform, ranging from \textit{None}, \textit{Basic Alerts (SMS/email)}, \textit{Interactive Dashboard}, or \textit{Integrated Smart-city Platform}.
    \item \text{Funding Source / Jurisdiction:} the source/jurisdiction of funding, that we will denote by \textit{F1}, \textit{F2}, \textit{F3}, \textit{F4}, and \textit{F5}.
    \item \text{Project Maintenance Horizon:} the maintenance horizon, available in options of \textit{1 year}, \textit{3 years}, \textit{5 years}, or \textit{10 years}.
    \item \text{Number of Access Points to Critical Infrastructure:} the number of access points, which can be \textit{0 (sealed)}, \textit{1--2}, or \textit{3--4}.
    \item \text{Capacity of Adaptation Facility:} the capacity of an adaptation facility, with options for \textit{<50}, \textit{50--200}, \textit{200--500}, or \textit{500+}.
    \item \text{Cooling Provision:} indicates the type of cooling provision (\textit{None} or \textit{Mechanical (AC)}).
    \item \text{Decision-Support / Early-Warning System:} the level of decision-support, which can be \textit{None}, \textit{Basic (threshold alerts)}, or \textit{Advanced (predictive modeling)}.
    \item \text{Infrastructure Surface Type:} the surface type, categorized as \textit{Standard Pavement}, \textit{Permeable Pavement}, or \textit{Heat-Reflective Coating}.
    \item \text{Green Roof / Solar Installation:} indicates whether a roof retrofit is present (\textit{None} or \textit{Retrofit Present}).
    \item \text{Public Notification System:} the public notification system, which can be \textit{None}, \textit{Sirens}, \textit{Public Address (PA) System}, or \textit{Multilingual Broadcast}.
    \item \text{Automated Operation / Smart Controls:} indicates whether smart controls are present (\textit{None} or \textit{Fully automated (AI-driven)}).
    \item \text{Interoperability with External Sensors / IoT:} indicates whether IoT interoperability is present (\textit{No} or \textit{Full}).
\end{itemize}

Additionally, each adaptation record is assigned an \text{Intervention Recommendation} label based on a scoring system that considers various factors such as resilience rating, geographic scope, cooling provision, maintenance horizon, and others. The possible outcomes are classified into four categories: Greening \& Shade Expansion, Reflective Surface Retrofits, Cooling Services \& Outreach, and  Green Stormwater Infrastructure. % <-- placeholders

In our experiments, we then generate separate datasets for the collectives sampled without replacement from this base dataset. Unless otherwise specified, we choose $N=2{,}000{,}000$ for the training set and $N^{\text{test}}=100{,}000$ for the test set\cameraready{; the two-collective feature-label experiments, the four-collective heatmap, the unplanting experiments and the single-collective ablations state their own sizes}. In all the experiments, we set $\delta=0.05$ and $\varepsilon=0$. % <-- placeholders if you want to change

\vspace{-0.3cm}
\section{Reduction to the single-collective bounds of \cite{gauthier2025statistical}}\label{app:single_collective}\vspace{-0.2cm}

\paragraph{Relation to the single-collective bounds of \cite{gauthier2025statistical}.}
We show that our multi-collective guarantees specialize to the corresponding
single-collective guarantees of \cite{gauthier2025statistical} when $M=1$,
for both \emph{signal planting} and \emph{signal unplanting}. Throughout this
appendix we condition on the realized sample sizes, as in the main text.

\paragraph{Signal Planting.}
Set $M=1$ and $c=1$ in Theorem~\ref{thm:multi_feature_label}. Write
\[
n := n_{N,1},\qquad
\alpha := \alpha_{N,1}=\frac{n}{N},\qquad
\alpha_{N,0}=1-\alpha=\frac{N-n}{N},
\]
and denote $g:=g_1$, $y^*:=y_1^*$, and
$\mathcal X^*:=\mathcal X_1^*=g(\mathcal X)$. Since there are no other
collectives, the cross-collective term vanishes and
\eqref{eq:multi_FL_bound_final_clean_preonly} specializes to
\begin{equation}\label{eq:single_collective_planting}
\begin{split}
\hat S(\alpha)
&\;\ge\;
\PP_{x\sim\hat P_{n}^{\mathrm{pre},X}}
\Big[
\alpha\Bigl(
\hat p^X\bigl(g(x)\bigr)
-
2R_{\tilde\delta}(n)
\Bigr)
-
(1-\alpha)\,
\Delta_{0}^{(n)}\bigl(g(x)\bigr)
\\
&\qquad\qquad
-
2(1-\alpha)\Bigl(
R_{\tilde\delta}(N-n)+R_{\tilde\delta}(n)
\Bigr)
-
\frac{\varepsilon}{1-\varepsilon}
>0
\Big]
\\
&\hspace{2.6em}
-\;R_{\delta/4}(n)-R_{\delta/4}(N^{\mathrm{test}}),
\end{split}
\end{equation}
where, consistently with Section~\ref{sec:MACA},
$\hat P_{n}^{\mathrm{pre}}$ is the empirical pre-edit distribution on the
collective sample, $\hat p^X(x^*):=\hat P_n^X(x^*)$ is the empirical post-edit
feature marginal under the chosen planting strategy, and
\begin{equation}\label{eq:single_collective_resistance}
\Delta_{0}^{(n)}(x^*)
:=
\max_{y'\in\mathcal Y\setminus\{y^*\}}
\bigl(
\hat P_{n}^{\mathrm{pre}}(x^*,y')
-
\hat P_{n}^{\mathrm{pre}}(x^*,y^*)
\bigr).
\end{equation}
Here
\[
\tilde\delta=\frac{\delta}{10|\mathcal X^*|\,|\mathcal Y|}.
\]
We now match \eqref{eq:single_collective_planting} with
\cite[Theorems~3.3 and 3.5]{gauthier2025statistical}.

\smallskip\noindent
\emph{Identification of empirical objects.}
Let $D^{(n)}$ denote the pre-edit collective dataset and
$\tilde D^{(n)}$ its post-edit version. Under our notation,
\[
\hat P_{D^{(n)}} \equiv \hat P_{n}^{\mathrm{pre}},
\qquad
\hat P_{\tilde D^{(n)}} \equiv \hat P_{n},
\qquad
\hat P_{\tilde D^{(n)}}^X \equiv \hat P_{n}^{X}.
\]
Under feature--label planting, $h(x,y)=(g(x),y^*)$, so for any
$x^*\in\mathcal X^*$ we have
$\hat P_{\tilde D^{(n)}}(x^*,y')=0$ for all $y'\neq y^*$ and hence
\[
\hat P_{\tilde D^{(n)}}^X(x^*)
=
\hat P_{\tilde D^{(n)}}(x^*,y^*).
\]
Therefore the signal term in \eqref{eq:single_collective_planting} can be
read either as a feature-marginal prevalence, as we do, or as a joint mass on
$(x^*,y^*)$. Moreover, the non-collective conflict term
\eqref{eq:single_collective_resistance} coincides with the corresponding
empirical conflict quantity in \cite[Theorem~3.3]{gauthier2025statistical}
after identifying $x^*$ with the planted feature value.

\smallskip\noindent
\emph{Outer probability / change of variables.}
Our outer probability is written as
$\PP_{x\sim \hat P_{n}^{\mathrm{pre},X}}[\cdot]$ with the margin evaluated at
$g(x)$. Since
\[
(g)_{\#}\hat P_{n}^{\mathrm{pre},X}
=
\hat P_{n}^{X}
\]
by construction, for any measurable $\Psi$ on $\mathcal X^*$ we have
\[
\PP_{x\sim\hat P_{n}^{\mathrm{pre},X}}
\bigl[\Psi(g(x))\bigr]
=
\PP_{x^*\sim \hat P_{n}^{X}}
\bigl[\Psi(x^*)\bigr].
\]
Thus \eqref{eq:single_collective_planting} can equivalently be written with
the outer probability taken over post-edit signal features, matching the
presentation in \cite{gauthier2025statistical}.

\smallskip\noindent
\emph{Comparison to \cite{gauthier2025statistical}.}
With the identifications above, \eqref{eq:single_collective_planting} has the
same functional structure as \cite[Theorem~3.3]{gauthier2025statistical}
(feature--label planting) and \cite[Theorem~3.5]{gauthier2025statistical}
(feature-only planting): an outer empirical probability over collective
features, a positive term proportional to $n/N$ times the post-edit prevalence
of the planted feature, a negative term proportional to $(N-n)/N$ times an
empirical conflict term computed from the collective's pre-edit data, a
suboptimality offset $\varepsilon/(1-\varepsilon)$, and explicit Hoeffding
radii. The remaining differences are only in the bookkeeping of constants,
through the particular split of $\delta$ and the resulting Hoeffding radii.

\paragraph{Signal Unplanting.}
Set $M=1$ and $c=1$ in Theorem~\ref{thm:multi_unplant}. Keep the shorthand
\[
n:=n_{N,1},\qquad
\alpha:=\frac{n}{N},
\qquad
\alpha_{N,0}=1-\alpha=\frac{N-n}{N},
\]
and let $n_e:=n_{e,1}$ denote the size of the estimation sub-collective.
Write again $g:=g_1$, $y^*:=y_1^*$, and $\mathcal X^*:=g(\mathcal X)$.
Let $E_N\subseteq I_{N,1}$ be the estimation index set of size $n_e$ and
$R_N:=I_{N,1}\setminus E_N$ its complement. Define the associated empirical
measures
\[
\hat P_{n}^{\mathrm{pre},E}
:=
\frac{1}{n_e}\sum_{i\in E_N}\delta_{Z_i},
\qquad
\hat P_{n}^{\mathrm{pre},R}
:=
\frac{1}{n-n_e}\sum_{i\in R_N}\delta_{Z_i}.
\]
For each $x^*\in\mathcal X^*$, define the data-driven alternative label
\begin{equation}\label{eq:single_collective_hat_y}
y'(x^*)
\;:=\;
\arg\max_{y\in\mathcal Y\setminus\{y^*\}}
\hat P_{n}^{\mathrm{pre},E}(x^*,y),
\end{equation}
with deterministic tie-breaking, and consider the adaptive unplanting rule
\[
h(x,y)=(g(x),y'(g(x))).
\]
In the single-collective case, the cross-collective term disappears and
\eqref{eq:multi_unplant_bound} specializes to
\begin{equation}\label{eq:single_collective_unplanting}
\begin{split}
\hat S(\alpha)
&\;\ge\;
\PP_{x\sim\hat P_{n}^{\mathrm{pre},R,X}}
\Big[
\alpha\Bigl(
\hat p^X\bigl(g(x)\bigr)
-
2R_{\tilde\delta}(n)
\Bigr)
-
(1-\alpha)\,
\Delta_{0}^{\mathrm{unp},(n-n_e)}\bigl(g(x)\bigr)
\\
&\qquad\qquad
-
2(1-\alpha)\Bigl(
R_{\tilde\delta}(N-n)+R_{\tilde\delta}(n-n_e)
\Bigr)
-
\frac{\varepsilon}{1-\varepsilon}
>0
\Big]
\\
&\hspace{2.6em}
-\;R_{\delta/4}(n-n_e)-R_{\delta/4}(N^{\mathrm{test}}),
\end{split}
\end{equation}
where the approximate non-collective conflict term is computed from the
independent holdout split $R_N$:
\begin{equation}\label{eq:single_collective_unplant_resistance}
\Delta_{0}^{\mathrm{unp},(n-n_e)}(x^*)
:=
\hat P_{n}^{\mathrm{pre},R}(x^*,y^*)
-
\hat P_{n}^{\mathrm{pre},R}\bigl(x^*,y'(x^*)\bigr).
\end{equation}

\smallskip\noindent
\emph{Match to \cite{gauthier2025statistical}.}
Under the identification
\[
D^{(n_e)}
\leftrightarrow
\text{estimation split }E_N,
D^{(n-n_e)}
\leftrightarrow
\text{holdout split }R_N,
\tilde D^{(n)}
\leftrightarrow
\text{post-edit collective dataset},
\]
the label-selection rule \eqref{eq:single_collective_hat_y} coincides with the
adaptive unplanting label-selection rule in \cite{gauthier2025statistical}
up to notation, and the conflict term
\eqref{eq:single_collective_unplant_resistance} matches the corresponding
holdout-based empirical conflict term in
\cite[Theorem~3.7]{gauthier2025statistical}. The outer probability in
\eqref{eq:single_collective_unplanting} is taken over the holdout empirical
feature distribution, because the same holdout split is used in the proof to
estimate the non-collective conflict independently of the label-selection step.
Equivalently, since the post-edit holdout feature distribution is the pushforward
of $\hat P_{n}^{\mathrm{pre},R,X}$ under $g$, one may rewrite the outer
probability over the corresponding post-edit holdout signal features.

Therefore, \eqref{eq:single_collective_unplanting} specializes the
multi-collective adaptive unplanting guarantee to the single-collective setting
and has the same functional structure as
\cite[Theorem~3.7]{gauthier2025statistical}: a positive signal-prevalence term,
a holdout-based empirical conflict term for the non-collective population, the
suboptimality offset $\varepsilon/(1-\varepsilon)$, and explicit concentration
radii. The remaining differences are again only in constant bookkeeping through
the chosen confidence split.

\vspace{-0.3cm}
\section{Additional Experimental Results}
\label{app:experiment-results}
\vspace{-0.2cm}

\subsection{Two Collective Experiments}

\cameraready{The two-collective feature-label bounds of Figure~\ref{fig:feature_analysis} and Figures~\ref{fig:feature_analysis_fl_per-collective_17v17}--\ref{fig:feature_analysis_fl_known-interventions_17v13} are computed with an implementation of Theorem~\ref{thm:multi_feature_label} that follows its statement exactly, available in our code, with $N=2{,}000{,}000$, $N^{\mathrm{test}}=100{,}000$, total collective mass $\bar\alpha_N\in\{0.1,0.2,\dots,0.9\}$ (combined collective size from $200{,}000$ to $1{,}800{,}000$) and $99$ equally spaced shares of Collective~1 at each mass. In the known-interventions setting, each collective knows a pool of possible interventions: the nested patterns that plant $13$, $14$, $15$, $16$ or $17$ of the $17$ features, each with every target label other than its own, for $15$ interventions. It estimates the post-edit distribution of each intervention by applying it to its own pre-edit sample, which is valid under neutrality, averages the resulting conflicts over the pool, and weights the average by the total mass of the other collectives.}

We report additional experiments in the 2-collective feature overlap setting described in \cameraready{Section}~\ref{sec:experiments} in both the feature-label and the feature-only settings. In the feature-label setting we observe a qualitative behavior consistent with the trend observed in the main text: reallocating mass toward one collective improves its lower bound while reducing the other's, and increasing the total planted mass shifts the attainable guarantees upward. The symmetric 17-vs.-17 case is driven primarily by the mass split (Figures~\ref{fig:feature_analysis_fl_per-collective_17v17} \&~\ref{fig:feature_analysis_fl_known-interventions_17v17}). Reducing the number of planted features for \textit{Collective~2} introduces the asymmetry observed in Figure~\ref{fig:feature_analysis}: As \textit{Collective~2}'s planted features decrease, \textit{Collective~2} loses signal prevalence, whereas \textit{Collective~1} benefits from features that are no longer directly contested (Figures~\ref{fig:feature_analysis_fl_per-collective_17v15},~\ref{fig:feature_analysis_fl_per-collective_17v13},~\ref{fig:feature_analysis_fl_known-interventions_17v15}, \&~\ref{fig:feature_analysis_fl_known-interventions_17v13}).

In the 17 vs. 13 case, under the \textit{per-collective} setting, we observe that collective 1 is insensitive to the mass of collective 2 (Figure~\ref{fig:feature_analysis_fl_per-collective_17v13}). The quantity of features \cameraready{prevails} over the mass proportions of each collective. This is not the case in the \textit{known-interventions} setting (Figure~\ref{fig:feature_analysis_fl_known-interventions_17v13})\cameraready{, where the averaged conflict does not reflect how little the realized 13-feature pattern contests Collective~1's planted values, and Collective~1's threshold grows with the combined size, from 33{,}200 to 64{,}800 samples. More generally, the known-interventions estimate is more conservative than the per-collective bound when the actual opponent conflicts less than the average intervention in the pool (Collective~1 in the 17-vs.-15 and 17-vs.-13 cases), and optimistic when it conflicts more (both collectives in the 17-vs.-17 case and Collective~2 in the 17-vs.-15 case)}. Like the \textit{per-collective} setting, the maximum success \cameraready{that Collective~2 can reach decreases as Collective~2 plants fewer features ($0.99$, $0.42$ and $0.33$ for 17, 15 and 13 features)}.

\begin{figure*}[ht!] 
    \centering
    \begin{subfigure}[]{0.32\linewidth}
        \includegraphics[height=1.5in]{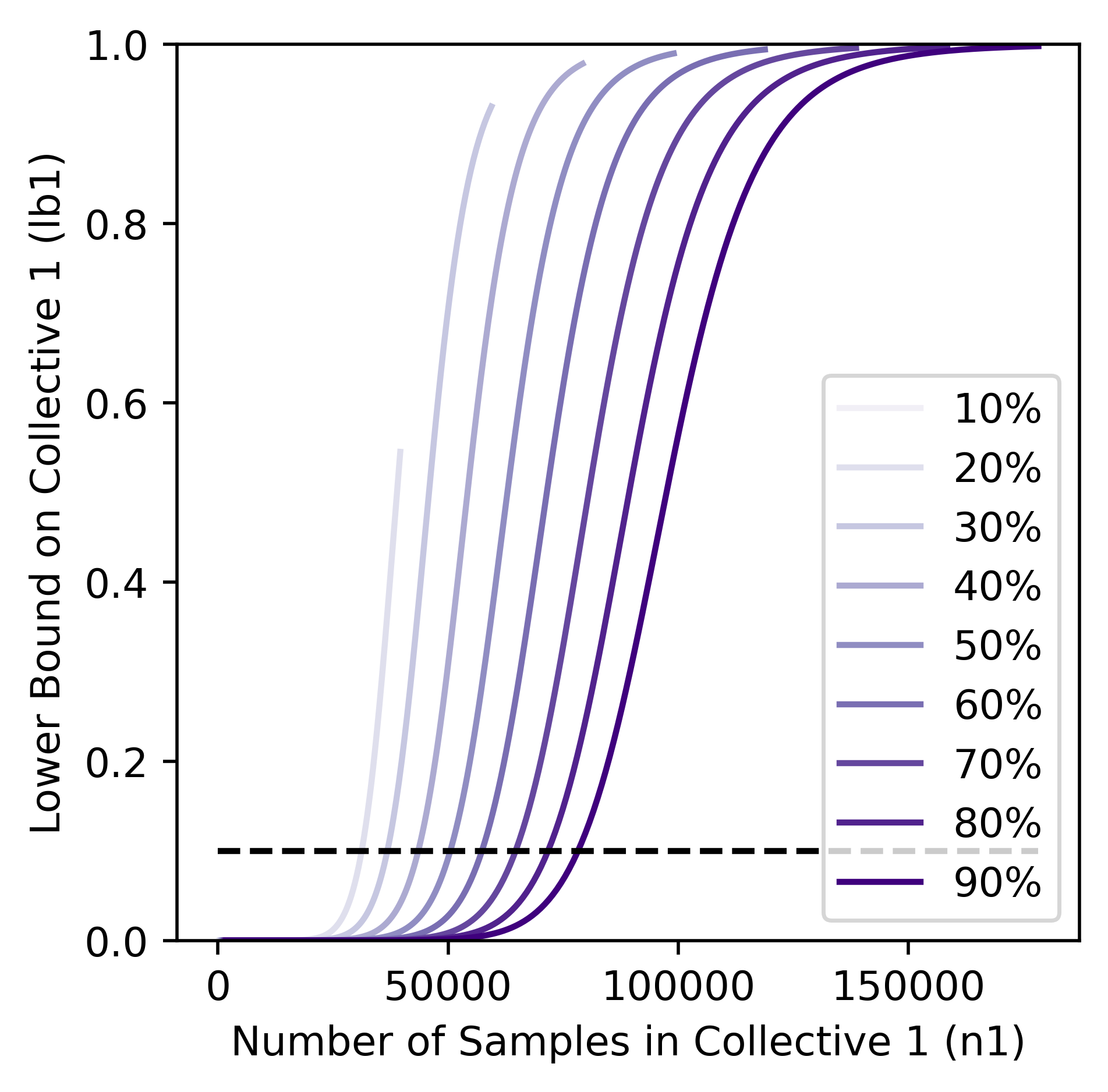}
        \caption{Collective 1 lower bound on success with a partial conflict with 17 vs. 17 features and a per-collective strategy.}
    \end{subfigure}
    \centering
    \begin{subfigure}[]{0.32\linewidth}
        \includegraphics[height=1.5in]{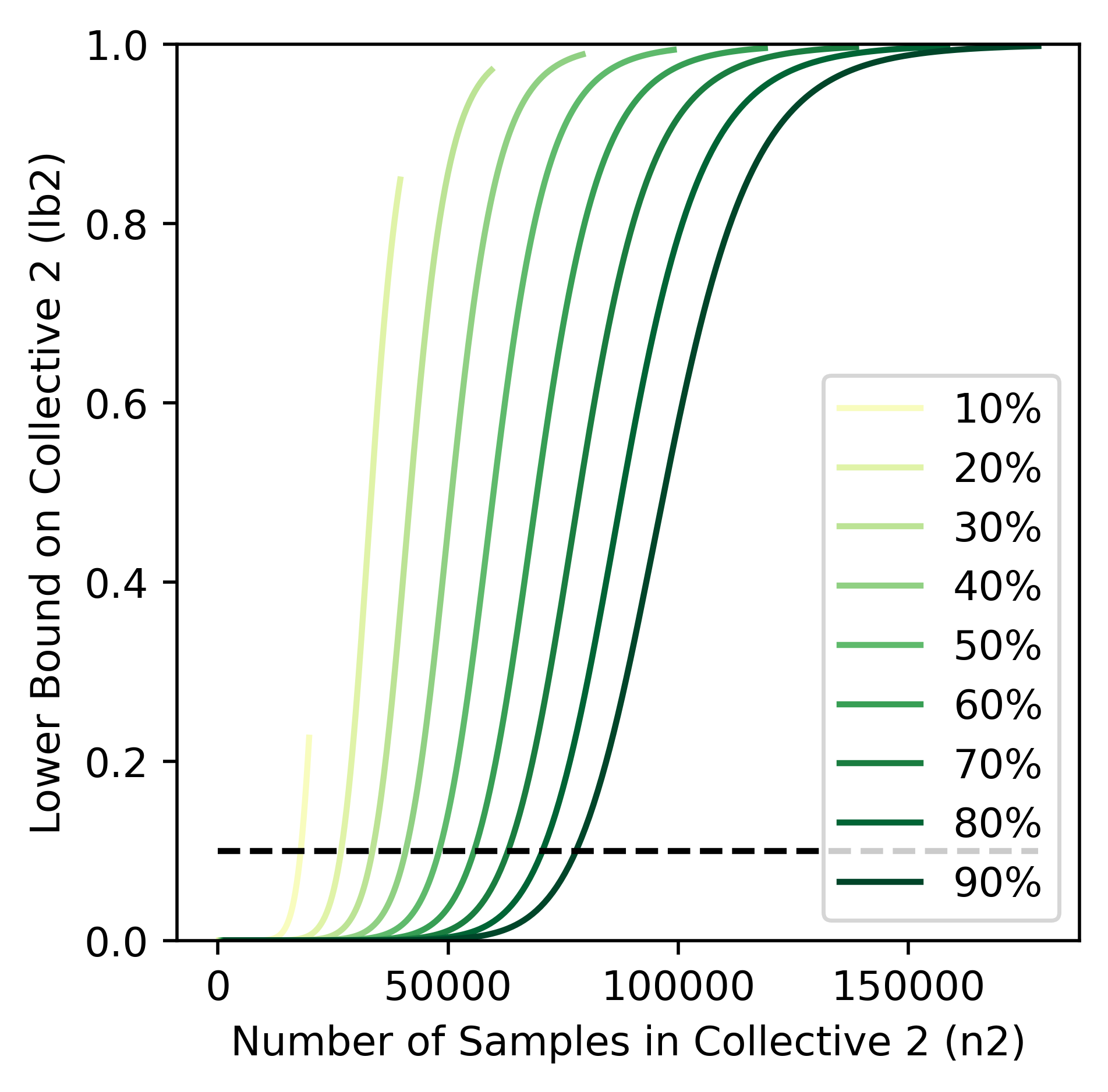}
        \caption{Collective 2 lower bound on success with a partial conflict with 17 vs. 17 features and a per-collective strategy.}
    \end{subfigure}
    \begin{subfigure}[]{0.32\linewidth}
        \includegraphics[height=1.5in]{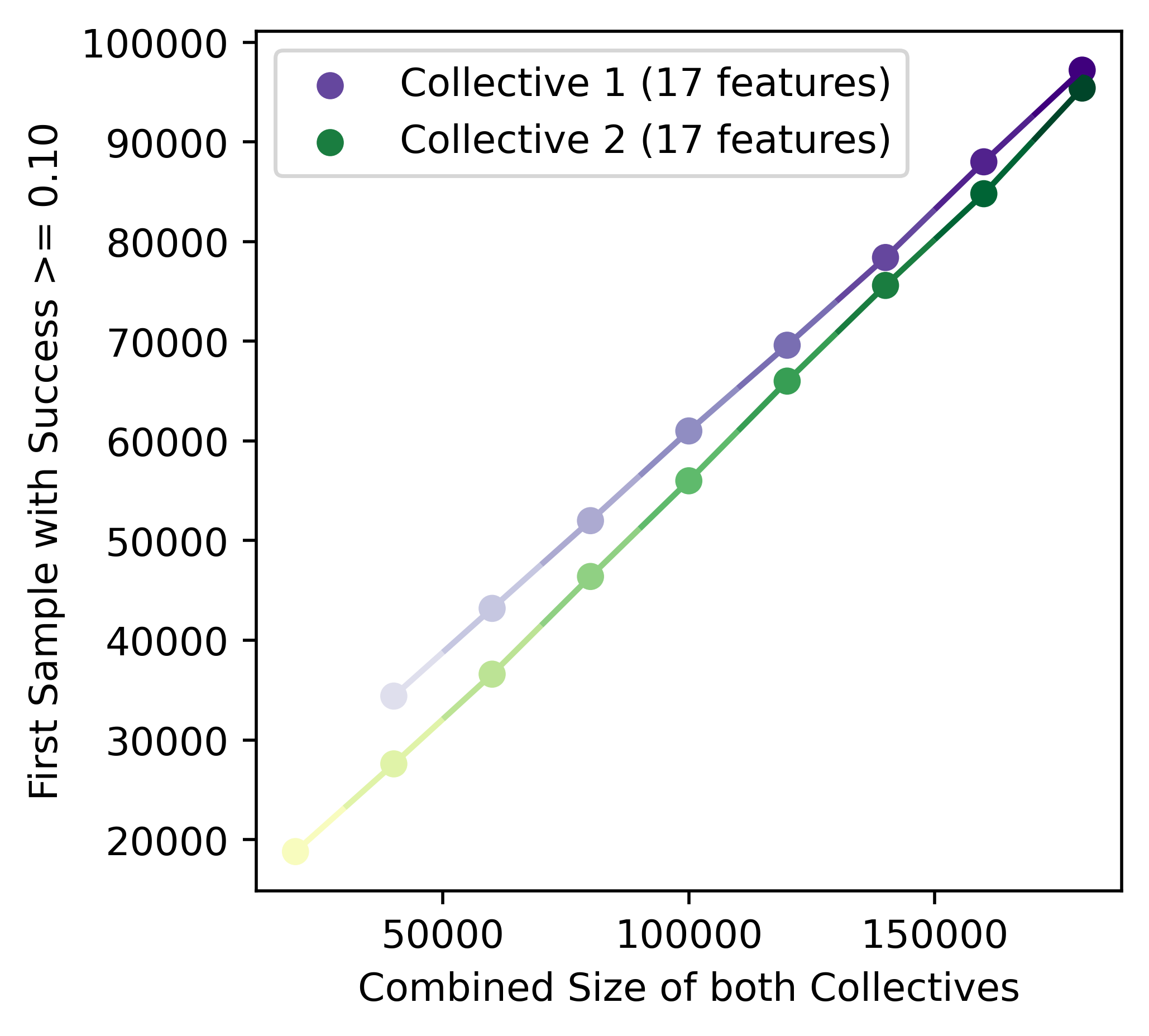}
        \caption{The number of samples required to exceed a 0.1 lower bound on success in a 17 vs. 17 collective scenario. }
    \end{subfigure}
    \caption{The bound on success for two collectives under varying population sizes and collective sizes.}
    \label{fig:feature_analysis_fl_per-collective_17v17}
\end{figure*}

\begin{figure*}[ht!] 
    \centering
    \begin{subfigure}[]{0.32\linewidth}
        \includegraphics[height=1.5in]{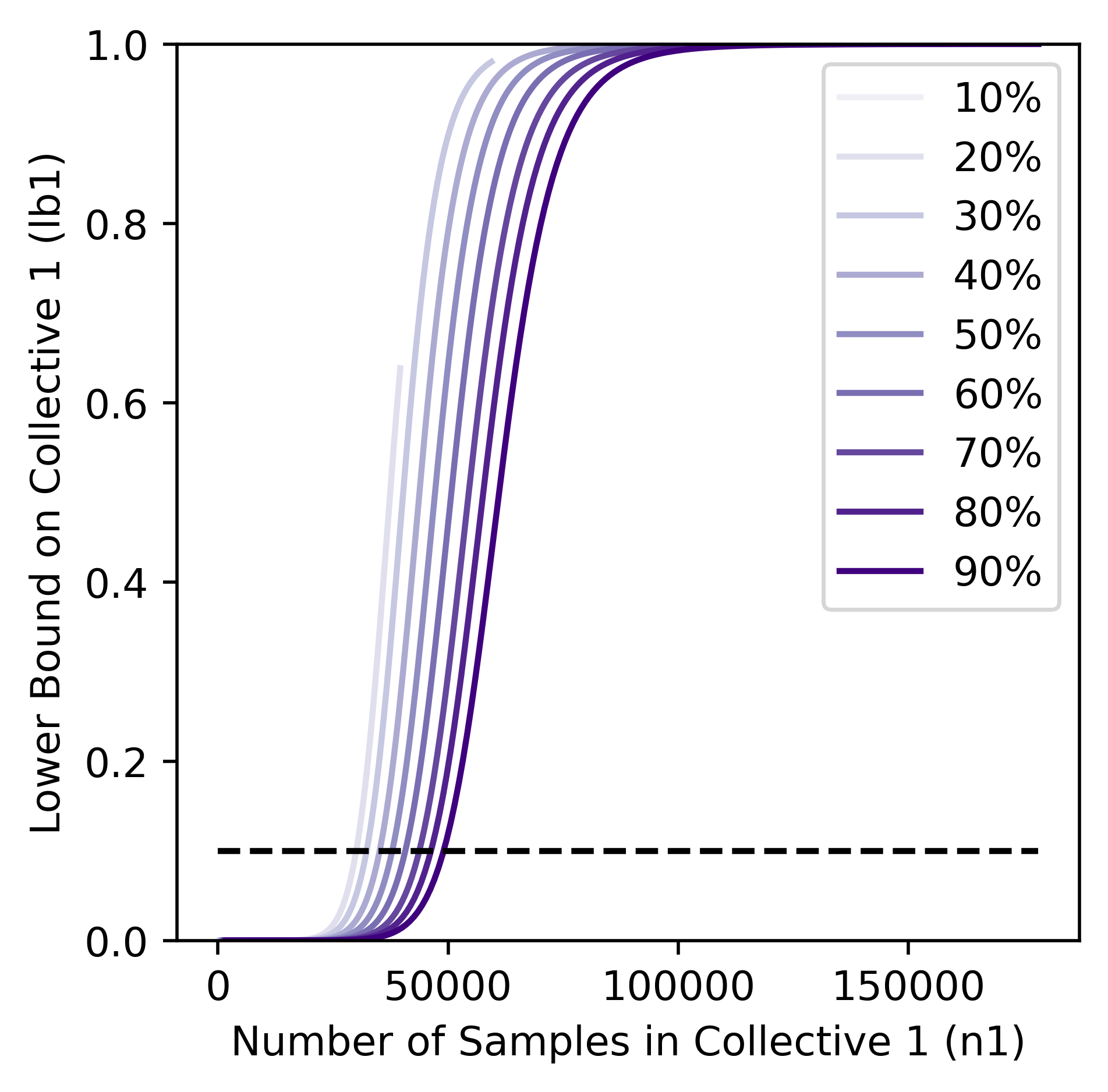}
        \caption{Collective 1 lower bound on success with a partial conflict with 17 vs. 15 features and a per-collective strategy.}
    \end{subfigure}
    \centering
    \begin{subfigure}[]{0.32\linewidth}
        \includegraphics[height=1.5in]{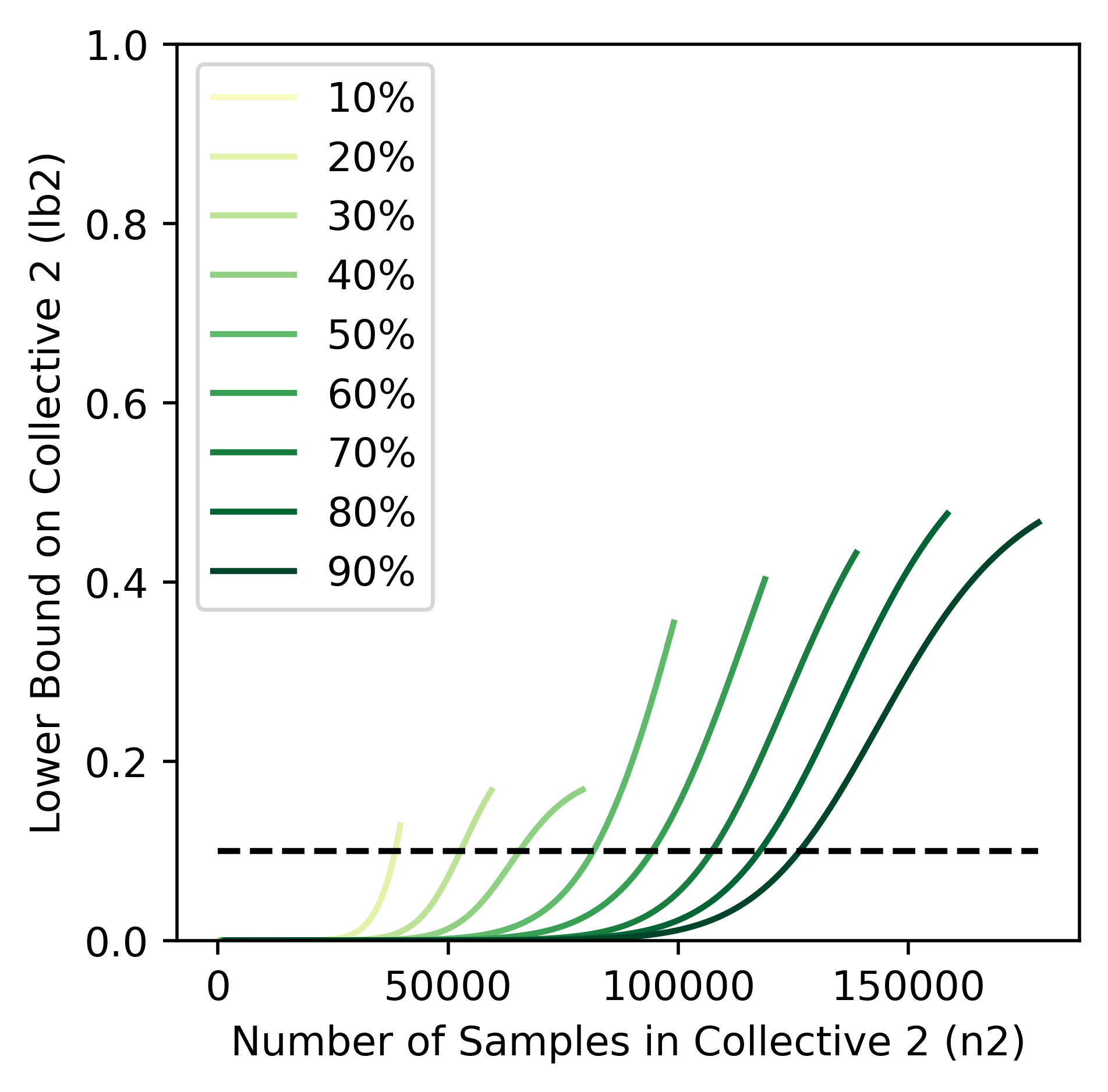}
        \caption{Collective 2 lower bound on success with a partial conflict with 17 vs. 15 features and a per-collective strategy.}
    \end{subfigure}
    \begin{subfigure}[]{0.32\linewidth}
        \includegraphics[height=1.5in]{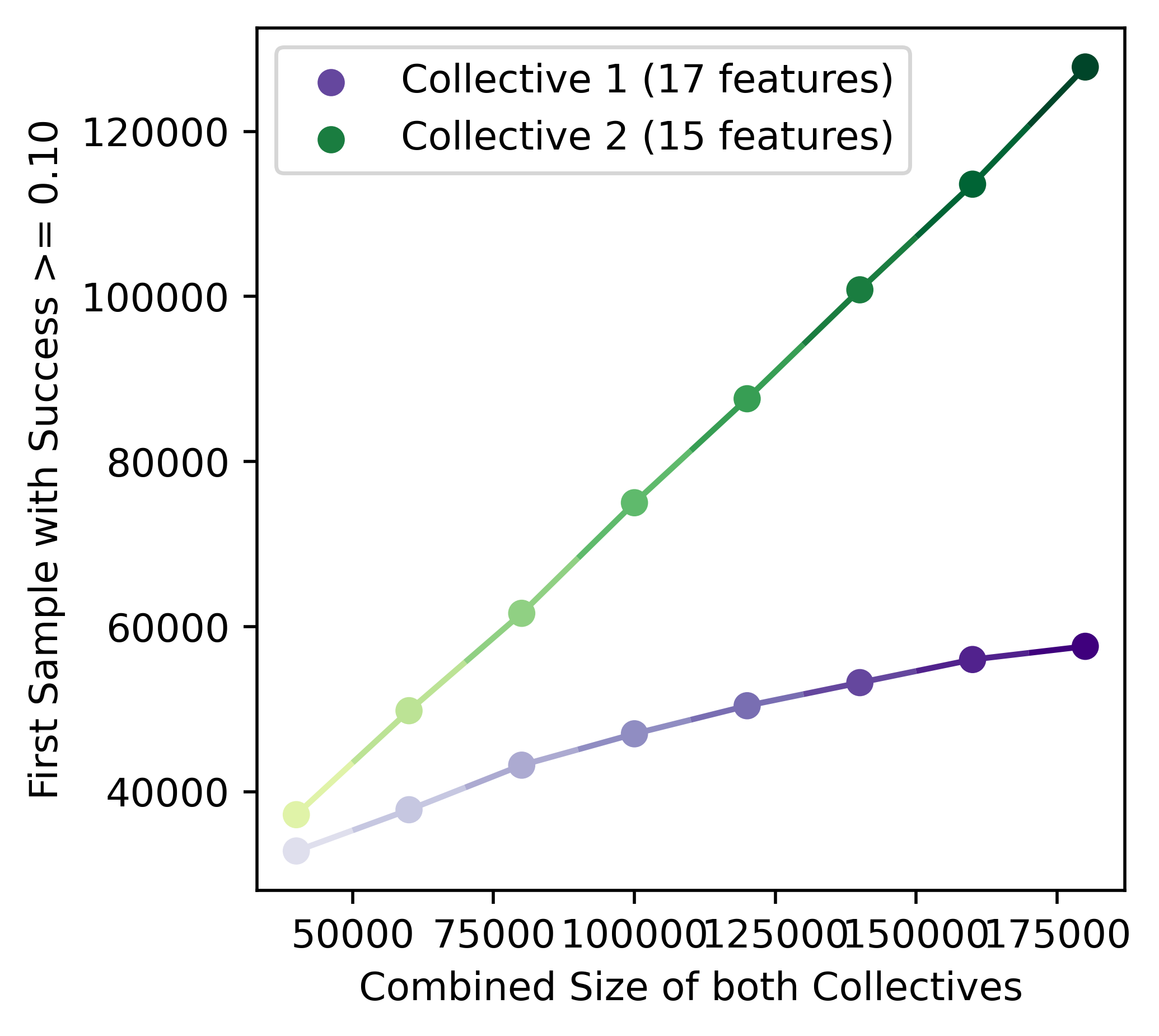}
        \caption{The number of samples required to exceed a 0.1 lower bound on success in a 17 vs. 15 collective scenario. }
    \end{subfigure}
    \caption{The bound on success for two collectives under varying population sizes and collective sizes.}
    \label{fig:feature_analysis_fl_per-collective_17v15}
\end{figure*}

\begin{figure*}[ht!] 
    \centering
    \begin{subfigure}[]{0.32\linewidth}
        \includegraphics[height=1.5in]{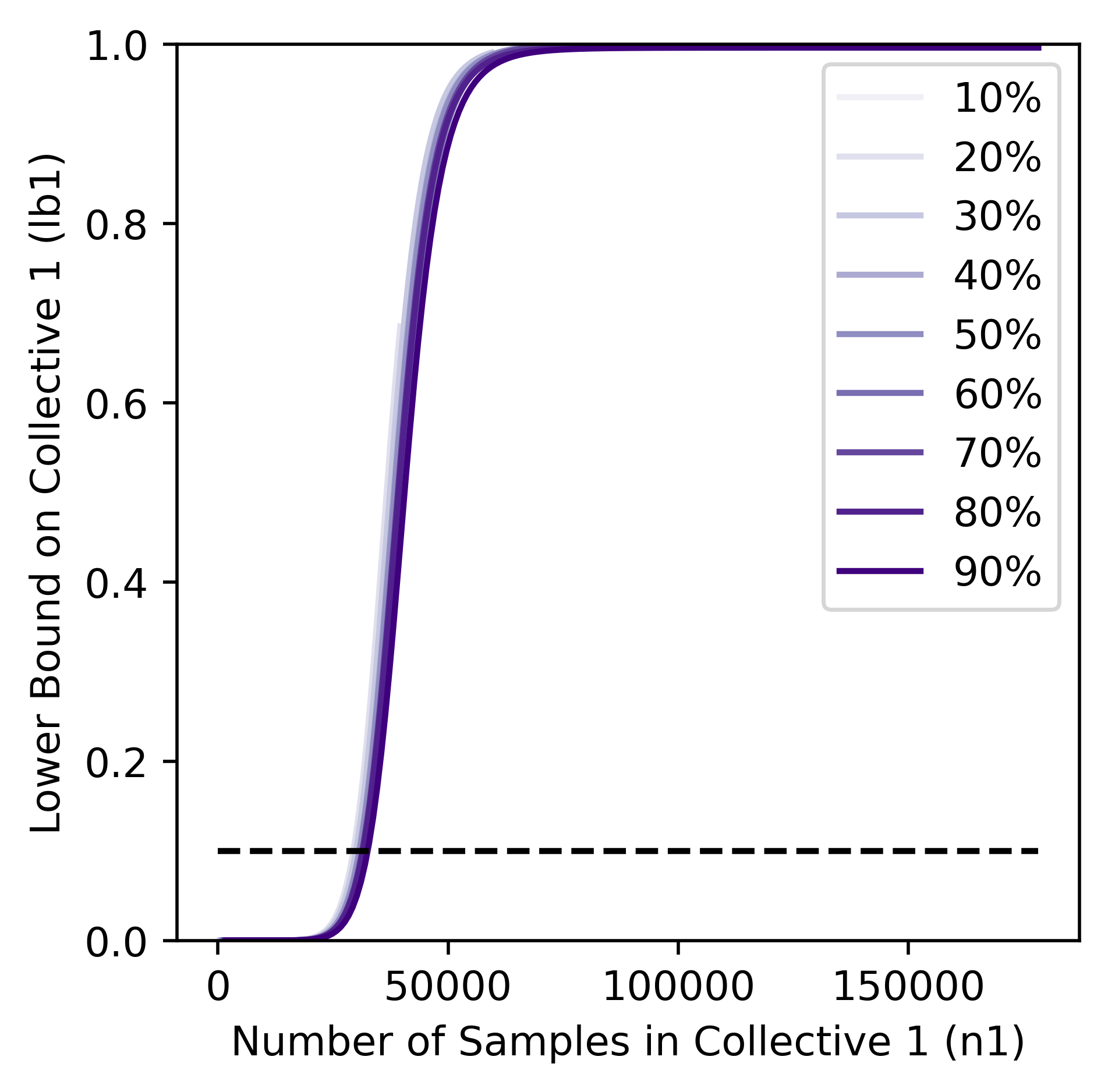}
        \caption{Collective 1 lower bound on success with a partial conflict with 17 vs. 13 features and a per-collective strategy.}
    \end{subfigure}
    \centering
    \begin{subfigure}[]{0.32\linewidth}
        \includegraphics[height=1.5in]{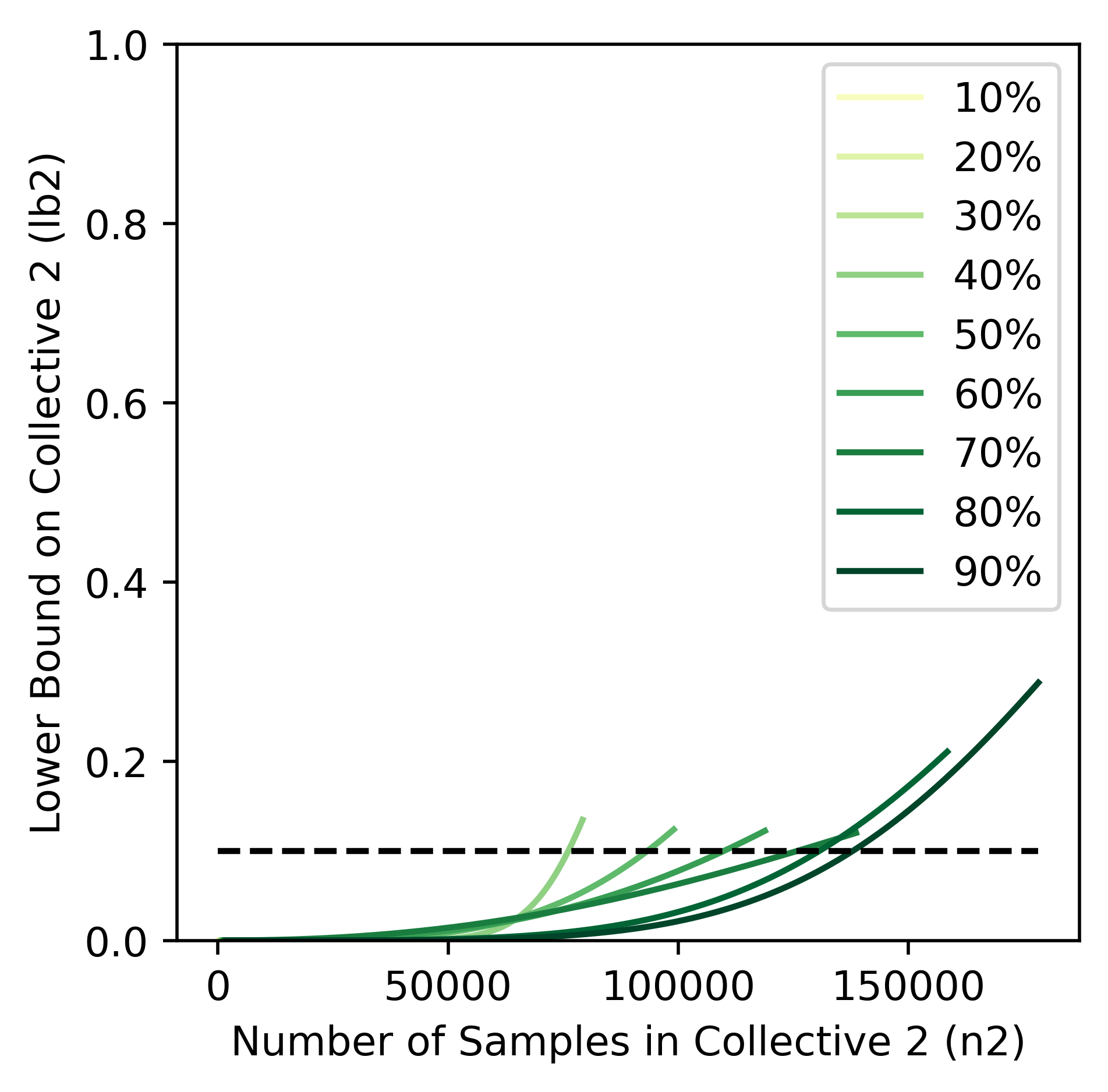}
        \caption{Collective 2 lower bound on success with a partial conflict with 17 vs. 13 features and a per-collective strategy.}
    \end{subfigure}
    \begin{subfigure}[]{0.32\linewidth}
        \includegraphics[height=1.5in]{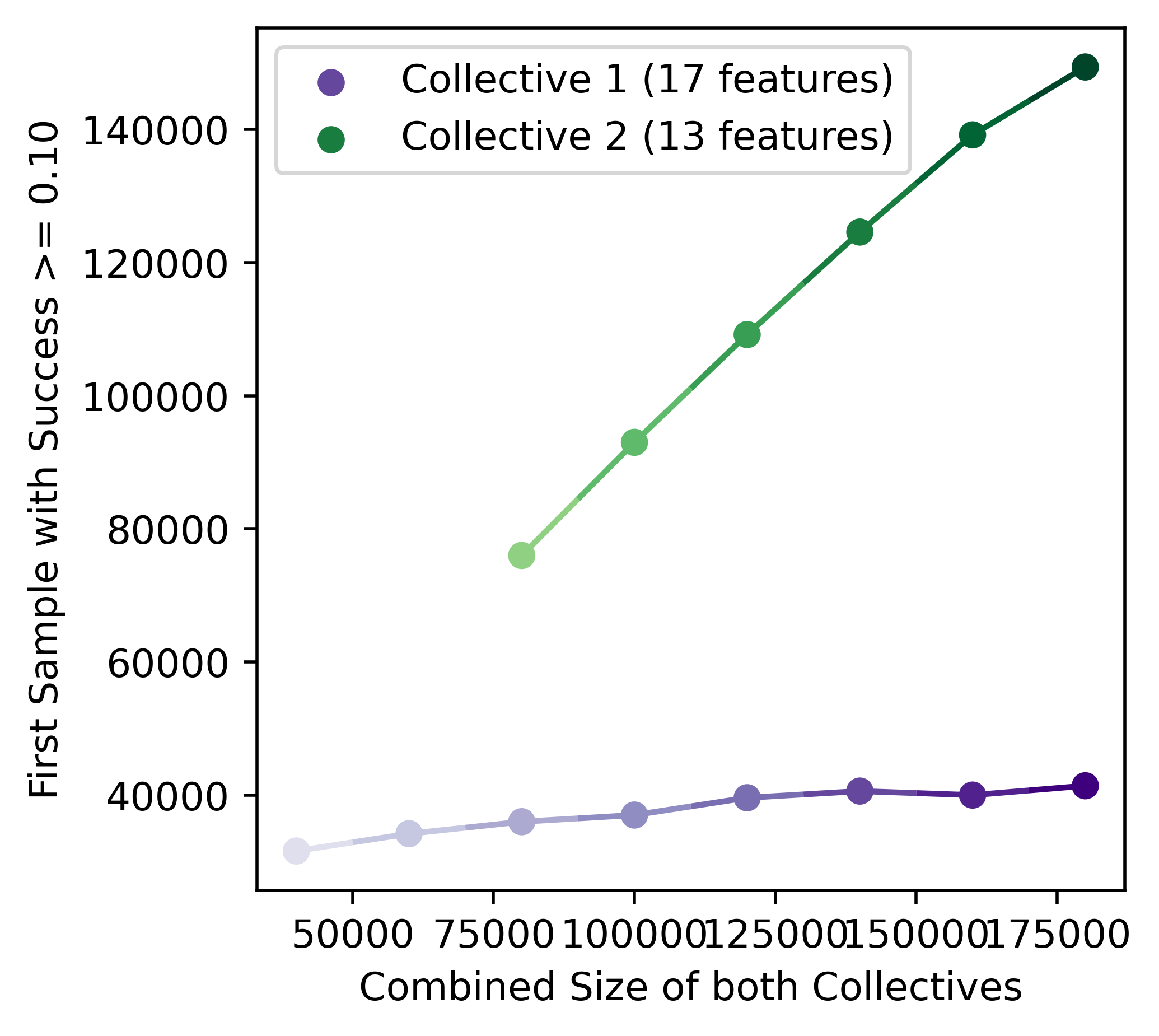}
        \caption{The number of samples required to exceed a 0.1 lower bound on success in a 17 vs. 13 collective scenario. }
    \end{subfigure}
    \caption{The bound on success for two collectives under varying population sizes and collective sizes.}
    \label{fig:feature_analysis_fl_per-collective_17v13}
\end{figure*}

% known interventions

\begin{figure*}[ht!] 
    \centering
    \begin{subfigure}[]{0.32\linewidth}
        \includegraphics[height=1.5in]{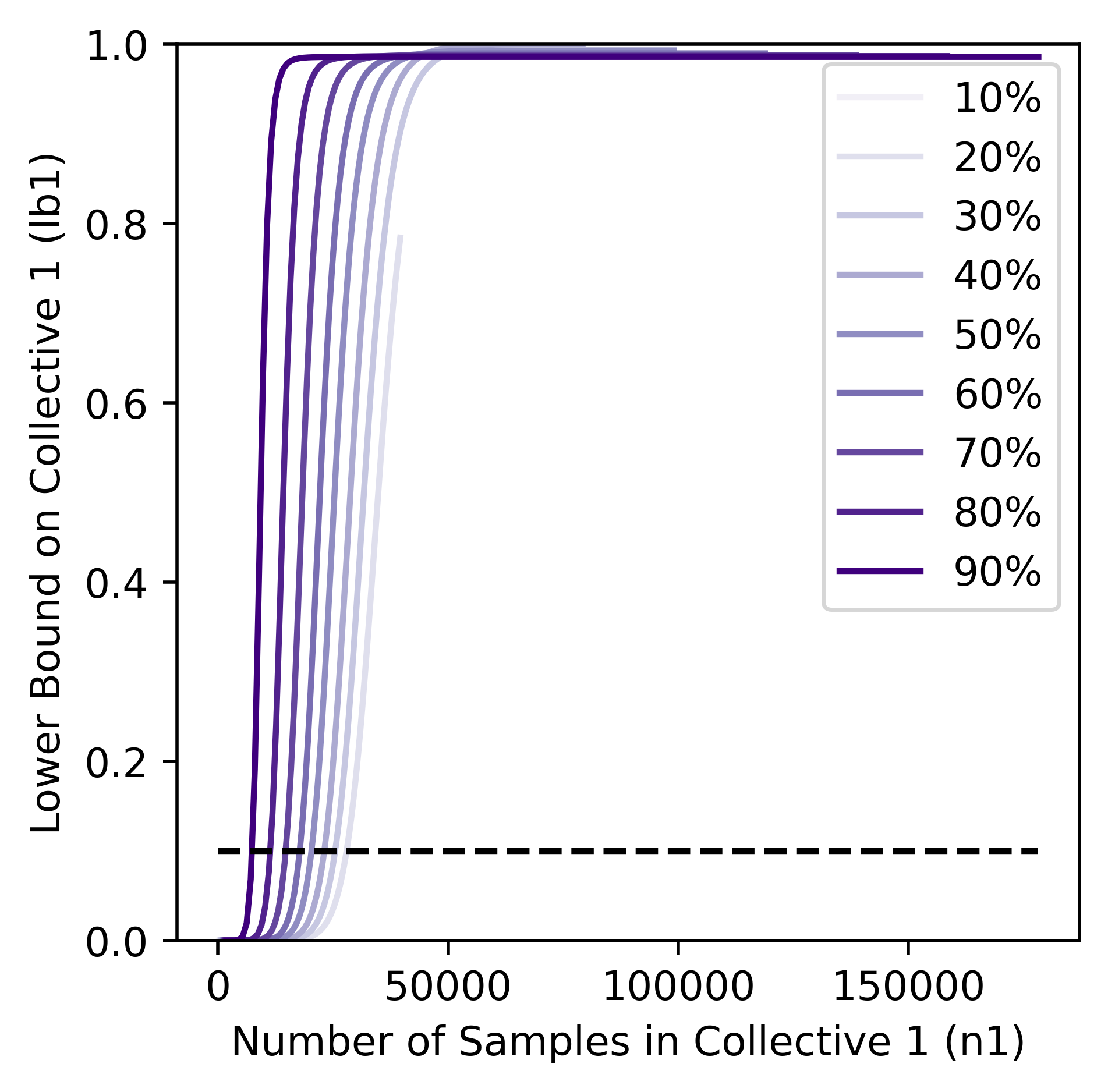}
        \caption{Collective 1 lower bound on success with a partial conflict with 17 vs. 17 features and a known-interventions strategy.}
    \end{subfigure}
    \centering
    \begin{subfigure}[]{0.32\linewidth}
        \includegraphics[height=1.5in]{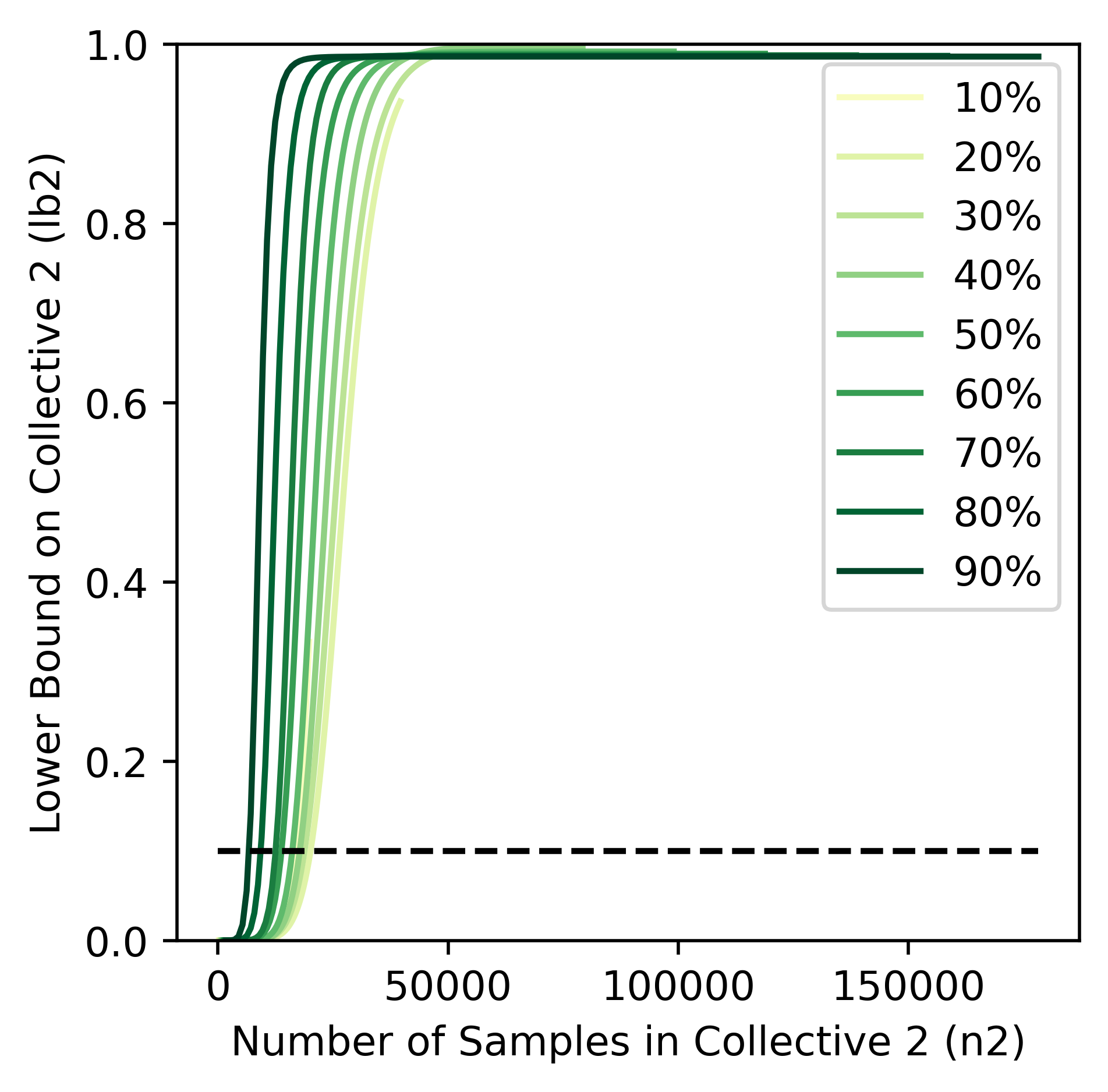}
        \caption{Collective 2 lower bound on success with a partial conflict with 17 vs. 17 features and a known-interventions strategy.}
    \end{subfigure}
    \begin{subfigure}[]{0.32\linewidth}
        \includegraphics[height=1.5in]{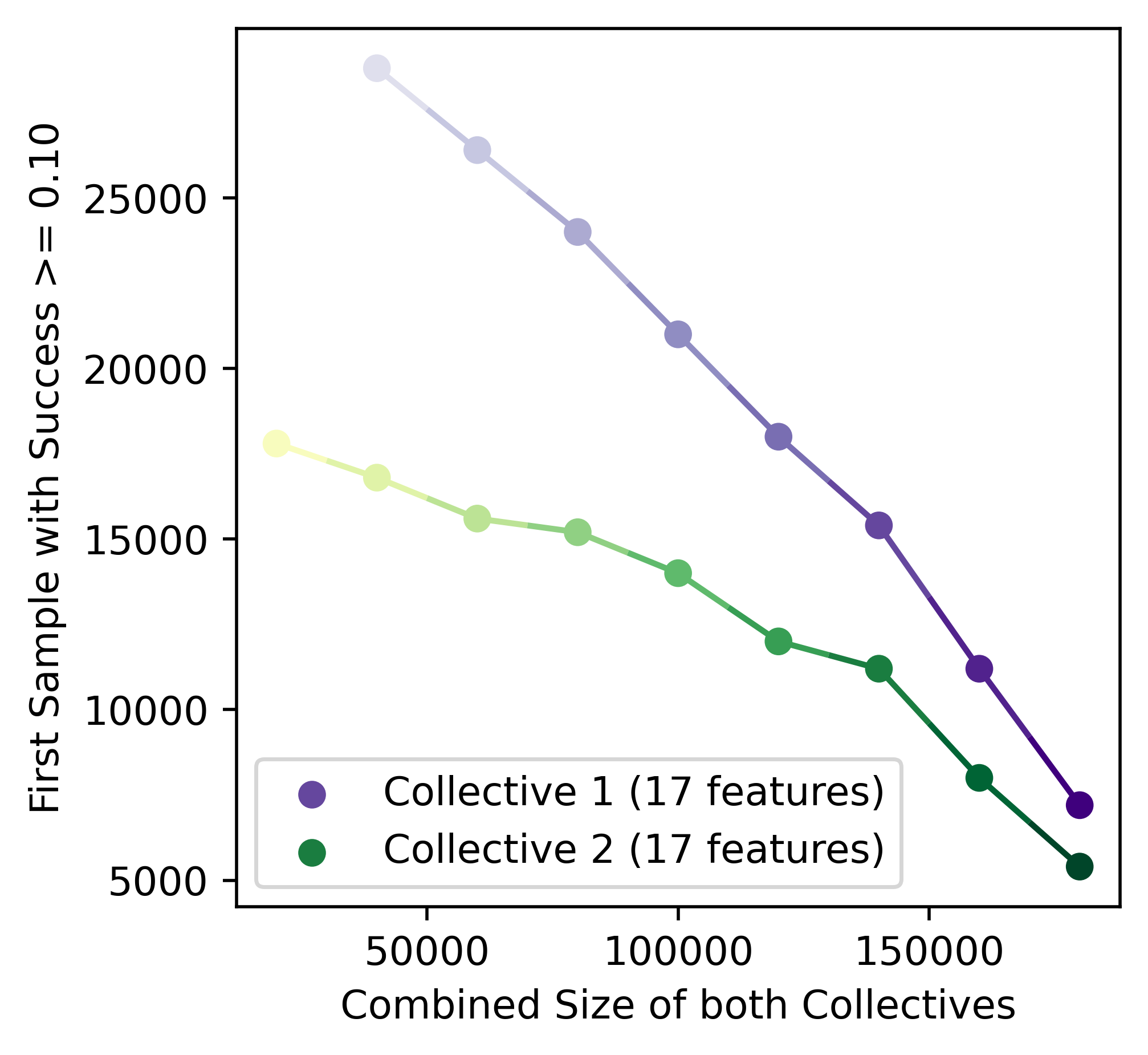}
        \caption{The number of samples required to exceed a 0.1 lower bound on success in a 17 vs. 17 collective scenario. }
    \end{subfigure}
    \caption{The bound on success for two collectives under varying population sizes and collective sizes.}
    \label{fig:feature_analysis_fl_known-interventions_17v17}
\end{figure*}

\begin{figure*}[ht!] 
    \centering
    \begin{subfigure}[]{0.32\linewidth}
        \includegraphics[height=1.5in]{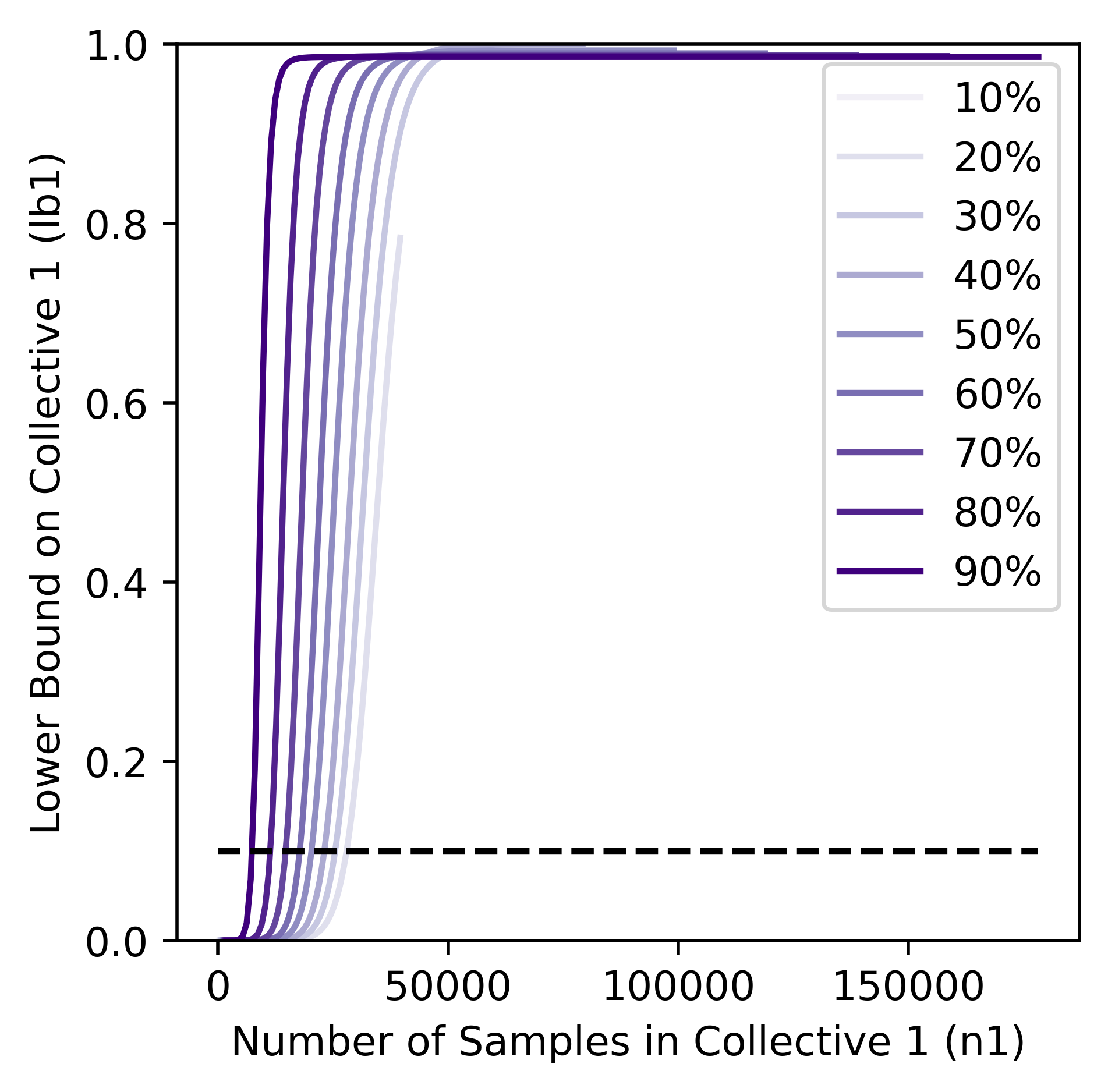}
        \caption{Collective 1 lower bound on success with a partial conflict with 17 vs. 15 features and a known-interventions strategy.}
    \end{subfigure}
    \centering
    \begin{subfigure}[]{0.32\linewidth}
        \includegraphics[height=1.5in]{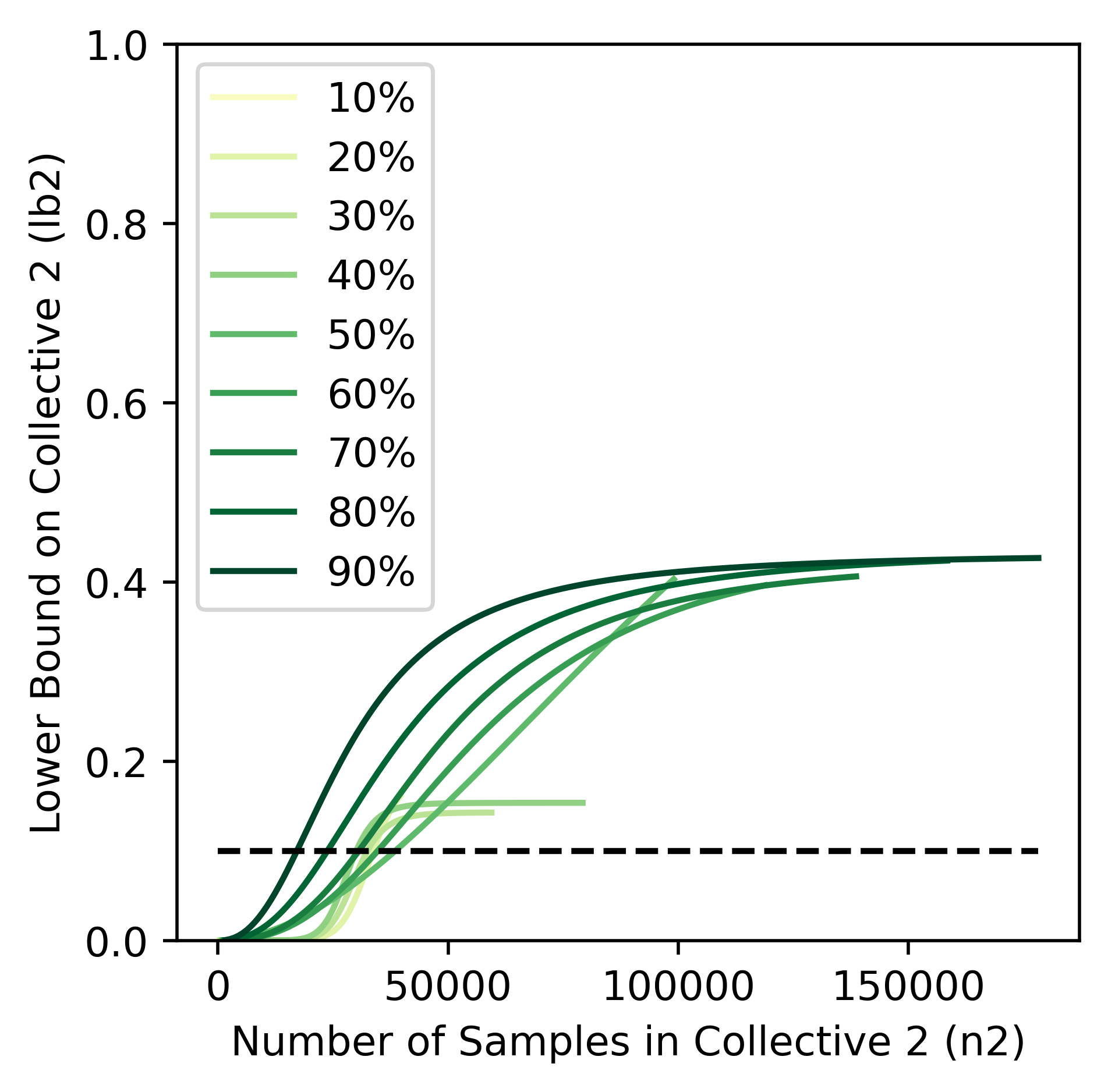}
        \caption{Collective 2 lower bound on success with a partial conflict with 17 vs. 15 features and a known-interventions strategy.}
    \end{subfigure}
    \begin{subfigure}[]{0.32\linewidth}
        \includegraphics[height=1.5in]{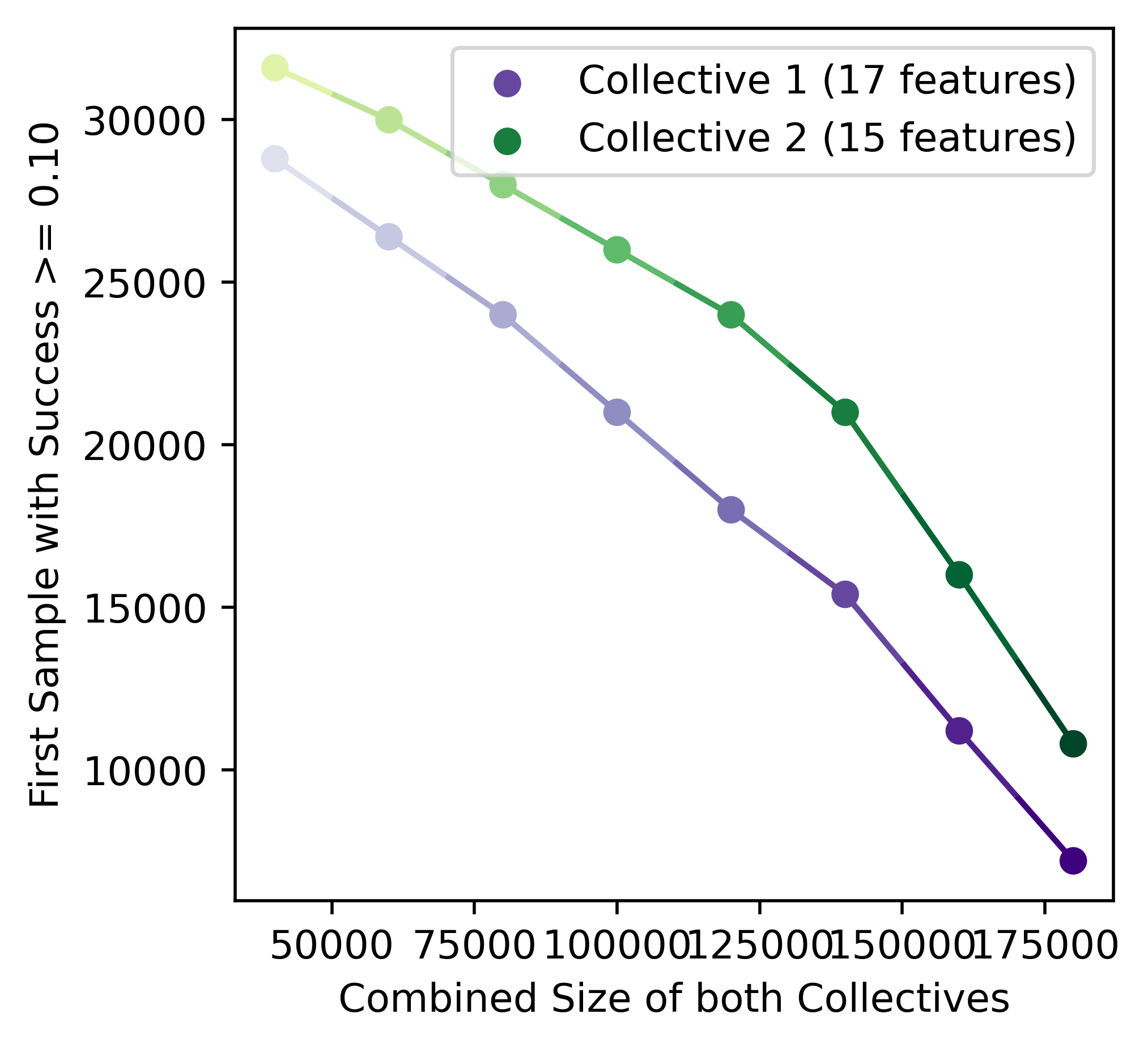}
        \caption{The number of samples required to exceed a 0.1 lower bound on success in a 17 vs. 15 collective scenario. }
    \end{subfigure}
    \caption{The bound on success for two collectives under varying population sizes and collective sizes.}
    \label{fig:feature_analysis_fl_known-interventions_17v15}
\end{figure*}

\begin{figure*}[ht!] 
    \centering
    \begin{subfigure}[]{0.32\linewidth}
        \includegraphics[height=1.5in]{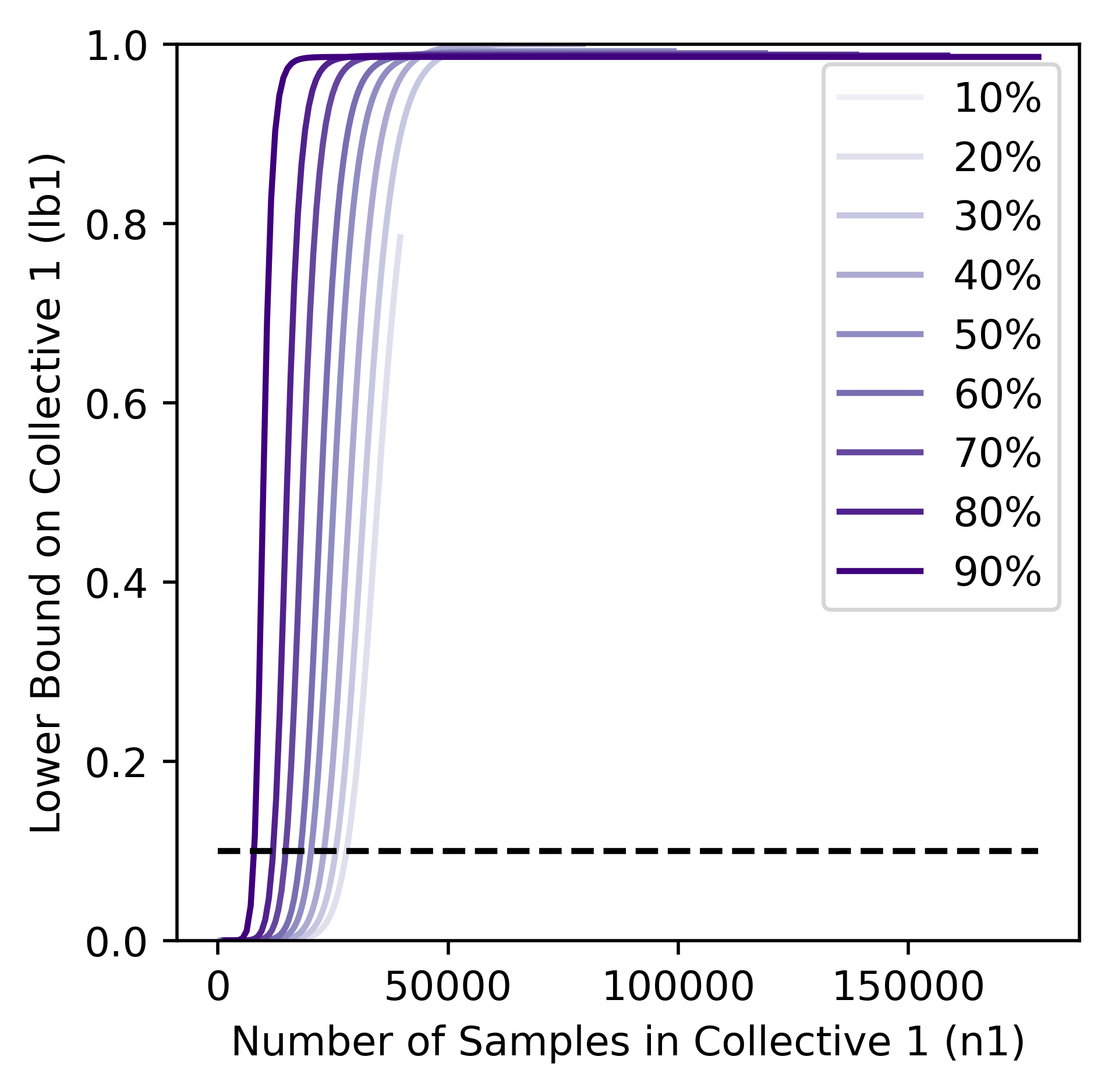}
        \caption{Collective 1 lower bound on success with a partial conflict with 17 vs. 13 features and a known-interventions strategy.}
    \end{subfigure}
    \centering
    \begin{subfigure}[]{0.32\linewidth}
        \includegraphics[height=1.5in]{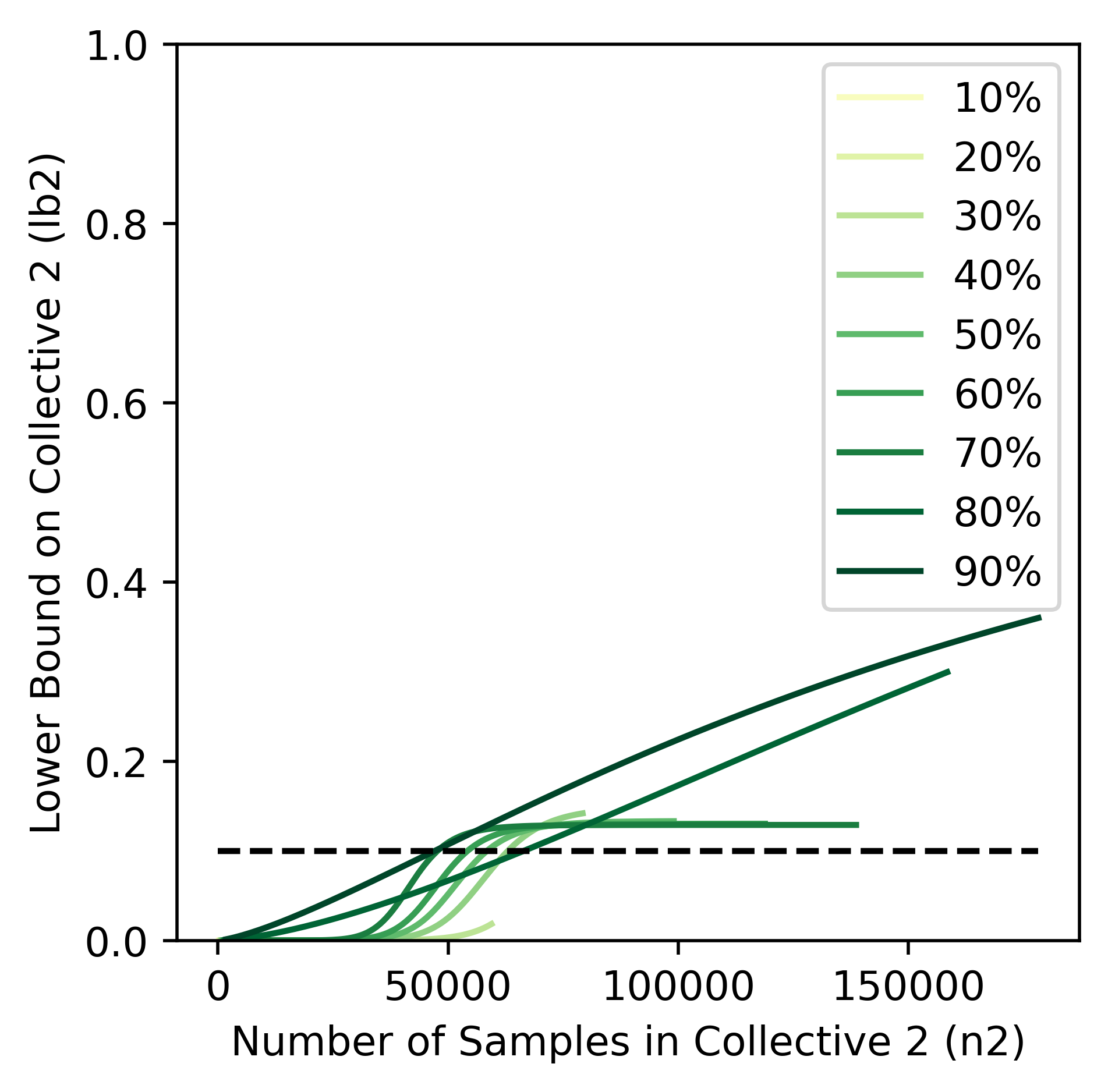}
        \caption{Collective 2 lower bound on success with a partial conflict with 17 vs. 13 features and a known-interventions strategy.}
    \end{subfigure}
    \begin{subfigure}[]{0.32\linewidth}
        \includegraphics[height=1.5in]{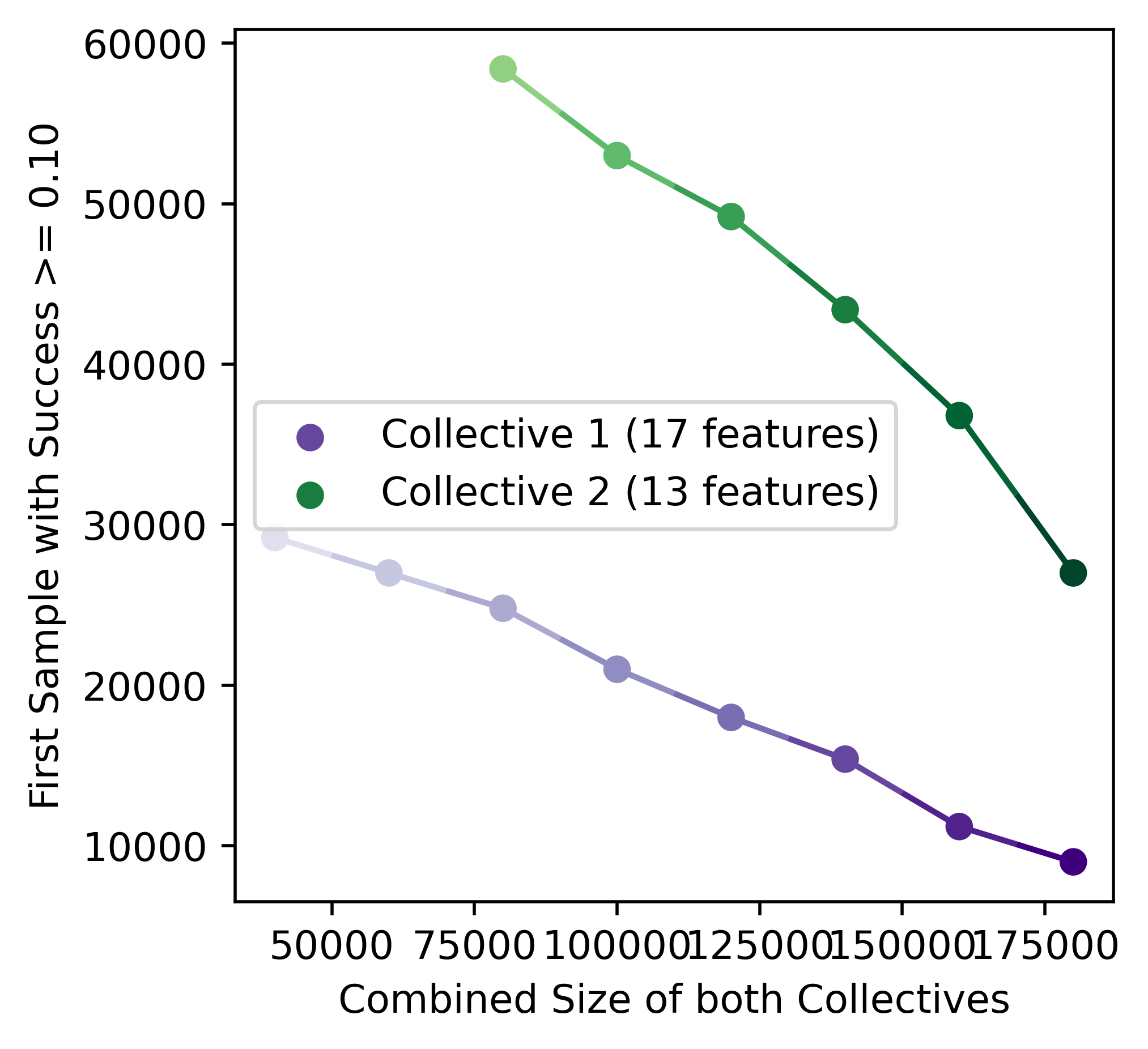}
        \caption{The number of samples required to exceed a 0.1 lower bound on success in a 17 vs. 13 collective scenario. }
    \end{subfigure}
    \caption{The bound on success for two collectives under varying population sizes and collective sizes.}
    \label{fig:feature_analysis_fl_known-interventions_17v13}
\end{figure*}

\newpage

\cameraready{For feature-only planting, we report the per-collective setting} (Figures~\ref{fig:feature_analysis_fo_per-collective_17v17},~\ref{fig:feature_analysis_fo_per-collective_17v16},~\ref{fig:feature_analysis_fo_per-collective_17v15},~\ref{fig:feature_analysis_fo_per-collective_17v14}, \&~\ref{fig:feature_analysis_fo_per-collective_17v13}). As the overlap decreases from 17-vs.-17 to 17-vs.-13, the conflict between collectives weakens for the collective retaining more planted features, while the collective with fewer planted features requires more favorable mass allocation to achieve comparable lower bounds.

\begin{figure*}[ht!]{}
    \centering
    \includegraphics[height=1.5in,width=\linewidth]{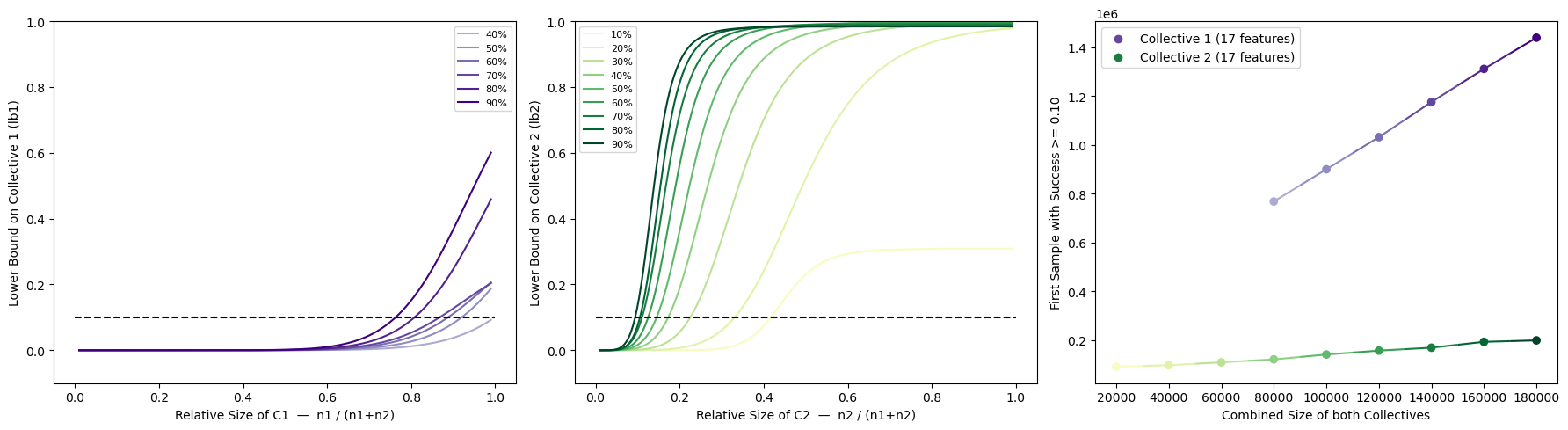}
    \caption{The bound on feature-only planting success for two collectives, 17 vs 17 features, under varying population sizes and collective sizes.}
    \label{fig:feature_analysis_fo_per-collective_17v17}
\end{figure*}

\begin{figure*}[ht!]{}
    \centering
    \includegraphics[height=1.5in,width=\linewidth]{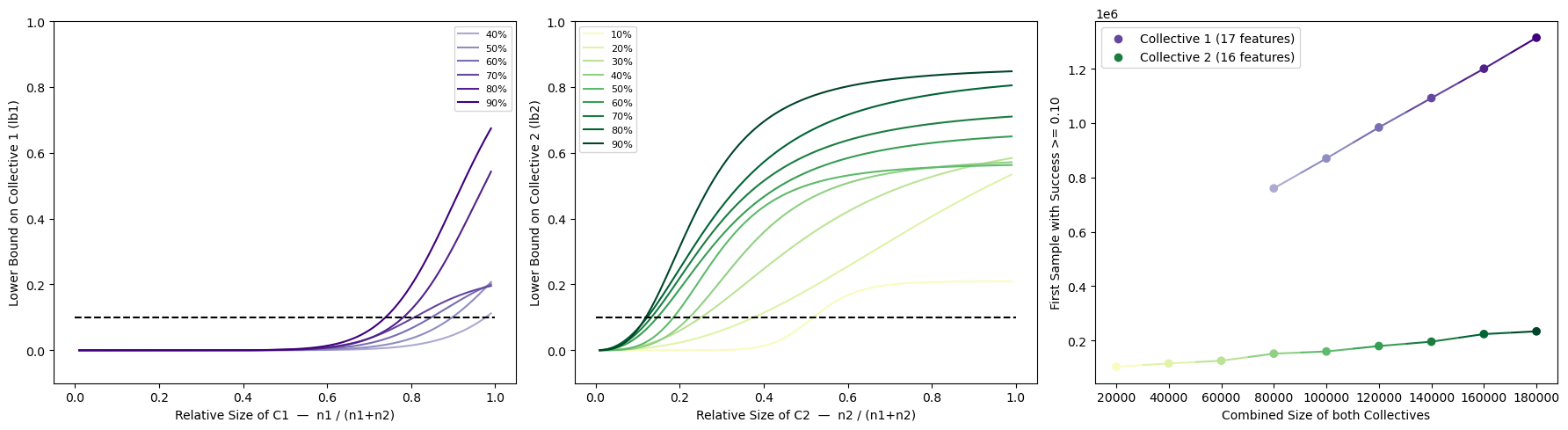}
    \caption{The bound on feature-only planting success for two collectives, 17 vs 16 features, under varying population sizes and collective sizes.}
    \label{fig:feature_analysis_fo_per-collective_17v16}
\end{figure*}

\begin{figure*}[ht!]{}
    \centering
    \includegraphics[height=1.5in,width=\linewidth]{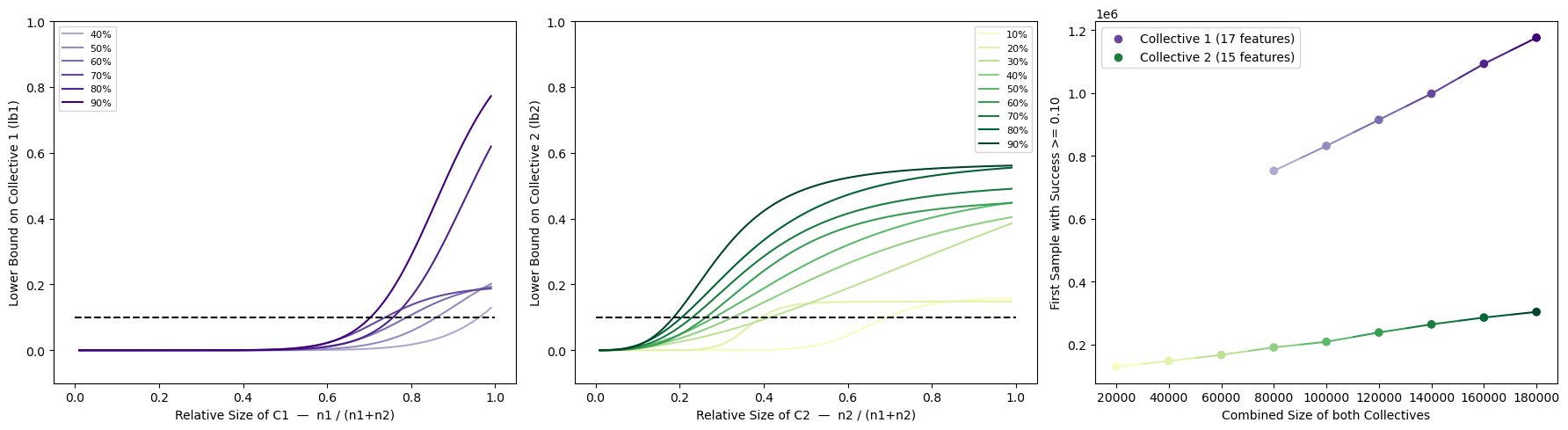}
    \caption{The bound on feature-only planting success for two collectives, 17 vs 15 features, under varying population sizes and collective sizes.}
    \label{fig:feature_analysis_fo_per-collective_17v15}
\end{figure*}

\begin{figure*}[ht!]{}
    \centering
    \includegraphics[height=1.5in,width=\linewidth]{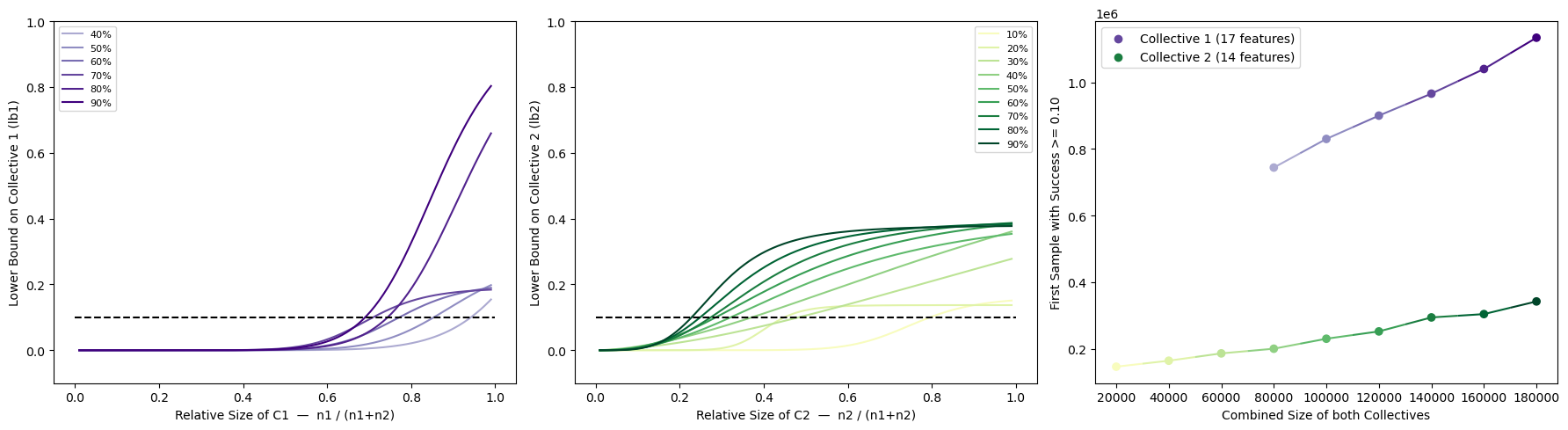}
    \caption{The bound on feature-only planting success for two collectives, 17 vs 14 features, under varying population sizes and collective sizes.}
    \label{fig:feature_analysis_fo_per-collective_17v14}
\end{figure*}

\begin{figure*}[ht!]{}
    \centering
    \includegraphics[height=1.5in,width=\linewidth]{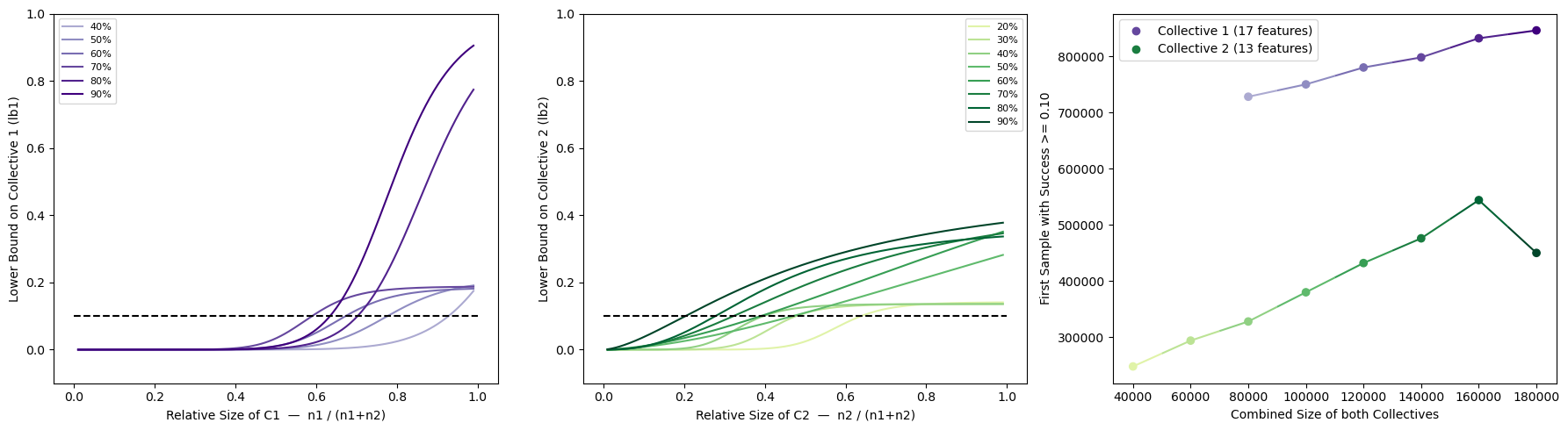}
    \caption{The bound on feature-only planting success for two collectives, 17 vs 13 features, under varying population sizes and collective sizes.}
    \label{fig:feature_analysis_fo_per-collective_17v13}
\end{figure*}

\newpage

\subsection{Single-Collective Ablations}
\begin{crblock}
We reproduce the single-collective lower bound numbers of \citep{gauthier2025statistical} by running our multi-collective implementation with a single collective ($M=1$). We use the same 17-feature signal $g$, $N^{\mathrm{test}}=100{,}000$, $\delta=0.05$ and $\varepsilon=0$ as the original work, and recreate the following experiments:
\begin{itemize}
    \item Feature-label planting of the label \textit{Poor} with total dataset size $N\in\{500\text{K}, 1\text{M}, 2\text{M}\}$. The original paper sweeps the collective size from 1\% to 30\% of $N$; we evaluate the bound on the active window 6\% to 16\% in steps of 0.5\%, outside of which the bound is flat at $0$ or at its saturation value. We report the results in Figure~\ref{fig:feature_label_ablation};
    \item Feature-only planting with a fixed dataset size of $N=1$M, the transformation $x_0$ of the original work, and the label $y^*$ varying between \textit{Poor}, \textit{Average} and \textit{Excellent}. The collective size spans the active window of each label: 5\% to 20\% for \textit{Excellent}, 55\% to 99.5\% for \textit{Poor} and 60\% to 99.5\% for \textit{Average}, in steps of 0.5\%. We report the results in Figure~\ref{fig:feature_only_ablation};
    \item Signal unplanting of the label \textit{Good} for a dataset of size $N=1$M with the adaptive strategy and an estimation split of $n_e=2000$, for collective sizes from 0.5\% to 15\% in steps of 0.5\%. We report the results in Figure~\ref{fig:unplanting_ablation}.
\end{itemize}

For the two planting experiments, the $M=1$ bounds coincide with the single-collective numbers of \citep{gauthier2025statistical}: on every collective size evaluated by both, the bound values agree to within $0.002$, with the exception of a handful of points located exactly at a jump of the step-like bound, where the transition occurs one grid step (0.5\% of $N$) earlier or later. The only systematic difference is that our implementation clips the bound at $0$, where the original reports values of about $-0.01$. This validates the method proposed here as a strict generalization of the single-collective approach.
\end{crblock}

% \subsection{Single-Collective Ablations}

% We reproduce the single-collective lower bound numbers of \citep{gauthier2025statistical}. We follow the original work settings closely in order to recreate the following experiments:
% \begin{itemize}
%     \item Feature-label planting with total dataset size 500K, 1 million and 2 millions. We explore a collective size varying from 1\% to 30\% like in the original paper. We report the results in figure \ref{fig:feature_label_ablation};
%     \item Feature-only planting, with a fixed dataset size of 1 million and the label $y^*$ varying between \textit{Poor}, \textit{Average} and \textit{Excellent}, with collective size varying between 1\% and 80\%. We report the results in figure \ref{fig:feature_only_ablation};
%     \item Signal unplanting for a dataset of size 2 millions, $n_e=2000$, and the ``naive'' strategy to unplant the label \textit{Good}. We report the results in figure \ref{fig:unplanting_ablation}.
% \end{itemize}

% We observe complete fidelity to the single-collective method numbers across the board, validating the method proposed here as a strict generalization to the multiple collective case of the studied approach.

\begin{figure*}[t!]{}
    \centering
    \includegraphics[height=1.5in]{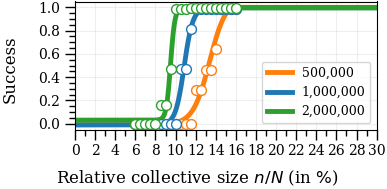}
    \caption{Re-creating lower bound numbers for the feature-label case from the Single-collective reference paper}
    \label{fig:feature_label_ablation}
\end{figure*}

\begin{figure*}[t!]{}
    \centering
    \includegraphics[height=1.5in]{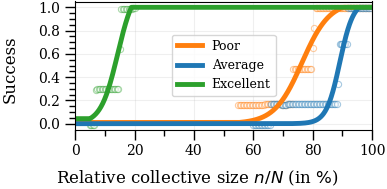}
    \caption{Re-creating lower bound numbers for the feature-only case from the Single-collective reference paper}
    \label{fig:feature_only_ablation}
\end{figure*}

\begin{figure*}[t!]{}
    \centering
    \includegraphics[height=1.5in]{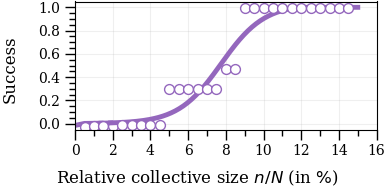}
    \caption{Re-creating lower bound numbers for the \cameraready{signal-unplanting} case from the Single-collective reference paper}
    \label{fig:unplanting_ablation}
\end{figure*}

\subsection{\cameraready{Signal Unplanting with Two Collectives}}\label{app:unplanting}
\begin{crblock}
We complement the planting experiments of Section~\ref{sec:experiments} with two collectives that unplant. Collective~1 maps 17 features through $g_1$ and removes their association with \textit{Greening \& Shade Expansion}; Collective~2 maps a nested subset of 15 of those features through $g_2$ and targets \textit{Green Stormwater Infrastructure}. We set $N=200{,}000$, $N^{\mathrm{test}}=50{,}000$, $\delta=0.05$ and $\varepsilon=0$, and use per-collective conflict terms. The total collective mass is $\bar\alpha_N\in\{0.90,0.75,0.60\}$, of which Collective~1 holds a share $s_1$ and Collective~2 the rest. We sweep $s_1$ over $50$ equally spaced values in $[0.01,0.99]$ and report five representative shares. Table~\ref{tab:unplanting} reports the lower bound of Theorem~\ref{thm:multi_unplant} for the adaptive strategy, with $n_{e,c}=2000$ (Appendix~\ref{app:ene-ablation}), and for the naive strategy of Section~\ref{sec:signal_unplanting}; we clip small negative bound values to $0$.

Neither strategy dominates. For Collective~1, the adaptive bound is at least as large as the naive one in 11 of the 15 cells, and in the remaining four the naive bound is larger by about $0.17$. For Collective~2, whose map is nested in Collective~1's so that all of its planted values are contested, both bounds are low: the adaptive bound falls below $0.05$ in 9 of the 15 cells, and the naive bound is larger in 8 of them. For Collective~1, the adaptive bound is never vacuous and, at the shares reported in the table, increases with $s_1$ at every mass level, whereas the naive bound is vacuous in 6 of the 30 cells and is not monotone (Collective~1 at $\bar\alpha_N=0.90$: $0.978$, $0$ and $0.983$ for $s_1=0.11$, $0.21$ and $0.31$). Once Collective~1 holds half of the collective mass, its adaptive bound exceeds $0.98$ at every mass level.%
\end{crblock}

\begin{table}[h]
\crcolor\centering\small
\caption{\cameraready{Signal unplanting with two collectives: lower bounds of Theorem~\ref{thm:multi_unplant} for the adaptive strategy ($n_{e,c}=2000$) and the naive strategy, by Collective~1's share $s_1$ of the total collective mass $\bar\alpha_N$. A bound of $0$ is vacuous.}}
\label{tab:unplanting}
\begin{tabular}{lcccccc}
\toprule
& \multicolumn{3}{c}{Adaptive, $\bar\alpha_N=$} & \multicolumn{3}{c}{Naive, $\bar\alpha_N=$}\\
\cmidrule(lr){2-4}\cmidrule(lr){5-7}
$s_1$ & $0.90$ & $0.75$ & $0.60$ & $0.90$ & $0.75$ & $0.60$\\
\midrule
\multicolumn{7}{l}{\emph{Collective 1 (17 features)}}\\
0.11 & 0.812 & 0.642 & 0.290 & 0.978 & 0.000 & 0.000\\
0.21 & 0.815 & 0.815 & 0.814 & 0.000 & 0.981 & 0.288\\
0.31 & 0.987 & 0.816 & 0.815 & 0.983 & 0.983 & 0.981\\
0.51 & 0.988 & 0.988 & 0.987 & 0.985 & 0.984 & 0.984\\
0.71 & 0.989 & 0.989 & 0.988 & 0.986 & 0.986 & 0.985\\
\midrule
\multicolumn{7}{l}{\emph{Collective 2 (15 features)}}\\
0.11 & 0.422 & 0.423 & 0.144 & 0.418 & 0.364 & 0.137\\
0.21 & 0.200 & 0.145 & 0.144 & 0.193 & 0.364 & 0.362\\
0.31 & 0.042 & 0.011 & 0.012 & 0.364 & 0.362 & 0.138\\
0.51 & 0.015 & 0.012 & 0.013 & 0.363 & 0.000 & 0.137\\
0.71 & 0.013 & 0.014 & 0.041 & 0.000 & 0.134 & 0.000\\
\bottomrule
\end{tabular}
\end{table}

\subsection{\cameraready{Size of the Estimation Split}}\label{app:ene-ablation}
\begin{crblock}
The adaptive strategy selects the alternative labels $y'_c(x^*)$ on the estimation split $E_{N,c}$ and evaluates the bound on the holdout $R_{N,c}$. We repeat the experiment of Appendix~\ref{app:unplanting} at $\bar\alpha_N=0.90$ with $n_{e,c}\in\{100,2000,10000\}$ (Table~\ref{tab:ene}). The bounds for $n_{e,c}=2000$ and $n_{e,c}=10000$ agree to within $0.02$ in every cell, so the label selection has converged at $n_{e,c}=2000$. With $n_{e,c}=100$ the selected labels are noisy, and the bound, which is valid for the strategy it certifies, differs by more than $0.1$ in four of the ten cells (e.g., Collective~2: $0.256$, $0.226$ and $0.142$, against $0.042$, $0.015$ and $0.013$ at $n_{e,c}=2000$). Since a larger $n_{e,c}$ removes data from the holdout and inflates the radii, $n_{e,c}=2000$ is the smallest converged choice; it is also the value used by \citet{gauthier2025statistical} for a single collective.%
\end{crblock}

\begin{table}[h]
\crcolor\centering\small
\caption{\cameraready{Size $n_{e,c}$ of the estimation split: adaptive unplanting bounds at total collective mass $\bar\alpha_N=0.90$.}}
\label{tab:ene}
\begin{tabular}{lcccccc}
\toprule
& \multicolumn{3}{c}{Collective 1, $n_{e,c}=$} & \multicolumn{3}{c}{Collective 2, $n_{e,c}=$}\\
\cmidrule(lr){2-4}\cmidrule(lr){5-7}
$s_1$ & $100$ & $2000$ & $10000$ & $100$ & $2000$ & $10000$\\
\midrule
0.11 & 0.811 & 0.812 & 0.812 & 0.422 & 0.422 & 0.422\\
0.21 & 0.986 & 0.815 & 0.815 & 0.200 & 0.200 & 0.200\\
0.31 & 0.987 & 0.987 & 0.986 & 0.256 & 0.042 & 0.042\\
0.51 & 0.988 & 0.988 & 0.988 & 0.226 & 0.015 & 0.012\\
0.71 & 0.989 & 0.989 & 0.989 & 0.142 & 0.013 & 0.014\\
\bottomrule
\end{tabular}
\end{table}

\section{\cameraready{Collectives that Are Not Neutral}}\label{app:tv}
\begin{crblock}
Neutrality requires each collective to be a random slice of the population before editing, which fails for collectives that form around shared characteristics. Here we relax it to closeness in total variation. Let $P_{0,c}:=\PP[f_0(\omega)\in\cdot\mid\omega\in\Omega_c]$ be the pre-edit law of collective $c$ and $P_{0,\mathrm{nc}}:=\PP[f_0(\omega)\in\cdot\mid\omega\notin\cup_{d=1}^M\Omega_d]$ that of the non-collective users, and assume
\begin{equation}\label{eq:tv_assumption}
\TV(P_0,P_{0,c})\;\le\;\phi_{N,c},\qquad c=1,\dots,M.
\end{equation}
The bias of the non-collective population is not a new primitive: the mixture identity $P_0=(1-\bar\alpha)P_{0,\mathrm{nc}}+\sum_c\alpha_cP_{0,c}$ gives $(1-\bar\alpha)(P_0-P_{0,\mathrm{nc}})=\sum_c\alpha_c(P_{0,c}-P_0)$, hence $\phi_{N,0}:=\TV(P_0,P_{0,\mathrm{nc}})\le\sum_c\alpha_c\phi_{N,c}/(1-\bar\alpha)$. Neutrality is the case $\phi_{N,c}=\phi_{N,0}=0$.
\begin{proposition}\label{prop:tv}
Under the assumptions of Theorem~\ref{thm:multi_feature_label}, with neutrality replaced by \eqref{eq:tv_assumption}, with probability at least $1-\delta$,
\begin{multline*}
\hat S_c(\alpha_{N,c})\;\ge\;\PP_{x\sim\hat P^{\mathrm{pre},X}_{N,c}}\Bigl[\,m_c\bigl(g_c(x)\bigr)-2\alpha_{N,0}\bigl(\phi_{N,c}+\phi_{N,0}\bigr)>0\Bigr]\\
-R_{\delta/4}(n_{N,c})-R_{\delta/4}(N^{\mathrm{test}})-\phi_{N,c},
\end{multline*}
where $m_c(x^*)$ denotes the margin inside the probability in \eqref{eq:multi_FL_bound_final_clean_preonly}. The analogous statement holds for Theorem~\ref{thm:multi_unplant}, with the holdout sample $R_{N,c}$ of size $n_{N,c}-n_{e,c}$ in place of the sample of collective $c$.
\end{proposition}
\begin{proof}
Neutrality enters the proof of Theorem~\ref{thm:multi_feature_label} in two places. In Step~3 it centers the non-collective sample and the pre-edit sample of collective $c$ at $P_0$, which yields \eqref{eq:hoeff_noncoll_upper}--\eqref{eq:hoeff_collpre_lower} and, from them, the conflict comparisons \eqref{eq:delta_exact_vs_pop} and \eqref{eq:delta_pop_vs_proxy}. In Step~5 it makes the pre-edit features of collective $c$ i.i.d.\ with law $P_0^X$, which yields \eqref{eq:P0_vs_collective_pre_pop}. The other steps do not use it: the margin lemma and the mixture decomposition are distribution-free, the cross-collective conflicts $\Delta_{c\to d}$ are empirical quantities of the other collectives' post-edit samples, and the concentration of the signal term holds verbatim once $p^X_{c,\mathrm{pop}}$ is taken to be the feature law of the edited points of collective $c$, $((h_c)_{\#}P_{0,c})^X$, which does not appear in the final bound.

Without neutrality, Hoeffding's inequality centers the non-collective sample at $P_{0,\mathrm{nc}}$ and the pre-edit sample of collective $c$ at $P_{0,c}$. Since $|Q(B)-Q'(B)|\le\TV(Q,Q')$ for every event $B$, \eqref{eq:hoeff_noncoll_upper}--\eqref{eq:hoeff_noncoll_lower} hold with $R_{\tilde\delta}(n_{N,0})+\phi_{N,0}$, and \eqref{eq:hoeff_collpre_upper}--\eqref{eq:hoeff_collpre_lower} with $R_{\tilde\delta}(n_{N,c})+\phi_{N,c}$, in place of the radii. A conflict term is a maximum of differences of two such probabilities, so \eqref{eq:delta_exact_vs_pop} and \eqref{eq:delta_pop_vs_proxy} gain the additive terms $2\phi_{N,0}$ and $2\phi_{N,c}$, respectively. Both comparisons enter the margin multiplied by $\alpha_{N,0}$, which subtracts $2\alpha_{N,0}(\phi_{N,c}+\phi_{N,0})$ from it. In Step~5, Hoeffding's inequality applies to the law $P^X_{0,c}$ of the pre-edit features of collective $c$, and $\TV(P_0^X,P^X_{0,c})\le\TV(P_0,P_{0,c})\le\phi_{N,c}$ gives $P_0^X(A_c^{\mathrm{pop}})\ge P^X_{0,c}(A_c^{\mathrm{pop}})-\phi_{N,c}$, which adds $-\phi_{N,c}$ outside the probability. The concentration events, and hence the failure probability, are unchanged. In the proof of Theorem~\ref{thm:multi_unplant}, neutrality enters in the same two places, through the holdout sample.
\end{proof}
Three consequences follow. First, the degradation is linear in $(\phi_{N,c},\phi_{N,0})$, with no threshold effect, and Theorems~\ref{thm:multi_feature_label} and~\ref{thm:multi_unplant} are recovered at $\phi=0$. Second, a non-representative collective pays twice: its data are a biased mirror of the non-collective population, which costs $2\alpha_{N,0}\phi_{N,c}$ inside the margin, and its features are a biased proxy for where the test mass lies, which costs $\phi_{N,c}$ outside it. Third, the penalty inside the margin comes from estimating the non-collective conflict through the population law, so it does not depend on the information regime for the cross-collective term: even in the best case of Section~\ref{sec:signal_planting}, where an aggregator publishes the cross-collective term, the non-collective conflict is still estimated from the collective's own sample.%
\end{crblock}

\section{\cameraready{Validity for Trained Classifiers}}\label{app:trained}
\begin{crblock}
Theorem~\ref{thm:multi_feature_label} holds for any classifier that is $\varepsilon$-contamination-suboptimal under the post-edit distribution. Here we test it on classifiers that a platform would actually train on collectively edited data; Table~\ref{tab:trained} in Section~\ref{sec:experiments} summarizes the results.

\textbf{Protocol.} Two collectives perform feature-label planting with total mass $\bar\alpha_N=0.3$, split as $\alpha_{N,1}=x\,\bar\alpha_N$ and $\alpha_{N,2}=(1-x)\,\bar\alpha_N$ for $x\in\{0.2,0.35,0.5,0.65,0.8\}$; each training point joins a collective independently, as in the sampling model of Section~\ref{sec:MACA}. In the \emph{aligned} scenario both collectives plant the same feature values with the same target label; under \emph{full conflict} they plant the same feature values with different labels; under \emph{partial conflict} Collective~2 plants a subset of Collective~1's feature values with a different label. The platform trains multinomial logistic regression (LR) on one-hot features (L2 penalty with $C=1$, up to 500 L-BFGS iterations) or histogram gradient-boosted trees (GBT) with native categorical features and default hyperparameters, both from scikit-learn~\cite{pedregosa2011scikit}. The empirical success of collective $c$ is the fraction of $N^{\mathrm{test}}=100{,}000$ fresh test points $x$ with $\hat m(g_c(x))=y^*_c$, and its bound is Theorem~\ref{thm:multi_feature_label} in the per-collective regime with $\delta=0.05$ and $\varepsilon=0$. Setting $\varepsilon=0$ gives the most demanding test: the bound is largest at $\varepsilon=0$ and decreases with $\varepsilon$, so validity at $\varepsilon=0$ implies validity for every $\varepsilon>0$. All draws use a single fixed seed.

\textbf{Data.} \emph{Synthetic:} a generator written for these experiments, with the 18 features and 4 labels of Appendix~\ref{app:simulated_adaptation_records}. It differs from the generator behind Figures~\ref{fig:feature_analysis} and~\ref{fig:heatmap}: features are drawn independently from fixed non-uniform marginals, and labels from an additive per-feature score plus Gumbel noise. We use $N=2{,}000{,}000$. Collective~1 plants 17 of the 18 features (all but cooling provision); under partial conflict Collective~2 plants 15 of them. \emph{ACSIncome}~\cite{ding2021retiring}: US Census American Community Survey microdata (2018, one-year) for the 12 most populous states, with the binary label ``income above \$50K'' (positive rate $0.381$) and 10 features, discretized to keep $\mathcal X$ finite: age, class of worker, education, marital status, occupation group, place of birth, relationship, weekly hours, sex and race. We use $N=877{,}909$ training and $100{,}000$ test points. Collective~1 plants 9 of the 10 features (all but sex) with target ``above \$50K''; Collective~2 plants the same 9 features (8 under partial conflict, also leaving marital status free), with the same target when aligned and ``at most \$50K'' otherwise.

\textbf{Results.} Table~\ref{tab:trained_full} reports every configuration and Figure~\ref{fig:trained} plots them. GBT never violates the bound, on either dataset, and the bound is close to tight where the signal is well identified: under synthetic partial conflict, Collective~2 is certified $0.780$ against an empirical success of $0.786$. LR satisfies the bound on ACSIncome but violates it twice on the synthetic data, for Collective~2 under partial conflict at $x=0.5$ and $x=0.65$. There, Collective~1's signal is a conjunction of feature values that Collective~2's nested pattern contests with a different label, which a model linear in one-hot features cannot represent; LR is then far from Bayes-optimal and the premise of Theorem~\ref{thm:multi_feature_label} fails. Removing the regularization ($C=\infty$) leaves LR's empirical success unchanged in all 30 configurations, so the violations are not a regularization artifact. The phenomena predicted by the theory appear on both datasets. Aligned collectives pool their masses and are both certified at about $0.99$ for every split (Corollary~\ref{cor:symmetric}(i)). Under full conflict the larger collective is certified and the smaller one is not, and at the even split neither is (Corollary~\ref{cor:symmetric}(iii)). Under partial conflict Collective~1 wins the contested values once its mass is large enough ($x\ge0.35$ on the synthetic data, $x\ge0.5$ on ACSIncome), while Collective~2 keeps its uncontested planted values. On ACSIncome, Collective~2's bound under partial conflict ($0.31$ against $0.46$ at $x=0.5$ and $x=0.65$) is limited by finite-sample concentration over its 10 planted values, some of which carry too little mass to be certified at this sample size. A run takes about 9 minutes (synthetic) and 2 minutes (ACSIncome) on an Apple M3 Pro CPU; the code is available at \url{https://github.com/clabat9/multi-collective-aca}.%
\end{crblock}

\begin{table}[h]
\crcolor\centering\small
\caption{\cameraready{Trained classifiers: Theorem~\ref{thm:multi_feature_label} lower bound and empirical success of gradient-boosted trees (GBT) and logistic regression (LR), for each scenario, collective $c$ and mass split $x$. A bound of $0$ is vacuous. Violations of the bound are in bold.}}
\label{tab:trained_full}
\begin{tabular}{llc ccc ccc}
\toprule
& & & \multicolumn{3}{c}{Synthetic, $N=2{,}000{,}000$} & \multicolumn{3}{c}{ACSIncome, $N=877{,}909$}\\
\cmidrule(lr){4-6}\cmidrule(lr){7-9}
Scenario & $c$ & $x$ & Bound & GBT & LR & Bound & GBT & LR\\
\midrule
Aligned & 1 & 0.20 & 0.991 & 1.000 & 1.000 & 0.989 & 1.000 & 1.000 \\
 & 1 & 0.35 & 0.992 & 1.000 & 1.000 & 0.990 & 1.000 & 1.000 \\
 & 1 & 0.50 & 0.993 & 1.000 & 1.000 & 0.991 & 1.000 & 1.000 \\
 & 1 & 0.65 & 0.993 & 1.000 & 1.000 & 0.992 & 1.000 & 1.000 \\
 & 1 & 0.80 & 0.993 & 1.000 & 1.000 & 0.992 & 1.000 & 1.000 \\
\addlinespace[2pt]
 & 2 & 0.20 & 0.993 & 1.000 & 1.000 & 0.992 & 1.000 & 1.000 \\
 & 2 & 0.35 & 0.993 & 1.000 & 1.000 & 0.992 & 1.000 & 1.000 \\
 & 2 & 0.50 & 0.993 & 1.000 & 1.000 & 0.991 & 1.000 & 1.000 \\
 & 2 & 0.65 & 0.992 & 1.000 & 1.000 & 0.990 & 1.000 & 1.000 \\
 & 2 & 0.80 & 0.991 & 1.000 & 1.000 & 0.989 & 1.000 & 1.000 \\
\addlinespace[2pt]
\midrule
Full conflict & 1 & 0.20 & 0.000 & 0.000 & 0.000 & 0.000 & 0.000 & 0.000 \\
 & 1 & 0.35 & 0.000 & 0.000 & 0.000 & 0.000 & 0.000 & 0.000 \\
 & 1 & 0.50 & 0.000 & 1.000 & 0.517 & 0.000 & 0.520 & 0.520 \\
 & 1 & 0.65 & 0.993 & 1.000 & 1.000 & 0.992 & 1.000 & 1.000 \\
 & 1 & 0.80 & 0.993 & 1.000 & 1.000 & 0.992 & 1.000 & 1.000 \\
\addlinespace[2pt]
 & 2 & 0.20 & 0.993 & 1.000 & 1.000 & 0.992 & 1.000 & 1.000 \\
 & 2 & 0.35 & 0.993 & 1.000 & 1.000 & 0.992 & 1.000 & 1.000 \\
 & 2 & 0.50 & 0.000 & 0.000 & 0.483 & 0.000 & 0.480 & 0.480 \\
 & 2 & 0.65 & 0.000 & 0.000 & 0.000 & 0.000 & 0.000 & 0.000 \\
 & 2 & 0.80 & 0.000 & 0.000 & 0.000 & 0.000 & 0.000 & 0.000 \\
\addlinespace[2pt]
\midrule
Partial conflict & 1 & 0.20 & 0.000 & 1.000 & 0.000 & 0.000 & 0.000 & 0.000 \\
 & 1 & 0.35 & 0.992 & 1.000 & 1.000 & 0.000 & 0.480 & 0.000 \\
 & 1 & 0.50 & 0.993 & 1.000 & 1.000 & 0.991 & 1.000 & 1.000 \\
 & 1 & 0.65 & 0.993 & 1.000 & 1.000 & 0.992 & 1.000 & 1.000 \\
 & 1 & 0.80 & 0.993 & 1.000 & 1.000 & 0.992 & 1.000 & 1.000 \\
\addlinespace[2pt]
 & 2 & 0.20 & 0.780 & 0.786 & 1.000 & 0.911 & 1.000 & 1.000 \\
 & 2 & 0.35 & 0.780 & 0.786 & 0.786 & 0.309 & 0.756 & 1.000 \\
 & 2 & 0.50 & 0.779 & 0.786 & \textbf{0.694} & 0.310 & 0.460 & 0.460 \\
 & 2 & 0.65 & 0.692 & 0.786 & \textbf{0.607} & 0.309 & 0.460 & 0.460 \\
 & 2 & 0.80 & 0.000 & 0.786 & 0.438 & 0.000 & 0.460 & 0.327 \\
\bottomrule
\end{tabular}
\end{table}

\begin{figure}[h]
\crcolor\centering
\includegraphics[width=0.9\linewidth]{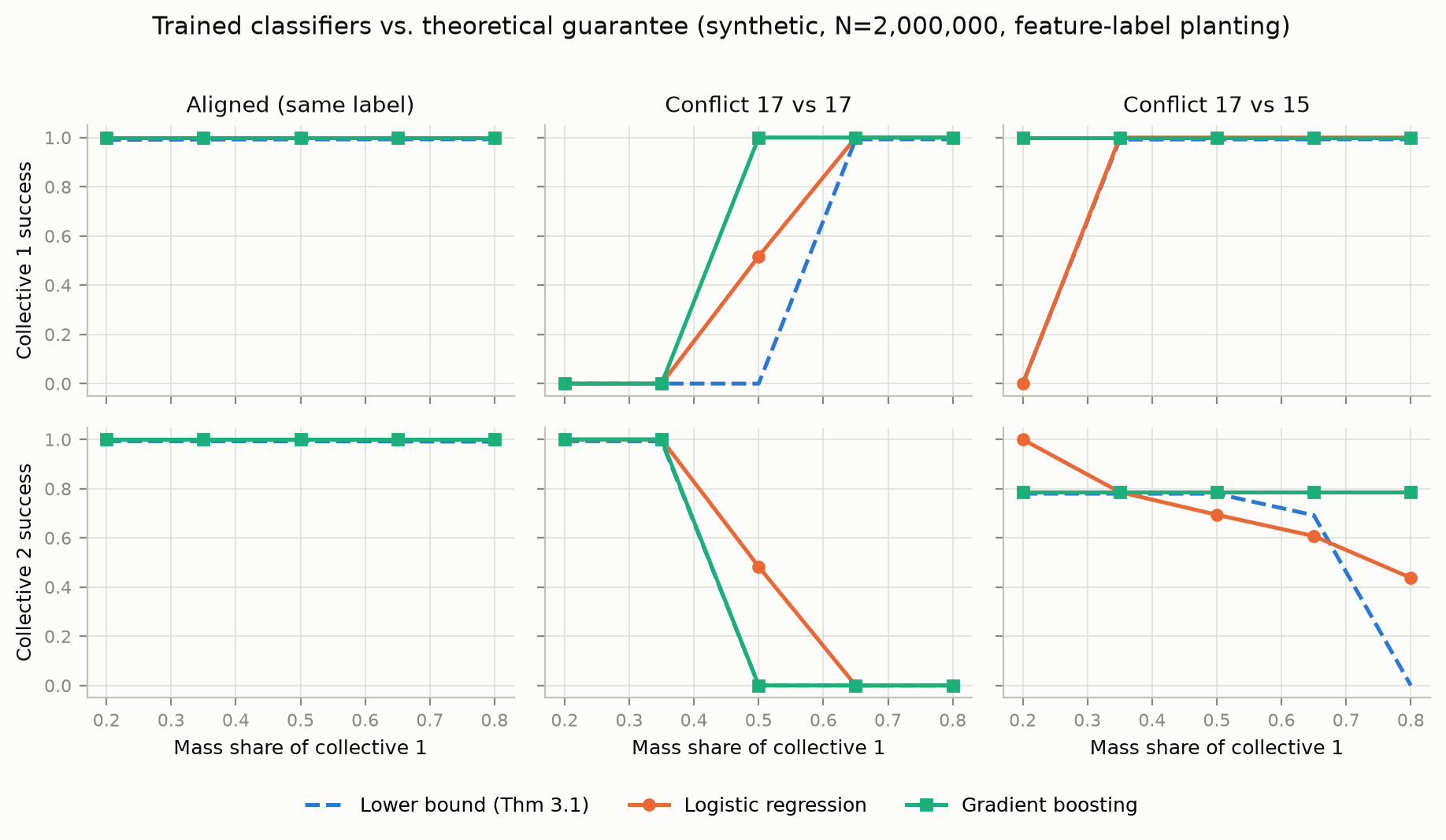}

\vspace{0.3cm}
\includegraphics[width=0.9\linewidth]{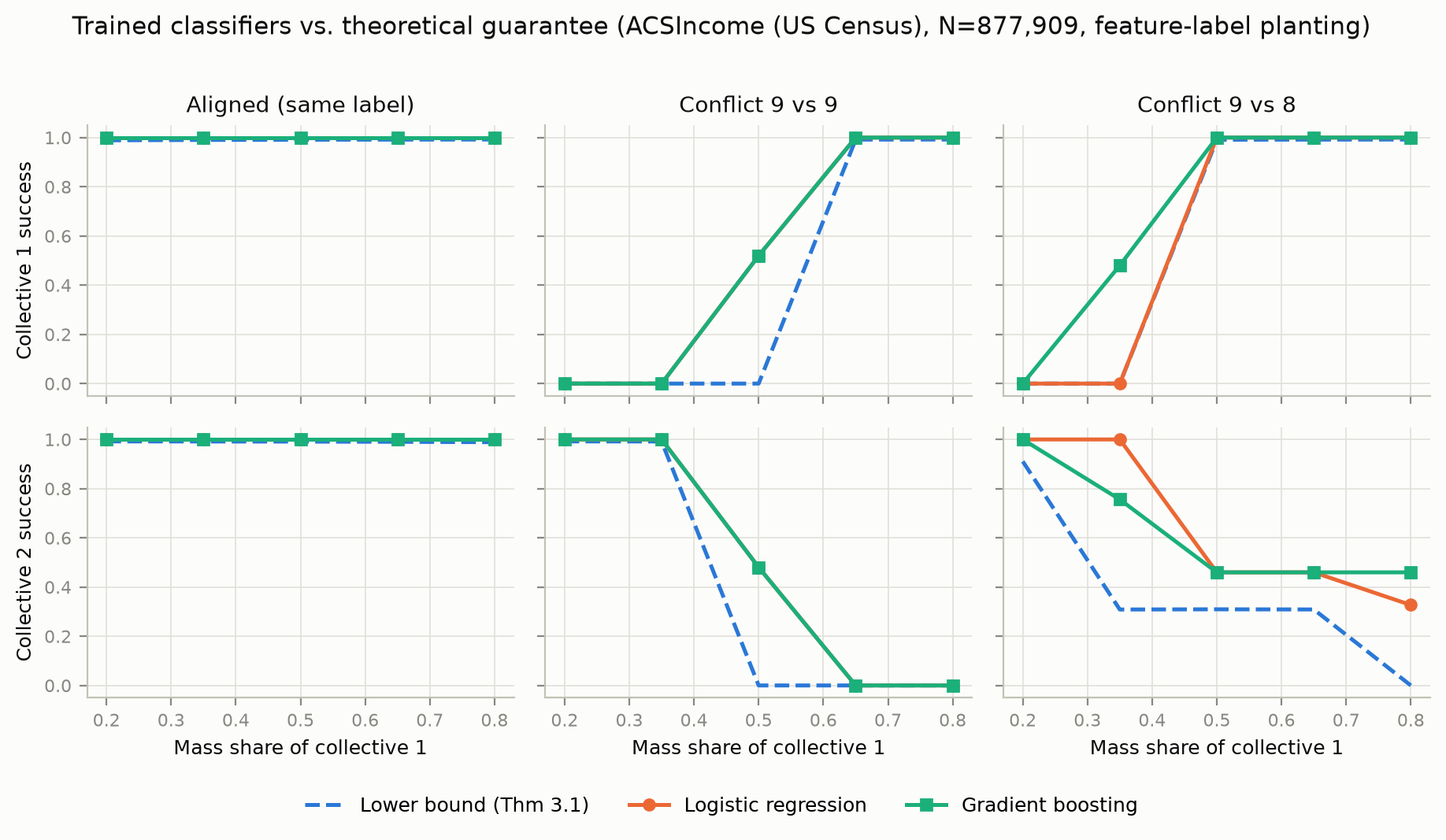}
\caption{\cameraready{Trained classifiers: lower bound of Theorem~\ref{thm:multi_feature_label} (dashed) and empirical success of logistic regression and gradient-boosted trees, by Collective~1's share of the total collective mass. Top: synthetic data. Bottom: ACSIncome.}}
\label{fig:trained}
\end{figure}

%\clearpage
%\input{checklist.tex}
\end{document}